\documentclass[a4paper,thm-restate,USenglish]{lipics-v2019}

\hideLIPIcs
\nolinenumbers

\usepackage[utf8]{inputenc}
\usepackage[T1]{fontenc}
\usepackage{needspace}
\usepackage{comment} 

\usepackage{libertine} 
\usepackage{inconsolata} 

\usepackage{amsmath, amsthm, amssymb}
\usepackage{stmaryrd}
\usepackage{graphicx}
\usepackage{tikz}
\usepackage{enumerate}
\usepackage{subcaption}

\usepackage{thmtools}
\usepackage{float}
\usepackage{xcolor}
\usepackage{doi} 
\usepackage[numbers, sort&compress]{natbib} 

\declaretheorem[name=Question, sibling=theorem]{question}

\newcommand{\Hm}{H^{-}}

\title{Local certification of graph decompositions and applications to minor-free classes}

\funding{This work was supported by ANR project GrR (ANR-18-CE40-0032).}

\author{Nicolas Bousquet}{Univ. Lyon, Université Lyon 1, LIRIS UMR CNRS 5205, F-69621, Lyon, France}{nicolas.bousquet@univ-lyon1.fr}{https://orcid.org/0000-0003-0170-0503}{}

\author{Laurent Feuilloley}{Univ. Lyon, Université Lyon 1, LIRIS UMR CNRS 5205, F-69621, Lyon, France}{laurent.feuilloley@univ-lyon1.fr}{https://orcid.org/0000-0002-3994-0898}{}

\author{Théo Pierron}{Univ. Lyon, Université Lyon 1, LIRIS UMR CNRS 5205, F-69621, Lyon, France}{theo.pierron@univ-lyon1.fr}{https://orcid.org/0000-0002-5586-5613}{}


\authorrunning{N. Bousquet, L. Feuilloley, T. Pierron}

\begin{document}

\maketitle

\begin{abstract}
Local certification consists in assigning labels to the nodes of a network to certify that some given property is satisfied, in such a way that the labels can be checked locally.
In the last few years, certification of graph classes received a considerable attention. The goal is to certify that a graph $G$ belongs to a given graph class~$\mathcal{G}$. Such certifications with labels of size $O(\log n)$ (where $n$ is the size of the network) exist for trees, planar graphs and graphs embedded on surfaces. Feuilloley et al. ask if this can be extended to any class of graphs defined by a finite set of forbidden minors.

In this work, we develop new decomposition tools for graph certification, and apply them to show that for every small enough minor $H$, $H$-minor-free graphs can indeed be certified with labels of size $O(\log n)$. We also show matching lower bounds using a new proof technique. 
\end{abstract}

\section{Introduction}

Local certification is an active field of research in the theory of distributed computing. On a high level it consists in certifying global properties in such a way that the verification can be done locally. 
More precisely, for a given property, a local certification consists of a labeling (called a \emph{certificate assignment}), and of a local verification algorithm. 
If the configuration of the network is correct, then there should exist a labeling of the nodes that is accepted by the verification algorithm, whereas if the configuration is incorrect no labeling should make the verification algorithm accept.

Local certification originates from self-stabilization, and was first concerned with certifying that a solution to an algorithmic problem is correct.
However, it is also important to understand how to certify properties of the network itself,
that is, to find locally checkable proofs that the network belongs to some graph class. There are several reasons for that.
First, because certifying some solutions can be hard in general graphs, while they become simpler on more restricted classes. To make use of this fact, it is important to be able to certify that the network does belong to the restricted class.
Second, because some distributed algorithms work only on some specific graph classes, and we need a way to ensure that the network does belong to the class, before running the algorithm. 
Third, the distinction between certifying solutions and network properties is rather weak, in the sense that the techniques are basically the same. So we should take advantage of the fact that a lot is known about graph classes to learn more about certification. 

In the domain of graph classes certification, 
there have been several results on various classes such as trees~\cite{KormanKP10}, bipartite graphs \cite{GoosS16} or graphs of bounded diameter \cite{Censor-HillelPP20}, but until two years ago little was known about essential classes, such as planar graphs. 
Recently, it has been shown that planar graphs and graphs of bounded genus can be certified with $O(\log n)$-bit labels~\cite{FeuilloleyFMRRT20, FeuilloleyFMRRT21, EsperetL21}. 
This size, $O(\log n)$, is the gold standard of certification, in the sense that little can be achieved with $o(\log n)$ bits, thus $O(\log n)$ is often the best we can hope for.

Planar and bounded-genus graphs are classic examples of graphs classes defined by forbidden minors, a type of characterization that has become essential in graph theory since the Graph minor series of Robertson and Seymour~\cite{RobertsonS85}. 
Remember that a graph $H$ is a minor of a graph $G$, is it possible to obtain $H$ from $G$ by deleting vertices, deleting edges, contracting edges. 
At this point, the natural research direction is to try to get the big picture of graph classes  certification, by understanding all classes defined by forbidden minors. In particular, we want to answer the following concrete question.

\begin{question}
[\cite{FeuilloleyFMRRT21, Feuilloley19}]
\label{conj:minorconj}
Can any graph class defined by a finite set of forbidden minors be certified with $O(\log n)$-bit certificates?
\end{question}

This open question is quite challenging: there are as many good reasons to believe that the answer is positive as negative.

First, the literature provides some reasons to believe that the conjecture is true. 
Properties that are known to be hard to certify, that is, that are known to require large certificates, are very different from minor-freeness. 
Specifically, all these properties (\emph{e.g.} small diameter \cite{Censor-HillelPP20}, non-3-colorability~\cite{GoosS16}, having a non-trivial automorphism \cite{GoosS16}) are non-hereditary. 
That is, removing a node or an edge may yield a graph that is not in the class. Intuitively, hereditary properties might be easier to certify in the sense that one does not need to encode information about every single edge or node, as the class is stable by removal of edges and nodes.  
Minor-freeness is a typical example of hereditary property. Moreover, this property, that has been intensively studied in the last decades, is known to carry a lot of structure, which is an argument in favor of the existence of a compact certification (that is a certification with $O(\log n)$-bit labels). 

On the other hand, from a graph theory perspective, it might be surprising that a general compact certification existed for minor-free graphs. 
Indeed, for the known results, obtaining a compact certification  is tightly linked to the existence of a precise constructive characterization of the class (\emph{e.g.} a planar embedding for planar graphs \cite{FeuilloleyFMRRT20, EsperetL21}, or a canonical path to the root for trees \cite{KormanKP10}). Intuitively, this is because forbidden minor characterizations are about structures that are absent from the graphs, and local certification is often about certifying the existence of some structures.
While such a characterization is known for some restricted minor-closed classes, we are far from having such a characterization for every minor-closed class. 
Note that there are a lot of combinatorial and algorithmic results on $H$-minor free graphs, but they actually follow from properties satisfied by $H$-minor free graphs, not from exact characterizations of such graphs. For certification, we need to rule out the graphs that do not belong to the class, hence a characterization is somehow necessary.

\subsection{Our results}

Answering Question~\ref{conj:minorconj} seems unfortunately out of reach, at the current state of our knowledge. We have explained above about why designing compact certification is hard for classes that do not have  a constructive characterization. We will later give some intuition about why lower bounds seem equally difficult to get. 
In this paper, we intend to build the foundations needed to tackle Question~\ref{conj:minorconj}. 
More precisely, we have four types of contributions. 

First, we show how to certify some graph decompositions. Such decompositions state how to build a class based on a few elementary graphs and a few simple operations. 
They are essential in structural graph theory, and more specifically in the study of minor-closed classes. 
Amongst the most famous examples of these theorems is the proof of the $4$-Color Theorem~\cite{appel1976every} or the Strong Perfect Graph Theorem~\cite{chudnovsky2006strong}.  

Second, we show that by directly applying these tools, we can design compact certification for several $H$-minor free classes, for which a precise characterization is known. See Fig.~\ref{fig:our-results-certification} and \ref{fig:minors}. That is, we answer positively Question~\ref{conj:minorconj}, for several small minors, and show that our decomposition tools can easily be used.

\begin{figure}[!h]
    \begin{center}
    \begin{tabular}{|c|c|c|c|}
    \hline
        Class 
        & Optimal size
        & Result \\
    \hline
        $K_3$-minor free 
        & $\Theta(\log n)$ 
        & 
        \begin{minipage}{3.0cm}
            \centering
            \vspace{0.1cm}
            Equivalent to\\ 
            acyclicity \cite{KormanKP10, GoosS16}.
            \vspace{0.1cm}
        \end{minipage} \\
    \hline
        Diamond-minor-free 
        & $\Theta(\log n)$ 
        & Corollary~\ref{coro:diamond-etc} \\
    \hline
        $K_4$-minor-free 
        & $\Theta(\log n)$ 
        & Corollary~\ref{coro:diamond-etc}\\
    \hline
        $K_{2,3}$-minor-free 
        & $\Theta(\log n)$ 
        & Corollary~\ref{coro:diamond-etc}\\
    \hline
        \begin{minipage}{3.5cm}
            \centering
            \vspace{0.1cm}
            $(K_{2,3},K_4)$-minor-free \\
            (\emph{i.e.} outerplanar)
            \vspace{0.1cm}
        \end{minipage}
        & $\Theta(\log n)$ 
        & Corollary~\ref{coro:diamond-etc}\\
    \hline
        $K_{2,4}$-minor-free 
        & $\Theta(\log n)$
        & Lemma~\ref{lem:2-connected-K24}\\
    \hline
    \end{tabular} 
    \caption{\label{fig:our-results-certification}
    Our main results for the certification of minor-closed classes.}
    \end{center}
\end{figure}

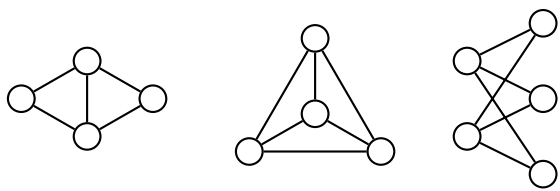
\begin{figure}[!ht]
\centering
\begin{tikzpicture}[thick,v/.style={draw, circle, minimum size=10pt}]
\node[v] (A) at (0,0.5) {};
\node[v] (B) at (0,-0.5) {};
\node[v] (C) at (-0.866,0) {};
\node[v] (D) at (0.866,0) {};
\draw (A) -- (C) -- (B) -- (D) -- (A) -- (B);
\tikzset{xshift=3cm,yshift=-0.2cm}
\node[v] (A) at (90:1) {};
\node[v] (B) at (210:1) {};
\node[v] (C) at (-30:1) {};
\node[v] (D) at (0,0) {};
\draw (D) -- (A) -- (B) -- (D) -- (C) -- (A);
\draw (B) -- (C);
\tikzset{xshift=2cm,yshift=0.2cm}
\node[v] (A) at (0,0.5) {};
\node[v] (B) at (0,-0.5) {};
\node[v] (C) at (1,-1) {};
\node[v] (D) at (1,0) {};
\node[v] (E) at (1,1) {};
\draw (A) -- (C) -- (B) -- (D) -- (A) -- (E) -- (B);
\end{tikzpicture}
\caption{
\label{fig:minors}
From left to right: the diamond, the clique on 4 vertices $K_4$, and the complete bipartite graph~$K_{2,3}$.}
\end{figure}

Third, we do a systematic study of small minors to identify which is the first one that we cannot tackle. 
We first prove the following theorem. 

\begin{restatable}{theorem}{ThmFourVertices}
\label{thm:4vertices}
$H$-minor-free classes can be certified in $O(\log n)$ bits when $H$ has at most 4 vertices.
\end{restatable}

Then, we extend this theorem to minors on five vertices with a specific shape, proving along the way new purely graph-theoretic characterizations for the associated classes. 
After this study, we can conclude that the next challenge is to understand $K_5$-minor free graphs.

Finally, we prove a general $\Omega(\log n)$ lower bounds for $H$-minor-freeness for all 2-connected graphs~$H$. 
This generalizes and simplifies the lower bounds  of \cite{FeuilloleyFMRRT20} which apply only to $K_k$ and $K_{p,q}$-minor-free graphs, and use ad-hoc and more complicated techniques.  

At the end of the paper, we discuss why the current tools we have, both in terms of upper and lower bounds, do not allow settling Question~\ref{conj:minorconj}. 
We list a few key questions that we need to answer before we can fully understand the certification of minor-closed classes, from the certification of classes with no tree minors to the certification $k$-connectivity, for arbitrary $k$. 

\subsection{Our techniques}
\label{subsec:our-techniques}

\paragraph*{General approach and challenges}
To give some intuition about our techniques, let us focus on a concrete example: $K_4$-minor-free graphs. Remember that a graph has $K_4$-minor if we can get a $K_4$ by deleting vertices and edges, and contracting edges. 
An alternative definition is that a graph has a $K_4$-minor, if it is possible to find four disjoint sets of vertices, called \emph{bags}, such that: each bag is connected, there is a  path between each pair of bags, these paths and bags are all vertex-disjoint (except for the endpoints of the paths that coincide with vertices of the bags). 
See Figure~\ref{fig:K4-minor}.

\begin{figure}[!h]
    \centering
    \begin{tabular}{ccc}
    \begin{minipage}{0.32 \textwidth}
        \scalebox{0.6}{
    \tikzset{every picture/.style={line width=0.75pt}} 

\begin{tikzpicture}[x=0.75pt,y=0.75pt,yscale=-1,xscale=1]

\draw [line width=1.5]    (237.55,53.55) -- (295.55,73.55) ;
\draw [line width=1.5]    (143.3,119.55) -- (177.55,75.55) ;
\draw [line width=1.5]    (177.55,75.55) -- (237.55,53.55) ;
\draw [line width=1.5]    (217.55,124.55) -- (177.55,75.55) ;
\draw [line width=1.5]    (217.55,124.55) -- (144.6,122.3) ;
\draw [line width=1.5]    (151.55,190.55) -- (144.6,122.3) ;
\draw [line width=1.5]    (151.55,190.55) -- (208.55,231.55) ;
\draw [line width=1.5]    (208.55,231.55) -- (281.55,231.55) ;
\draw [line width=1.5]    (281.55,231.55) -- (340.55,191.55) ;
\draw [line width=1.5]    (340.55,191.55) -- (345.02,118.02) ;
\draw [line width=1.5]    (345.02,118.02) -- (295.55,73.55) ;
\draw [line width=1.5]    (340.55,191.55) -- (272.55,174.55) ;
\draw [line width=1.5]    (272.55,174.55) -- (217.55,124.55) ;
\draw [line width=1.5]    (271.55,123.55) -- (255.74,137.7) -- (214.55,174.55) ;
\draw [line width=1.5]    (214.55,174.55) -- (208.55,231.55) ;
\draw [line width=1.5]    (295.55,73.55) -- (271.55,123.55) ;
\draw  [fill={rgb, 255:red, 255; green, 255; blue, 255 }  ,fill opacity=1 ][line width=1.5]  (261.8,123.55) .. controls (261.8,118.17) and (266.17,113.8) .. (271.55,113.8) .. controls (276.93,113.8) and (281.3,118.17) .. (281.3,123.55) .. controls (281.3,128.93) and (276.93,133.3) .. (271.55,133.3) .. controls (266.17,133.3) and (261.8,128.93) .. (261.8,123.55) -- cycle ;
\draw  [fill={rgb, 255:red, 255; green, 255; blue, 255 }  ,fill opacity=1 ][line width=1.5]  (285.8,73.55) .. controls (285.8,68.17) and (290.17,63.8) .. (295.55,63.8) .. controls (300.93,63.8) and (305.3,68.17) .. (305.3,73.55) .. controls (305.3,78.93) and (300.93,83.3) .. (295.55,83.3) .. controls (290.17,83.3) and (285.8,78.93) .. (285.8,73.55) -- cycle ;
\draw  [fill={rgb, 255:red, 255; green, 255; blue, 255 }  ,fill opacity=1 ][line width=1.5]  (227.8,53.55) .. controls (227.8,48.17) and (232.17,43.8) .. (237.55,43.8) .. controls (242.93,43.8) and (247.3,48.17) .. (247.3,53.55) .. controls (247.3,58.93) and (242.93,63.3) .. (237.55,63.3) .. controls (232.17,63.3) and (227.8,58.93) .. (227.8,53.55) -- cycle ;
\draw  [fill={rgb, 255:red, 255; green, 255; blue, 255 }  ,fill opacity=1 ][line width=1.5]  (167.8,75.55) .. controls (167.8,70.17) and (172.17,65.8) .. (177.55,65.8) .. controls (182.93,65.8) and (187.3,70.17) .. (187.3,75.55) .. controls (187.3,80.93) and (182.93,85.3) .. (177.55,85.3) .. controls (172.17,85.3) and (167.8,80.93) .. (167.8,75.55) -- cycle ;
\draw  [fill={rgb, 255:red, 255; green, 255; blue, 255 }  ,fill opacity=1 ][line width=1.5]  (134.85,122.3) .. controls (134.85,116.92) and (139.22,112.55) .. (144.6,112.55) .. controls (149.98,112.55) and (154.35,116.92) .. (154.35,122.3) .. controls (154.35,127.68) and (149.98,132.05) .. (144.6,132.05) .. controls (139.22,132.05) and (134.85,127.68) .. (134.85,122.3) -- cycle ;
\draw  [fill={rgb, 255:red, 255; green, 255; blue, 255 }  ,fill opacity=1 ][line width=1.5]  (141.8,190.55) .. controls (141.8,185.17) and (146.17,180.8) .. (151.55,180.8) .. controls (156.93,180.8) and (161.3,185.17) .. (161.3,190.55) .. controls (161.3,195.93) and (156.93,200.3) .. (151.55,200.3) .. controls (146.17,200.3) and (141.8,195.93) .. (141.8,190.55) -- cycle ;
\draw  [fill={rgb, 255:red, 255; green, 255; blue, 255 }  ,fill opacity=1 ][line width=1.5]  (204.8,174.55) .. controls (204.8,169.17) and (209.17,164.8) .. (214.55,164.8) .. controls (219.93,164.8) and (224.3,169.17) .. (224.3,174.55) .. controls (224.3,179.93) and (219.93,184.3) .. (214.55,184.3) .. controls (209.17,184.3) and (204.8,179.93) .. (204.8,174.55) -- cycle ;
\draw  [fill={rgb, 255:red, 255; green, 255; blue, 255 }  ,fill opacity=1 ][line width=1.5]  (198.8,231.55) .. controls (198.8,226.17) and (203.17,221.8) .. (208.55,221.8) .. controls (213.93,221.8) and (218.3,226.17) .. (218.3,231.55) .. controls (218.3,236.93) and (213.93,241.3) .. (208.55,241.3) .. controls (203.17,241.3) and (198.8,236.93) .. (198.8,231.55) -- cycle ;
\draw  [fill={rgb, 255:red, 255; green, 255; blue, 255 }  ,fill opacity=1 ][line width=1.5]  (271.8,231.55) .. controls (271.8,226.17) and (276.17,221.8) .. (281.55,221.8) .. controls (286.93,221.8) and (291.3,226.17) .. (291.3,231.55) .. controls (291.3,236.93) and (286.93,241.3) .. (281.55,241.3) .. controls (276.17,241.3) and (271.8,236.93) .. (271.8,231.55) -- cycle ;
\draw  [fill={rgb, 255:red, 255; green, 255; blue, 255 }  ,fill opacity=1 ][line width=1.5]  (330.8,191.55) .. controls (330.8,186.17) and (335.17,181.8) .. (340.55,181.8) .. controls (345.93,181.8) and (350.3,186.17) .. (350.3,191.55) .. controls (350.3,196.93) and (345.93,201.3) .. (340.55,201.3) .. controls (335.17,201.3) and (330.8,196.93) .. (330.8,191.55) -- cycle ;
\draw  [fill={rgb, 255:red, 255; green, 255; blue, 255 }  ,fill opacity=1 ][line width=1.5]  (262.8,174.55) .. controls (262.8,169.17) and (267.17,164.8) .. (272.55,164.8) .. controls (277.93,164.8) and (282.3,169.17) .. (282.3,174.55) .. controls (282.3,179.93) and (277.93,184.3) .. (272.55,184.3) .. controls (267.17,184.3) and (262.8,179.93) .. (262.8,174.55) -- cycle ;
\draw  [fill={rgb, 255:red, 255; green, 255; blue, 255 }  ,fill opacity=1 ][line width=1.5]  (207.8,124.55) .. controls (207.8,119.17) and (212.17,114.8) .. (217.55,114.8) .. controls (222.93,114.8) and (227.3,119.17) .. (227.3,124.55) .. controls (227.3,129.93) and (222.93,134.3) .. (217.55,134.3) .. controls (212.17,134.3) and (207.8,129.93) .. (207.8,124.55) -- cycle ;
\draw [line width=1.5]    (402.55,88.55) -- (345.02,118.02) ;
\draw [line width=1.5]    (402.55,88.55) -- (405.55,155.55) ;
\draw [line width=1.5]    (345.02,118.02) -- (405.55,155.55) ;
\draw  [fill={rgb, 255:red, 255; green, 255; blue, 255 }  ,fill opacity=1 ][line width=1.5]  (392.8,88.55) .. controls (392.8,83.17) and (397.17,78.8) .. (402.55,78.8) .. controls (407.93,78.8) and (412.3,83.17) .. (412.3,88.55) .. controls (412.3,93.93) and (407.93,98.3) .. (402.55,98.3) .. controls (397.17,98.3) and (392.8,93.93) .. (392.8,88.55) -- cycle ;
\draw  [fill={rgb, 255:red, 255; green, 255; blue, 255 }  ,fill opacity=1 ][line width=1.5]  (395.8,155.55) .. controls (395.8,150.17) and (400.17,145.8) .. (405.55,145.8) .. controls (410.93,145.8) and (415.3,150.17) .. (415.3,155.55) .. controls (415.3,160.93) and (410.93,165.3) .. (405.55,165.3) .. controls (400.17,165.3) and (395.8,160.93) .. (395.8,155.55) -- cycle ;
\draw  [fill={rgb, 255:red, 255; green, 255; blue, 255 }  ,fill opacity=1 ][line width=1.5]  (334.73,118.02) .. controls (334.73,112.34) and (339.34,107.73) .. (345.02,107.73) .. controls (350.7,107.73) and (355.3,112.34) .. (355.3,118.02) .. controls (355.3,123.7) and (350.7,128.3) .. (345.02,128.3) .. controls (339.34,128.3) and (334.73,123.7) .. (334.73,118.02) -- cycle ;

\end{tikzpicture}}
    \end{minipage}
         & 
    \begin{minipage}{0.35 \textwidth}
        \scalebox{0.6}{
    \input{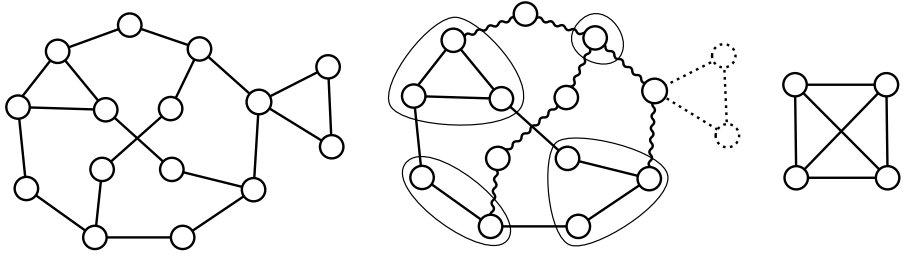}}
    \end{minipage}
    &
    \begin{minipage}{0.2 \textwidth}
        \scalebox{0.6}{
    \tikzset{every picture/.style={line width=0.75pt}} 

\begin{tikzpicture}[x=0.75pt,y=0.75pt,yscale=-1,xscale=1]

\draw [line width=1.5]    (574.55,94.55) -- (498.55,94.55) ;
\draw [line width=1.5]    (499.55,172.55) -- (498.55,94.55) ;
\draw [line width=1.5]    (575.55,172.55) -- (574.55,94.55) ;
\draw [line width=1.5]    (499.55,172.55) -- (575.55,172.55) ;
\draw [line width=1.5]    (499.55,172.55) -- (574.55,94.55) ;
\draw [line width=1.5]    (575.55,172.55) -- (498.55,94.55) ;
\draw  [fill={rgb, 255:red, 255; green, 255; blue, 255 }  ,fill opacity=1 ][line width=1.5]  (564.8,94.55) .. controls (564.8,89.17) and (569.17,84.8) .. (574.55,84.8) .. controls (579.93,84.8) and (584.3,89.17) .. (584.3,94.55) .. controls (584.3,99.93) and (579.93,104.3) .. (574.55,104.3) .. controls (569.17,104.3) and (564.8,99.93) .. (564.8,94.55) -- cycle ;
\draw  [fill={rgb, 255:red, 255; green, 255; blue, 255 }  ,fill opacity=1 ][line width=1.5]  (565.8,172.55) .. controls (565.8,167.17) and (570.17,162.8) .. (575.55,162.8) .. controls (580.93,162.8) and (585.3,167.17) .. (585.3,172.55) .. controls (585.3,177.93) and (580.93,182.3) .. (575.55,182.3) .. controls (570.17,182.3) and (565.8,177.93) .. (565.8,172.55) -- cycle ;
\draw  [fill={rgb, 255:red, 255; green, 255; blue, 255 }  ,fill opacity=1 ][line width=1.5]  (489.8,172.55) .. controls (489.8,167.17) and (494.17,162.8) .. (499.55,162.8) .. controls (504.93,162.8) and (509.3,167.17) .. (509.3,172.55) .. controls (509.3,177.93) and (504.93,182.3) .. (499.55,182.3) .. controls (494.17,182.3) and (489.8,177.93) .. (489.8,172.55) -- cycle ;
\draw  [fill={rgb, 255:red, 255; green, 255; blue, 255 }  ,fill opacity=1 ][line width=1.5]  (488.8,94.55) .. controls (488.8,89.17) and (493.17,84.8) .. (498.55,84.8) .. controls (503.93,84.8) and (508.3,89.17) .. (508.3,94.55) .. controls (508.3,99.93) and (503.93,104.3) .. (498.55,104.3) .. controls (493.17,104.3) and (488.8,99.93) .. (488.8,94.55) -- cycle ;

\end{tikzpicture}}
    \end{minipage}
    \end{tabular}
    
    \caption{The graph on the left has a $K_4$ minor. Indeed, the bags of the second definition are depicted in the picture in the middle, and it is easy to find the six disjoint paths that link them. 
   Alternatively, one can get a $K_4$ like the one of the right-most picture by contracting all the edges inside the bags, contracting the wavy paths between bags into edges, and deleting the dotted vertices and edges.}
    \label{fig:K4-minor}
\end{figure}

An important observation is that, if we take a collection $F_1$, ..., $F_k$ of $K_4$-minor-free graphs, and organize them into a tree, by identifying pairs of vertices like in Figure~\ref{fig:block-cut-tree}, we get a $K_4$-minor-free graph. 

\begin{figure}[!h]
    \centering
    \begin{tabular}{cc}
    \begin{minipage}{0.5 \textwidth}
        \scalebox{0.6}{
    \tikzset{every picture/.style={line width=0.75pt}} 

\begin{tikzpicture}[x=0.75pt,y=0.75pt,yscale=-1,xscale=1]

\draw [line width=1.5]    (273.3,100.55) -- (317.3,56.55) ;
\draw [line width=1.5]    (363.55,105.55) -- (317.3,56.55) ;
\draw [line width=1.5]    (363.55,105.55) -- (274.6,103.3) ;
\draw  [fill={rgb, 255:red, 255; green, 255; blue, 255 }  ,fill opacity=1 ][line width=1.5]  (307.55,56.55) .. controls (307.55,51.17) and (311.92,46.8) .. (317.3,46.8) .. controls (322.68,46.8) and (327.05,51.17) .. (327.05,56.55) .. controls (327.05,61.93) and (322.68,66.3) .. (317.3,66.3) .. controls (311.92,66.3) and (307.55,61.93) .. (307.55,56.55) -- cycle ;
\draw  [fill={rgb, 255:red, 255; green, 255; blue, 255 }  ,fill opacity=1 ][line width=1.5]  (264.85,103.3) .. controls (264.85,97.92) and (269.22,93.55) .. (274.6,93.55) .. controls (279.98,93.55) and (284.35,97.92) .. (284.35,103.3) .. controls (284.35,108.68) and (279.98,113.05) .. (274.6,113.05) .. controls (269.22,113.05) and (264.85,108.68) .. (264.85,103.3) -- cycle ;
\draw [line width=1.5]    (282.55,195.55) -- (326.55,161.3) ;
\draw [line width=1.5]    (366.55,200.55) -- (326.55,161.3) ;
\draw  [fill={rgb, 255:red, 255; green, 255; blue, 255 }  ,fill opacity=1 ][line width=1.5]  (316.8,161.3) .. controls (316.8,155.92) and (321.17,151.55) .. (326.55,151.55) .. controls (331.93,151.55) and (336.3,155.92) .. (336.3,161.3) .. controls (336.3,166.68) and (331.93,171.05) .. (326.55,171.05) .. controls (321.17,171.05) and (316.8,166.68) .. (316.8,161.3) -- cycle ;
\draw [line width=1.5]    (157.55,158.55) -- (215.55,159.55) ;
\draw [line width=1.5]    (322.6,232.3) -- (366.55,200.55) ;
\draw  [fill={rgb, 255:red, 255; green, 255; blue, 255 }  ,fill opacity=1 ][line width=1.5]  (147.8,158.55) .. controls (147.8,153.17) and (152.17,148.8) .. (157.55,148.8) .. controls (162.93,148.8) and (167.3,153.17) .. (167.3,158.55) .. controls (167.3,163.93) and (162.93,168.3) .. (157.55,168.3) .. controls (152.17,168.3) and (147.8,163.93) .. (147.8,158.55) -- cycle ;
\draw [line width=1.5]    (156.55,214.55) -- (214.55,215.55) ;
\draw [line width=1.5]    (282.55,195.55) -- (322.6,232.3) ;
\draw [line width=1.5]    (420.55,175.55) -- (418.55,122.55) ;
\draw [line width=1.5]    (418.55,122.55) -- (488.55,98.55) ;
\draw [line width=1.5]    (420.55,175.55) -- (493.55,191.55) ;
\draw [line width=1.5]    (493.55,191.55) -- (541.55,141.55) ;
\draw [line width=1.5]    (488.55,98.55) -- (541.55,141.55) ;
\draw [line width=1.5]    (488.55,98.55) -- (420.55,175.55) ;
\draw [line width=1.5]    (488.55,98.55) -- (493.55,191.55) ;
\draw  [fill={rgb, 255:red, 255; green, 255; blue, 255 }  ,fill opacity=1 ][line width=1.5]  (312.85,232.3) .. controls (312.85,226.92) and (317.22,222.55) .. (322.6,222.55) .. controls (327.98,222.55) and (332.35,226.92) .. (332.35,232.3) .. controls (332.35,237.68) and (327.98,242.05) .. (322.6,242.05) .. controls (317.22,242.05) and (312.85,237.68) .. (312.85,232.3) -- cycle ;
\draw  [fill={rgb, 255:red, 255; green, 255; blue, 255 }  ,fill opacity=1 ][line width=1.5]  (146.8,214.55) .. controls (146.8,209.17) and (151.17,204.8) .. (156.55,204.8) .. controls (161.93,204.8) and (166.3,209.17) .. (166.3,214.55) .. controls (166.3,219.93) and (161.93,224.3) .. (156.55,224.3) .. controls (151.17,224.3) and (146.8,219.93) .. (146.8,214.55) -- cycle ;
\draw  [fill={rgb, 255:red, 255; green, 255; blue, 255 }  ,fill opacity=1 ][line width=1.5]  (483.8,191.55) .. controls (483.8,186.17) and (488.17,181.8) .. (493.55,181.8) .. controls (498.93,181.8) and (503.3,186.17) .. (503.3,191.55) .. controls (503.3,196.93) and (498.93,201.3) .. (493.55,201.3) .. controls (488.17,201.3) and (483.8,196.93) .. (483.8,191.55) -- cycle ;
\draw  [fill={rgb, 255:red, 255; green, 255; blue, 255 }  ,fill opacity=1 ][line width=1.5]  (531.8,141.55) .. controls (531.8,136.17) and (536.17,131.8) .. (541.55,131.8) .. controls (546.93,131.8) and (551.3,136.17) .. (551.3,141.55) .. controls (551.3,146.93) and (546.93,151.3) .. (541.55,151.3) .. controls (536.17,151.3) and (531.8,146.93) .. (531.8,141.55) -- cycle ;
\draw  [fill={rgb, 255:red, 255; green, 255; blue, 255 }  ,fill opacity=1 ][line width=1.5]  (478.8,98.55) .. controls (478.8,93.17) and (483.17,88.8) .. (488.55,88.8) .. controls (493.93,88.8) and (498.3,93.17) .. (498.3,98.55) .. controls (498.3,103.93) and (493.93,108.3) .. (488.55,108.3) .. controls (483.17,108.3) and (478.8,103.93) .. (478.8,98.55) -- cycle ;
\draw [line width=1.5]  [dash pattern={on 1.69pt off 2.76pt}]  (215.55,159.55) -- (282.55,195.55) ;
\draw [line width=1.5]  [dash pattern={on 1.69pt off 2.76pt}]  (214.55,215.55) -- (282.55,195.55) ;
\draw [line width=1.5]  [dash pattern={on 1.69pt off 2.76pt}]  (366.55,200.55) -- (420.55,175.55) ;
\draw [line width=1.5]  [dash pattern={on 1.69pt off 2.76pt}]  (363.55,105.55) -- (418.55,122.55) ;
\draw  [fill={rgb, 255:red, 255; green, 255; blue, 255 }  ,fill opacity=1 ][line width=1.5]  (353.8,105.55) .. controls (353.8,100.17) and (358.17,95.8) .. (363.55,95.8) .. controls (368.93,95.8) and (373.3,100.17) .. (373.3,105.55) .. controls (373.3,110.93) and (368.93,115.3) .. (363.55,115.3) .. controls (358.17,115.3) and (353.8,110.93) .. (353.8,105.55) -- cycle ;
\draw  [fill={rgb, 255:red, 255; green, 255; blue, 255 }  ,fill opacity=1 ][line width=1.5]  (408.8,122.55) .. controls (408.8,117.17) and (413.17,112.8) .. (418.55,112.8) .. controls (423.93,112.8) and (428.3,117.17) .. (428.3,122.55) .. controls (428.3,127.93) and (423.93,132.3) .. (418.55,132.3) .. controls (413.17,132.3) and (408.8,127.93) .. (408.8,122.55) -- cycle ;
\draw  [fill={rgb, 255:red, 255; green, 255; blue, 255 }  ,fill opacity=1 ][line width=1.5]  (356.8,200.55) .. controls (356.8,195.17) and (361.17,190.8) .. (366.55,190.8) .. controls (371.93,190.8) and (376.3,195.17) .. (376.3,200.55) .. controls (376.3,205.93) and (371.93,210.3) .. (366.55,210.3) .. controls (361.17,210.3) and (356.8,205.93) .. (356.8,200.55) -- cycle ;
\draw  [fill={rgb, 255:red, 255; green, 255; blue, 255 }  ,fill opacity=1 ][line width=1.5]  (410.8,175.55) .. controls (410.8,170.17) and (415.17,165.8) .. (420.55,165.8) .. controls (425.93,165.8) and (430.3,170.17) .. (430.3,175.55) .. controls (430.3,180.93) and (425.93,185.3) .. (420.55,185.3) .. controls (415.17,185.3) and (410.8,180.93) .. (410.8,175.55) -- cycle ;
\draw  [fill={rgb, 255:red, 255; green, 255; blue, 255 }  ,fill opacity=1 ][line width=1.5]  (205.8,159.55) .. controls (205.8,154.17) and (210.17,149.8) .. (215.55,149.8) .. controls (220.93,149.8) and (225.3,154.17) .. (225.3,159.55) .. controls (225.3,164.93) and (220.93,169.3) .. (215.55,169.3) .. controls (210.17,169.3) and (205.8,164.93) .. (205.8,159.55) -- cycle ;
\draw  [fill={rgb, 255:red, 255; green, 255; blue, 255 }  ,fill opacity=1 ][line width=1.5]  (204.8,215.55) .. controls (204.8,210.17) and (209.17,205.8) .. (214.55,205.8) .. controls (219.93,205.8) and (224.3,210.17) .. (224.3,215.55) .. controls (224.3,220.93) and (219.93,225.3) .. (214.55,225.3) .. controls (209.17,225.3) and (204.8,220.93) .. (204.8,215.55) -- cycle ;
\draw  [fill={rgb, 255:red, 255; green, 255; blue, 255 }  ,fill opacity=1 ][line width=1.5]  (272.8,195.55) .. controls (272.8,190.17) and (277.17,185.8) .. (282.55,185.8) .. controls (287.93,185.8) and (292.3,190.17) .. (292.3,195.55) .. controls (292.3,200.93) and (287.93,205.3) .. (282.55,205.3) .. controls (277.17,205.3) and (272.8,200.93) .. (272.8,195.55) -- cycle ;

\end{tikzpicture}}
    \end{minipage}
         & 
    \begin{minipage}{0.4 \textwidth}
        \scalebox{0.6}{
    \tikzset{every picture/.style={line width=0.75pt}} 

\begin{tikzpicture}[x=0.75pt,y=0.75pt,yscale=-1,xscale=1]

\draw [line width=1.5]    (328.3,117.55) -- (372.3,73.55) ;
\draw [line width=1.5]    (418.55,122.55) -- (372.3,73.55) ;
\draw [line width=1.5]    (418.55,122.55) -- (329.6,120.3) ;
\draw  [fill={rgb, 255:red, 255; green, 255; blue, 255 }  ,fill opacity=1 ][line width=1.5]  (362.55,73.55) .. controls (362.55,68.17) and (366.92,63.8) .. (372.3,63.8) .. controls (377.68,63.8) and (382.05,68.17) .. (382.05,73.55) .. controls (382.05,78.93) and (377.68,83.3) .. (372.3,83.3) .. controls (366.92,83.3) and (362.55,78.93) .. (362.55,73.55) -- cycle ;
\draw  [fill={rgb, 255:red, 255; green, 255; blue, 255 }  ,fill opacity=1 ][line width=1.5]  (319.85,120.3) .. controls (319.85,114.92) and (324.22,110.55) .. (329.6,110.55) .. controls (334.98,110.55) and (339.35,114.92) .. (339.35,120.3) .. controls (339.35,125.68) and (334.98,130.05) .. (329.6,130.05) .. controls (324.22,130.05) and (319.85,125.68) .. (319.85,120.3) -- cycle ;
\draw [line width=1.5]    (338.91,194.49) -- (371.44,149.21) ;
\draw [line width=1.5]    (420.9,175.56) -- (371.44,149.21) ;
\draw  [fill={rgb, 255:red, 255; green, 255; blue, 255 }  ,fill opacity=1 ][line width=1.5]  (362.09,151.96) .. controls (360.57,146.79) and (363.52,141.37) .. (368.69,139.85) .. controls (373.85,138.33) and (379.27,141.29) .. (380.79,146.45) .. controls (382.32,151.62) and (379.36,157.04) .. (374.2,158.56) .. controls (369.03,160.08) and (363.61,157.13) .. (362.09,151.96) -- cycle ;
\draw [line width=1.5]    (288.71,165.07) -- (338.39,195.03) ;
\draw [line width=1.5]    (387.71,218.43) -- (420.9,175.56) ;
\draw  [fill={rgb, 255:red, 255; green, 255; blue, 255 }  ,fill opacity=1 ][line width=1.5]  (280.28,160.18) .. controls (282.98,155.53) and (288.94,153.94) .. (293.6,156.64) .. controls (298.26,159.34) and (299.85,165.3) .. (297.15,169.96) .. controls (294.45,174.62) and (288.48,176.21) .. (283.82,173.51) .. controls (279.16,170.81) and (277.58,164.84) .. (280.28,160.18) -- cycle ;
\draw [line width=1.5]    (286.92,222.62) -- (338.18,195.48) ;
\draw [line width=1.5]    (338.91,194.49) -- (387.71,218.43) ;
\draw [line width=1.5]    (420.55,175.55) -- (418.55,122.55) ;
\draw [line width=1.5]    (418.55,122.55) -- (488.55,98.55) ;
\draw [line width=1.5]    (420.55,175.55) -- (493.55,191.55) ;
\draw [line width=1.5]    (493.55,191.55) -- (541.55,141.55) ;
\draw [line width=1.5]    (488.55,98.55) -- (541.55,141.55) ;
\draw [line width=1.5]    (488.55,98.55) -- (420.55,175.55) ;
\draw [line width=1.5]    (488.55,98.55) -- (493.55,191.55) ;
\draw  [fill={rgb, 255:red, 255; green, 255; blue, 255 }  ,fill opacity=1 ][line width=1.5]  (378.36,221.18) .. controls (376.84,216.02) and (379.79,210.6) .. (384.96,209.08) .. controls (390.12,207.56) and (395.54,210.51) .. (397.06,215.67) .. controls (398.59,220.84) and (395.63,226.26) .. (390.47,227.78) .. controls (385.3,229.3) and (379.88,226.35) .. (378.36,221.18) -- cycle ;
\draw  [fill={rgb, 255:red, 255; green, 255; blue, 255 }  ,fill opacity=1 ][line width=1.5]  (278.38,227.33) .. controls (275.78,222.61) and (277.49,216.68) .. (282.21,214.08) .. controls (286.92,211.48) and (292.85,213.19) .. (295.45,217.91) .. controls (298.05,222.62) and (296.34,228.56) .. (291.62,231.16) .. controls (286.91,233.76) and (280.98,232.04) .. (278.38,227.33) -- cycle ;
\draw  [fill={rgb, 255:red, 255; green, 255; blue, 255 }  ,fill opacity=1 ][line width=1.5]  (483.8,191.55) .. controls (483.8,186.17) and (488.17,181.8) .. (493.55,181.8) .. controls (498.93,181.8) and (503.3,186.17) .. (503.3,191.55) .. controls (503.3,196.93) and (498.93,201.3) .. (493.55,201.3) .. controls (488.17,201.3) and (483.8,196.93) .. (483.8,191.55) -- cycle ;
\draw  [fill={rgb, 255:red, 255; green, 255; blue, 255 }  ,fill opacity=1 ][line width=1.5]  (531.8,141.55) .. controls (531.8,136.17) and (536.17,131.8) .. (541.55,131.8) .. controls (546.93,131.8) and (551.3,136.17) .. (551.3,141.55) .. controls (551.3,146.93) and (546.93,151.3) .. (541.55,151.3) .. controls (536.17,151.3) and (531.8,146.93) .. (531.8,141.55) -- cycle ;
\draw  [fill={rgb, 255:red, 255; green, 255; blue, 255 }  ,fill opacity=1 ][line width=1.5]  (478.8,98.55) .. controls (478.8,93.17) and (483.17,88.8) .. (488.55,88.8) .. controls (493.93,88.8) and (498.3,93.17) .. (498.3,98.55) .. controls (498.3,103.93) and (493.93,108.3) .. (488.55,108.3) .. controls (483.17,108.3) and (478.8,103.93) .. (478.8,98.55) -- cycle ;
\draw  [fill={rgb, 255:red, 155; green, 155; blue, 155 }  ,fill opacity=1 ][line width=1.5]  (408.8,122.55) .. controls (408.8,117.17) and (413.17,112.8) .. (418.55,112.8) .. controls (423.93,112.8) and (428.3,117.17) .. (428.3,122.55) .. controls (428.3,127.93) and (423.93,132.3) .. (418.55,132.3) .. controls (413.17,132.3) and (408.8,127.93) .. (408.8,122.55) -- cycle ;
\draw  [fill={rgb, 255:red, 155; green, 155; blue, 155 }  ,fill opacity=1 ][line width=1.5]  (411.55,178.31) .. controls (410.03,173.15) and (412.98,167.72) .. (418.15,166.2) .. controls (423.31,164.68) and (428.73,167.64) .. (430.25,172.8) .. controls (431.78,177.97) and (428.82,183.39) .. (423.66,184.91) .. controls (418.49,186.43) and (413.07,183.48) .. (411.55,178.31) -- cycle ;
\draw  [fill={rgb, 255:red, 155; green, 155; blue, 155 }  ,fill opacity=1 ][line width=1.5]  (329.95,190.14) .. controls (332.65,185.48) and (338.62,183.89) .. (343.28,186.59) .. controls (347.94,189.29) and (349.52,195.26) .. (346.82,199.92) .. controls (344.12,204.57) and (338.16,206.16) .. (333.5,203.46) .. controls (328.84,200.76) and (327.25,194.8) .. (329.95,190.14) -- cycle ;

\end{tikzpicture}}
    \end{minipage}
    \end{tabular}    
    \caption{The five graphs with plain edges on the left picture are $K_4$-minor free. Organizing them into a tree by identifying the nodes linked by dotted edges makes a larger $K_4$-minor-free graph. }
    \label{fig:block-cut-tree}
\end{figure}
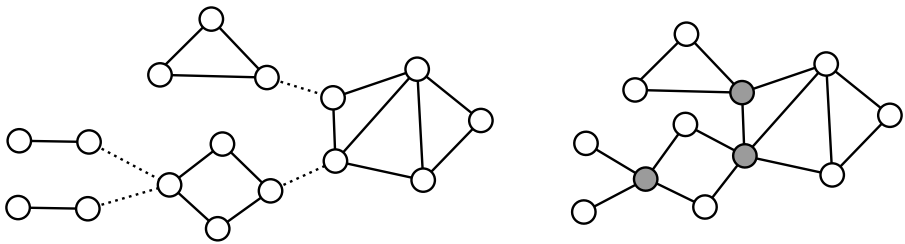

To see that, suppose that the graph we created has a $K_4$-minor. 
Then there exist bags and paths as described above. 
If the bags and paths are all contained in the same former $F_i$, then this $F_i$ would not be $K_4$-minor-free, which is a contradiction. 
If it is not the case, then the bags and paths use vertices that belong to different subgraphs $F_i$ and $F_j$. And because of connectivity, they should use a vertex $v$ that connects two such subgraphs (grey vertices in Figure~\ref{fig:block-cut-tree}). 
Then the bags and paths cannot be vertex-disjoint as required, because at least two of them should use the vertex $v$.

As a consequence of the observation above, a classic way to study $K_4$-minor-free graphs (as well as other classes) is to decompose the graph into maximal 2-connected components organized into a tree. 
This is called the \emph{block-cut tree} of the graph, where every maximal 2-connected component is called a \emph{block}. (Figure~\ref{fig:block-cut-tree} actually show the block-cut structure of the right-most graph.)
This is relevant here because $2$-connected $K_4$-minor-free graphs have a specific structure; we will come back to this later.

Now, from the certification point of view, there is a natural strategy: first certify the structure of the block-cut tree, and then certify the special structure of each block. There are several challenges to face with this approach. First, to certify the block-cut tree, it is essential to be able to certify the connectivity of the blocks. 
Second, we need to avoid what we call certificate congestion, which is the issue of having too large certificates because we use too many layers of certification on some nodes. We now detail these two aspects, starting with the latter.

\paragraph*{Avoiding certificate congestion}
In the block-cut tee of a graph, the blocks are attached to each other by shared vertices, the \emph{cut vertices}. There is no bound on the number of blocks that are attached to a given cut vertex, and this is problematic for certification. 
Indeed, we cannot give to every node the list of the blocks it belongs to, as we aim for $O(\log n)$ certificates, and such a list could contain $\Omega(n)$ blocks. 
And even if we could fix the certification of the block-cut tree, the same problem would appear with the certification of the specific structure of each block: the cut vertices would have to hold a piece of certification for each block. 

We basically have two tools to deal with this problem. The first one is not new, it is a degeneracy argument that already appeared in \cite{FeuilloleyFMRRT20, FeuilloleyFMRRT21}. 
A graph is $k$-degenerate if in every subgraph there exists a vertex that has degree at most $k$. Intuitively (and a bit incorrectly), this means that when we need to put a large certificate on a vertex, we can spread it on its some of its neighbors that have lower degree.
A more precise statement is that, for $k$-degenerate graphs, we can transform a certification with $O(f(n))$ labels \emph{on the edges of the graphs}, into a classic certification with $O(k \cdot f(n))$ labels on the vertices. 
This is relevant for our problem, as a priori there is less congestion on the edges, and minor-free classes have bounded degeneracy.
Unfortunately, this is not enough for our purpose. We then build a second, more versatile tool. 
It consists in proving that it is possible to transform in mechanical way any certification of a graph or subgraph, into a certification that would put an empty certificate on some given vertex. 
Once we have this tool, we can adapt the certification of the blocks to work well in the block-cut tree: build the block-cut tree by adding blocks iteratively, making sure that the connecting node has an empty label in the certification of the newly added block.

See Section~\ref{sec:avoiding-congestion} for the details on this topic.

\paragraph*{Certifying connectivity properties}

Connectivity properties have been studied before in distributed certification. 
Specifically, certifying that for two given vertices $s$ and $t$, the $st$-connectivity is at least $k$ has been studied in \cite{KormanKP10} and~\cite{GoosS16}. 
But here we are interested in the connectivity of the graph itself, or in other words, in the $st$-connectivity between any pair of vertices. 
Clearly, proving $st$-connectivity for any pair using the schemes of the literature would lead to huge certificates. 
Instead, we use the characterizations of $k$-connected graphs that are known for small values of $k$. There are various such characterizations, but they are all based on the same idea of \emph{ear decomposition}. 

To explain ear decompositions, consider a graph that we can build the following way (see Figure~\ref{fig:ear-decompo}). 
Start from an edge, and iteratively apply the following process: take two different nodes of the current graph and link them by a path whose internal nodes are new nodes of the graph. It is not hard to see that such a graph is always 2-connected. 
Remarkably, the converse is also true: any 2-connected graph can be built (or decomposed this way). This is called an open ear decomposition, and similar constructions characterize 2-edge connected graphs and 3-vertex-connected graphs.

\begin{figure}[!h]
    \centering
    \begin{tabular}{cc}
    \begin{minipage}{0.4 \textwidth}
        \scalebox{0.7}{
    \tikzset{every picture/.style={line width=0.75pt}} 

\begin{tikzpicture}[x=0.75pt,y=0.75pt,yscale=-1,xscale=1]

\draw [line width=1.5]    (227.65,51.65) -- (167.3,117.65) ;
\draw [line width=1.5]    (328.65,51.65) -- (227.65,51.65) ;
\draw [line width=1.5]    (328.65,187.65) -- (227.65,187.65) ;
\draw [line width=1.5]    (227.65,187.65) -- (167.3,117.65) ;
\draw [line width=1.5]    (384.65,117.3) -- (367.82,136.99) -- (328.65,187.65) ;
\draw [line width=1.5]    (384.65,117.3) -- (328.65,51.65) ;
\draw [line width=1.5]    (328.65,187.65) -- (328.65,51.65) ;
\draw [line width=1.5]    (227.65,187.65) -- (328.65,51.65) ;
\draw  [fill={rgb, 255:red, 255; green, 255; blue, 255 }  ,fill opacity=1 ][line width=1.5]  (216,51.65) .. controls (216,45.22) and (221.22,40) .. (227.65,40) .. controls (234.08,40) and (239.3,45.22) .. (239.3,51.65) .. controls (239.3,58.08) and (234.08,63.3) .. (227.65,63.3) .. controls (221.22,63.3) and (216,58.08) .. (216,51.65) -- cycle ;
\draw  [fill={rgb, 255:red, 255; green, 255; blue, 255 }  ,fill opacity=1 ][line width=1.5]  (317,51.65) .. controls (317,45.22) and (322.22,40) .. (328.65,40) .. controls (335.08,40) and (340.3,45.22) .. (340.3,51.65) .. controls (340.3,58.08) and (335.08,63.3) .. (328.65,63.3) .. controls (322.22,63.3) and (317,58.08) .. (317,51.65) -- cycle ;
\draw  [fill={rgb, 255:red, 255; green, 255; blue, 255 }  ,fill opacity=1 ][line width=1.5]  (317,187.65) .. controls (317,181.22) and (322.22,176) .. (328.65,176) .. controls (335.08,176) and (340.3,181.22) .. (340.3,187.65) .. controls (340.3,194.08) and (335.08,199.3) .. (328.65,199.3) .. controls (322.22,199.3) and (317,194.08) .. (317,187.65) -- cycle ;
\draw  [fill={rgb, 255:red, 255; green, 255; blue, 255 }  ,fill opacity=1 ][line width=1.5]  (373,117.3) .. controls (373,110.87) and (378.22,105.65) .. (384.65,105.65) .. controls (391.08,105.65) and (396.3,110.87) .. (396.3,117.3) .. controls (396.3,123.73) and (391.08,128.95) .. (384.65,128.95) .. controls (378.22,128.95) and (373,123.73) .. (373,117.3) -- cycle ;
\draw  [fill={rgb, 255:red, 255; green, 255; blue, 255 }  ,fill opacity=1 ][line width=1.5]  (216,187.65) .. controls (216,181.22) and (221.22,176) .. (227.65,176) .. controls (234.08,176) and (239.3,181.22) .. (239.3,187.65) .. controls (239.3,194.08) and (234.08,199.3) .. (227.65,199.3) .. controls (221.22,199.3) and (216,194.08) .. (216,187.65) -- cycle ;
\draw  [fill={rgb, 255:red, 255; green, 255; blue, 255 }  ,fill opacity=1 ][line width=1.5]  (155.65,117.65) .. controls (155.65,111.22) and (160.87,106) .. (167.3,106) .. controls (173.73,106) and (178.95,111.22) .. (178.95,117.65) .. controls (178.95,124.08) and (173.73,129.3) .. (167.3,129.3) .. controls (160.87,129.3) and (155.65,124.08) .. (155.65,117.65) -- cycle ;

\end{tikzpicture}}
    \end{minipage}
         & 
    \begin{minipage}{0.4 \textwidth}
        \scalebox{0.7}{
    \tikzset{every picture/.style={line width=0.75pt}} 

\begin{tikzpicture}[x=0.75pt,y=0.75pt,yscale=-1,xscale=1]

\draw [color={rgb, 255:red, 0; green, 0; blue, 0 }  ,draw opacity=1 ][line width=1.5]    (227.65,51.65) -- (167.3,117.65) ;
\draw [color={rgb, 255:red, 0; green, 0; blue, 0 }  ,draw opacity=1 ][line width=1.5]    (328.65,51.65) -- (227.65,51.65) ;
\draw [line width=3.75]    (328.65,187.65) -- (227.65,187.65) ;
\draw [color={rgb, 255:red, 0; green, 0; blue, 0 }  ,draw opacity=1 ][fill={rgb, 255:red, 155; green, 155; blue, 155 }  ,fill opacity=1 ][line width=1.5]    (227.65,187.65) -- (167.3,117.65) ;
\draw [line width=1.5]  [dash pattern={on 1.69pt off 2.76pt}]  (384.65,117.3) -- (367.82,136.99) -- (328.65,187.65) ;
\draw [line width=1.5]  [dash pattern={on 1.69pt off 2.76pt}]  (384.65,117.3) -- (328.65,51.65) ;
\draw [color={rgb, 255:red, 0; green, 0; blue, 0 }  ,draw opacity=1 ][line width=1.5]    (328.65,187.65) -- (328.65,51.65) ;
\draw [line width=1.5]    (227.65,187.65) .. controls (227.3,185.32) and (228.3,183.98) .. (230.63,183.64) .. controls (232.96,183.29) and (233.96,181.95) .. (233.61,179.62) .. controls (233.26,177.29) and (234.26,175.95) .. (236.59,175.61) .. controls (238.92,175.26) and (239.92,173.92) .. (239.57,171.59) .. controls (239.23,169.26) and (240.23,167.92) .. (242.56,167.58) .. controls (244.89,167.24) and (245.89,165.9) .. (245.54,163.57) .. controls (245.19,161.24) and (246.19,159.9) .. (248.52,159.55) .. controls (250.85,159.21) and (251.85,157.87) .. (251.5,155.54) .. controls (251.15,153.21) and (252.15,151.87) .. (254.48,151.52) .. controls (256.81,151.18) and (257.81,149.84) .. (257.46,147.51) .. controls (257.11,145.18) and (258.11,143.84) .. (260.44,143.49) .. controls (262.77,143.15) and (263.77,141.81) .. (263.42,139.48) .. controls (263.07,137.15) and (264.07,135.81) .. (266.4,135.47) .. controls (268.73,135.12) and (269.73,133.78) .. (269.39,131.45) .. controls (269.04,129.12) and (270.04,127.78) .. (272.37,127.44) .. controls (274.7,127.09) and (275.7,125.75) .. (275.35,123.42) .. controls (275,121.09) and (276,119.75) .. (278.33,119.41) .. controls (280.66,119.07) and (281.66,117.73) .. (281.31,115.4) .. controls (280.96,113.07) and (281.96,111.73) .. (284.29,111.38) .. controls (286.62,111.04) and (287.62,109.7) .. (287.27,107.37) .. controls (286.92,105.04) and (287.92,103.7) .. (290.25,103.35) .. controls (292.58,103.01) and (293.58,101.67) .. (293.23,99.34) .. controls (292.88,97.01) and (293.88,95.67) .. (296.21,95.33) .. controls (298.54,94.98) and (299.54,93.64) .. (299.2,91.31) .. controls (298.85,88.98) and (299.85,87.64) .. (302.18,87.3) .. controls (304.51,86.95) and (305.51,85.61) .. (305.16,83.28) .. controls (304.81,80.95) and (305.81,79.61) .. (308.14,79.27) .. controls (310.47,78.92) and (311.47,77.58) .. (311.12,75.25) .. controls (310.77,72.92) and (311.77,71.58) .. (314.1,71.24) .. controls (316.43,70.9) and (317.43,69.56) .. (317.08,67.23) .. controls (316.73,64.9) and (317.73,63.56) .. (320.06,63.21) .. controls (322.39,62.87) and (323.39,61.53) .. (323.04,59.2) .. controls (322.7,56.87) and (323.7,55.53) .. (326.03,55.18) -- (328.65,51.65) -- (328.65,51.65) ;
\draw  [fill={rgb, 255:red, 255; green, 255; blue, 255 }  ,fill opacity=1 ][line width=1.5]  (216,51.65) .. controls (216,45.22) and (221.22,40) .. (227.65,40) .. controls (234.08,40) and (239.3,45.22) .. (239.3,51.65) .. controls (239.3,58.08) and (234.08,63.3) .. (227.65,63.3) .. controls (221.22,63.3) and (216,58.08) .. (216,51.65) -- cycle ;
\draw  [fill={rgb, 255:red, 255; green, 255; blue, 255 }  ,fill opacity=1 ][line width=1.5]  (317,51.65) .. controls (317,45.22) and (322.22,40) .. (328.65,40) .. controls (335.08,40) and (340.3,45.22) .. (340.3,51.65) .. controls (340.3,58.08) and (335.08,63.3) .. (328.65,63.3) .. controls (322.22,63.3) and (317,58.08) .. (317,51.65) -- cycle ;
\draw  [fill={rgb, 255:red, 255; green, 255; blue, 255 }  ,fill opacity=1 ][line width=1.5]  (317,187.65) .. controls (317,181.22) and (322.22,176) .. (328.65,176) .. controls (335.08,176) and (340.3,181.22) .. (340.3,187.65) .. controls (340.3,194.08) and (335.08,199.3) .. (328.65,199.3) .. controls (322.22,199.3) and (317,194.08) .. (317,187.65) -- cycle ;
\draw  [fill={rgb, 255:red, 255; green, 255; blue, 255 }  ,fill opacity=1 ][line width=1.5]  (373,117.3) .. controls (373,110.87) and (378.22,105.65) .. (384.65,105.65) .. controls (391.08,105.65) and (396.3,110.87) .. (396.3,117.3) .. controls (396.3,123.73) and (391.08,128.95) .. (384.65,128.95) .. controls (378.22,128.95) and (373,123.73) .. (373,117.3) -- cycle ;
\draw  [fill={rgb, 255:red, 255; green, 255; blue, 255 }  ,fill opacity=1 ][line width=1.5]  (216,187.65) .. controls (216,181.22) and (221.22,176) .. (227.65,176) .. controls (234.08,176) and (239.3,181.22) .. (239.3,187.65) .. controls (239.3,194.08) and (234.08,199.3) .. (227.65,199.3) .. controls (221.22,199.3) and (216,194.08) .. (216,187.65) -- cycle ;
\draw  [fill={rgb, 255:red, 255; green, 255; blue, 255 }  ,fill opacity=1 ][line width=1.5]  (155.65,117.65) .. controls (155.65,111.22) and (160.87,106) .. (167.3,106) .. controls (173.73,106) and (178.95,111.22) .. (178.95,117.65) .. controls (178.95,124.08) and (173.73,129.3) .. (167.3,129.3) .. controls (160.87,129.3) and (155.65,124.08) .. (155.65,117.65) -- cycle ;

\end{tikzpicture}}
    \end{minipage}
    \end{tabular}       
    \caption{Illustration of an open ear decomposition. The graph on the left can be built with the ear decomposition described on the right. First, put the bold edge. Then add the path of plain edges. Finally, add the dotted path, and the wavy path, which is just one edge.}
    \label{fig:ear-decompo}
\end{figure}
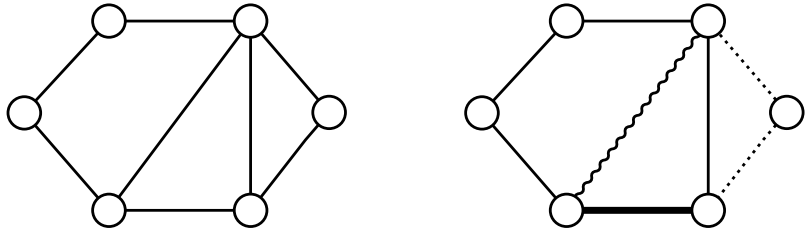

The good thing about these constructions is that we can certify them, by describing and certifying every step. 
This requires some care, as when certifying a new path, we could increase the size of the certificates of the endpoints, that are already in the graph. Fortunately, the tools developed to avoid certificate congestions allow us to control the certificate size.
 
The details about the connectivity certification can be found in Section~\ref{sec:connectivity}.

\paragraph*{Putting things together}
Combining these techniques, we can prove the following theorem.

\begin{restatable}{theorem}{ThmTwoConnected}
\label{thm:2-connected-to-general}
For any 2-connected graph $H$, if the 2-connected $H$-minor-free graphs can be certified with $f(n)$ bits, then the $H$-minor-free graphs can be certified with $O(f(n)+\log n)$ bits.
\end{restatable}

Going back to our example, $K_4$-minor-free graphs, given Theorem~\ref{thm:2-connected-to-general}, we are left with certifying the 2-connected $K_4$-minor-free graphs. 
As said above, these have a specific shape. 
More precisely, 2-connected $K_4$-minor-free graphs have a nested ear decomposition, which is yet another type of ear decomposition, this time with additional constraints related to outerplanarity. 
We can certify this structure by adapting a construction from \cite{FeuilloleyFMRRT20} for outerplanar graphs. 

More generally the 2-connected graphs corresponding to most of the classes of Figure~\ref{fig:our-results-certification} have specific shapes that we can certify quite easily, which imply our compact certification schemes. 
We do this in Section~\ref{sec:diamond-etc}.
A special case is $K_{2,4}$, that has a more complicated structure, requiring to consider 3-connected components, and some more complicated substructures. We study this case in Section~\ref{sec:K24}.

Finally, in Section~\ref{sec:4vertices}, we study all the minors on at most 4 vertices, and in Section~\ref{sec:5vertices} all the minors on 5 vertices of some simple form. 
For these, we do not need new techniques on the certification side, but we need to work on the graph theory side to establish new characterizations, as for these minors the literature does not help. 
The work we do in Section~\ref{sec:5vertices} might be of independent interest as we study the natural notion of $H$-minimal graph, which are the graph that have $H$ as a minor, but for which any vertex deletion would remove this property. 

\paragraph*{Lower bounds}

Towards the end of the paper, we show that $\Omega(\log n)$-bit labels are necessary to certify (2-connected) minor-free graph classes. 
When it comes to $\Omega(\log n)$ lower bounds in our model, there are basically two complementary techniques (called \emph{cut-and-plug techniques} in \cite{Feuilloley19}). 
Both techniques basically show that paths cannot be differentiated from cycles, if the certificates use $o(\log n)$ bits.
First, in~\cite{GoosS16}, the idea is to use many correct path instances, and to prove that we can plug them into an incorrect cycle instance, thanks to a combinatorial result from extremal graph theory. 
Second, in~\cite{FeuilloleyH18}, the idea is to consider a path, to cut it into small pieces, and to show via Sterling formula, that there exists a shuffle of these pieces that can be closed into a cycle.

Previous lower bounds for minor-free graphs in \cite{FeuilloleyFMRRT20} followed the same kind of strategies as \cite{GoosS16} and \cite{FeuilloleyH18}, with the same type of counting arguments, more complicated constructions, and tackled only minors that were cliques or bicliques. 

In this paper, we are able to do a black-box reduction between the path/cycle problem and the $H$-minor-freeness for any 2-connected $H$. This way we avoid explicit counting arguments, and get a more general result with a simpler proof.

\subsection{Related work}

Local certification first appeared under the name of \emph{proof-labeling schemes} in \cite{KormanKP10}, inspired by works on self-stabilizing algorithms (see \cite{Dolev2000} for a book on self-stabilization).
It has then been generalized under the name of \emph{locally checkable proofs} in \cite{GoosS16}, and the field has been very active since these seminal papers. 
In the following, we will focus on the papers about local certification of graph classes, but we refer to \cite{Feuilloley19} and \cite{FeuilloleyF16} for an introduction and a survey of local certification in general.

As said earlier, certification was first mostly about checking that the solution to an algorithmic problem was correct, a typical example being the verification of a spanning tree~\cite{KormanKP10}. 
Some graph properties have also been studied, for example symmetry in \cite{GoosS16}, or bounded diameter in \cite{Censor-HillelPP20}. 
Very recently, classes that are more central in graph theory have attracted attention. 
It was first proved in~\cite{NaorPY20}, as an application of a more general method, that planar graphs can be certified with $O(\log n)$ bits in the more general model of distributed interactive proofs. 
Then it was proved in \cite{FeuilloleyFMRRT20} that these graphs can actually be certified with $O(\log n)$ bits in the classic model, that is, without interaction. This result was extended to bounded-genus graphs in \cite{FeuilloleyFMRRT21}.
Later, \cite{EsperetL21} provided a simpler proof of both results via different techniques. It was also proved in \cite{MontealegreRR20} that cographs and distance-hereditary graphs have compact distributed  interactive proofs. 

Still in distributed computing, but outside local certification, the networks with some forbidden structures have attracted a lot of attention recently. 
A popular topic is the distributed detection of some subgraph $H$, which consists, in the CONGEST (or CONGEST-CLIQUE) model to decide whether the graph contains $H$ as a subgraph or not (see \cite{Censor-HillelFG20} and the references therein).
A related task is $H$-freeness testing, which is the similar but easier task consisting in deciding whether the graph is $H$-free or far from being $H$-free (in terms of the number of edges to modify to get a $H$-free graph). This line of work was formalized by \cite{Censor-HillelFS19} after the seminal work of \cite{BrakerskiP11} (see \cite{FraigniaudO19} and the references therein).
To our knowledge, no detection/testing algorithm or lower bounds have been designed for $H$-minor-freeness.

Finally, we have mentioned in the introduction that certifying that the graph belongs to some given class is important because some algorithms are specially designed to work on some specific classes. 
For example, there is a large and growing literature on approximation algorithms for \emph{e.g.} planar, bounded-genus, minor-free graphs. We refer to \cite{Feuilloley20} for a bibliography of this area.
There are also interesting works for exact problems in the CONGEST model, \emph{e.g.} in planar graphs~\cite{GhaffariH16}, graphs of bounded treewidth or genus~\cite{HaeuplerIZ16} and minor-free graphs~\cite{HaeuplerLZ18}. In particular the authors of~\cite{HaeuplerLZ18} justify the focus on minor-free graphs by the fact that this class allows for significantly better results than general graphs, while being large enough to capture many interesting networks. 
Very recently, \cite{GhaffariH21} proved general tight results on low-congestion short-cuts (an essential tool for algorithms in the CONGEST model) for graphs excluding a dense minor.

\section{Preliminaries}
\label{sec:preliminaries}

In this section, we define formally the notions we use and describe some useful known certification building blocks.

\subsection{Graphs and minors}

Let $G=(V,E)$ be a graph. 
Let $X \subseteq V$. The \emph{subgraph of $G$ induced by $X$} is the graph with vertex set $X$ and edge set $E \cap X^2$. The graph $G \setminus X$ is the subgraph of $G$ induced by $V \setminus X$. A graph $G'$ is a \emph{subgraph} of $G$ if $V' \subseteq V$ and $E' \subseteq E$.
For every $v \in V$, $N(v)$ denotes the \emph{neighborhood of $v$} that is the set of vertices adjacent to $v$. The graph $G$ is \emph{$d$-degenerate} if there exists an ordering $v_1,\ldots,v_n$ of the vertices such that, for every $i$, $N(v_i) \cap \{ v_{i+1},\ldots,v_n \}$ has size at most $d$. It refines the notion of maximum degree since any graph of maximum degree $\Delta$ are indeed $\Delta$-degenerate (but the gap between $\Delta$ and the degeneracy can be arbitrarily large). 
Let $u,v \in V$, a \emph{path} from $u$ to $v$ is a sequence of vertices $v_0=u,v_1,\ldots,v_\ell=v$ such that for every $i \le \ell-1$, $v_iv_{i+1}$ is an edge. It is a \emph{cycle} if $v_\ell v_0$ also exists.

A graph $G$ is \emph{connected} if there exists a path from $u$ to $v$ for every pair $u,v \in V$. All along the paper, we only consider connected graphs. Indeed, in certification, the nodes can only communicate with their neighbors, so no node can communicate with nodes of another connected component. 

A vertex $v$ is a \emph{cut-vertex} if $G \setminus \{ v \}$ is not connected. If $G$ does not contain any cut-vertex, $G$ is \emph{2-(vertex)-connected}. If the removal of any edge does not disconnect the graph, we say that $G$ is \emph{2-edge-connected}. A graph is \emph{$k$-(vertex)-connected} if there does not exist any set $X$ of size $k-1$ such that $G \setminus X$ is not connected. To avoid cumbersome notations, we will simply write $k$-connected for $k$-vertex-connected.

A graph $H$ is a \emph{minor of $G$} if $H$ can be obtained from $G$ by deleting vertices, deleting edges and contracting edges. Equivalently, it means that, if $G$ is connected, there exists a partition of $V$ into connected sets $V_1,\ldots,V_{|H|}$ such that there is (at least) an edge between $V_i$ and $V_j$ if $h_ih_j$ is an edge of $H$. We say that $V_1,\ldots,V_{|H|}$ is a \emph{model of $H$}. The graph $G$ is $H$-minor-free if it does not contain $H$ as a minor.

\subsection{Local computation and certification}
We assume that the graph is equipped with unique identifiers in polynomial range $[1, n^k]$, thus these identifiers can be encoded on $O(\log n)$ bits.

Local certification is a mechanism for verifying properties of labeled or unlabeled graphs. In this paper we will use a local certification at distance 1, which is basically the model called \emph{proof-labeling  schemes} \cite{KormanKP10}. 
A convenient way to describe a local certification is with a prover and a verifier. 
The \emph{prover} is an external entity that assigns to every node $v$ a certificate $c(v)$.  The \emph{verifier} is a distributed algorithm, in which every node $v$ acts as follows: $v$ collects the identifiers and the certificates of its neighbor and itself, and outputs a decision \emph{accept} or \emph{reject}. 
A local certification certifies a graph class $\mathcal{C}$ if the following two conditions are verified:
\begin{enumerate}
    \item For every graph of $\mathcal{C}$, the prover can find a certificate assignment such that the verifier accepts, that is, all nodes output \emph{accept}.
    \item For every graph not in $\mathcal{C}$, there is no certificate assignment that makes the verifier accept, that is for every assignment, there is at least one node that rejects.
\end{enumerate}
The size of the certificate of $\mathcal{C}$ is the largest size of a certificate assigned to a node of a graph of $\mathcal{C}$. 

Note that to describe a local certification, the only essential part is the verifier algorithm, the prover is just a way to facilitate the description of a scheme.

In this paper, we are going to use a variant of the model above, called \emph{edge certification}, where the certificates can be assigned on both the nodes and the edges. See Subsection~\ref{subsec:edge-certification}.

\subsection{Known building blocks for graph certification}
\label{subsec:building-blocks}

There are few known certification schemes that we are going to use intensively as building blocks in the paper. 
\begin{lemma}[\cite{KormanKP10,AfekKY90}]
\label{lem:acyclicity}
Acyclicity can be certified in $O(\log n)$ bits.
\end{lemma}

The classic way to certify that the graph is acyclic, is for the prover to choose a root node, and then to give to every node as its certificate its distance to the root. The nodes can simply check that the distances are consistent. 

The same idea can be used to certify a \emph{spanning tree} of the graph, encoded locally at each node by the pointer to its parent, which is simply the ID of this parent. 
The scheme is the same, except that the prover, in addition to the distances, gives the ID of the root, and the verification algorithm  checks that all nodes have been given the same root-ID, and only takes into account the edges that correspond to pointers (also the root checks that its ID is the root-ID).
A spanning tree is a very useful tool to broadcast the \emph{existence of a vertex satisfying a locally checkable property}: simply choose a spanning tree rooted at the special vertex, encode it locally with pointers and certify it. Then the root can check that indeed it has the right property, and all the other vertices know that such a vertex exists.

Finally, with the same ideas, one can easily deduce $O(\log n)$ certification for paths.
We just add to the acyclicity scheme the verification that the degree of every node is at most 2. 
Note that cycles do not need certificates to be verified: every node just checks that it has degree exactly 2.

Let us now define a graph class that will appear in several decompositions.

\begin{definition}
A \emph{path-outerplanar} graph is a graph that admits a  path $P$ that can be drawn on a horizontal line, such that all the edges that do not belong to $P$ can be drawn above that line without crossings. The edges are said to be \emph{nested}.
\end{definition}

We are going to use the following result as a black box.

\begin{lemma}[\cite{FeuilloleyFMRRT20}]\label{lem:path-outerplanar}
Path-outerplanar graphs can be certified with $O(\log n)$-bit certificates.
\end{lemma}

The following classic result will also be useful at some point of the paper.  

\begin{lemma}[\cite{KormanKP10}]
\label{lem:universal}
Every graph class can be certified with $O(n^2)$ bits.
\end{lemma}

The idea of the scheme is that the prover gives to every node $v$ the map of the graph, \emph{e.g.} as an adjacency matrix, along with the position of $v$ in this map. Then every node can check that it has been given the same map as its neighbors, and that the map is consistent with its neighborhood in the network.

\section{Avoiding certificate congestion}
\label{sec:avoiding-congestion}

One can obtain many structured graph classes like minor free graphs with "gluing" operations, for instance, by identifying vertices of two graphs of the class. If we have a certification for both graphs, we would like to simply take both certificate assignments to certify the new graph. However, for the vertex on which the two graphs are glued, the size of the certificate might have doubled. While it is not a problem for bounded degree graphs, it can become problematic if many gluing operations occur around the same vertex, since this vertex would get an additional certificate from each operation. In this section, we present two ways to tackle these issues, that will be used in the forthcoming sections. 

The first one consists in shifting the certification on edges instead of vertices, which helps in the sense that when gluing on vertices the edge certificate can remain unchanged. 
As we will see, the edge setting is equivalent to the usual vertex certification for nice enough classes. The second option uses that one can (almost) freely assume that a given vertex has an empty label in a correct certification.

\subsection{Edge certification and degeneracy}
\label{subsec:edge-certification}
Transforming a node certification into an edge certification can always be done without additional asymptotic costs: just copy on every edge the certificate of the two endpoints, and adapt the verification algorithm accordingly.
Transforming an edge certification into a node certification is also always possible, by giving a copy of the edge label to each of its endpoint. 
But this transformation can drastically increase the certificate size: if an edge certification uses $\Omega(f(n))$-bit labels, the associated node certification might use $\Omega(n \cdot f(n))$-bit  if the maximum degree of the graph is linear.
The following theorem ensures that in degenerate graph classes there is a more efficient transformation that permits to drastically reduce the size of the certificate.

\begin{theorem}[\cite{FeuilloleyFMRRT21}]\label{thm:degeneracy}
Consider an edge certification of a graph class $\mathcal{C}$ where the edges are labeled with $f(n)$-bit certificates. 
If $\mathcal{C}$ is $d$-degenerate, then there exists a (node) certification with $d\cdot f(n)$-bit certificates.
\end{theorem}

Note that $H$-minor free graphs have degeneracy $O(h \sqrt{\log h})$ where $h=|V(H)|$~\cite{kostochka1982minimum,thomason84}. Therefore, we can freely put labels on edges when certifying classes defined by forbidden minors.

\subsection{Certification with one empty label}

In this part, our goal is to erase the certificate of a node. To this end, we first consider certification of spanning trees and strengthen both Lemma~\ref{lem:acyclicity} and the discussion that followed in Subsection~\ref{subsec:building-blocks}. We then extend this intermediate step to every graph class in Lemma~\ref{lem:pointed}.

\begin{lemma}\label{lem:ST-with-no-root-label}
Let $T$ be a spanning tree of $G$. There exists a certification of $T$ that does not assign a label to the root, and uses the same certificate as the classic tree certification (cf. Subsection~\ref{subsec:building-blocks}) on the other nodes.
\end{lemma} 

\begin{proof}
On \emph{yes}-instances, the prover assigns the labels as in the classic scheme, and removes the label of the root.
Then the verification proceeds like in the classic scheme except for a node that has no label or a node that has a neighbor with no label. 
If two adjacent nodes have been given an empty label, then they reject. 
If a node with no label sees that two of its neighbors have been given different root-ID, then it rejects.
Otherwise, every node simulates the computation where the node with empty label has been given distance 0, and the same root-ID as its neighbors. 
Because of the previous checks, the labels used in the simulation are consistent, and on correct instance are the same as the one used in the classic certification.
Thus, the correctness follows from the correctness of the classic scheme.
\end{proof}

A \emph{pointed graph} is a graph with one selected node. Given a class, one can build its pointed version by taking for each graph all the pointed versions of it.

\begin{lemma}\label{lem:pointed}
Consider a class $\mathcal{C}$ that can be certified with certificates of size $f(n)$. 
One can certify the pointed class of $\mathcal{C}$ with  $O(f(n)+\log n)$ bits, without having to put certificates on the selected node. 
\end{lemma}

\begin{proof}
First, to certify that exactly one node is pointed, we can simply find a spanning tree rooted on the pointed vertex and assign to each node the spanning tree certification of Lemma~\ref{lem:ST-with-no-root-label} which uses $O(\log n)$ bits.
For the rest of the certification, on a \emph{yes}-instance, the prover first assigns the certificates following the original certification. 
Then it removes the certificate of the selected node and appends copies of it to the certificates of its neighbors. 

Every node $v$ runs the following verification. If $v$ is not the selected node, nor one of its neighbors, then it does the same verification as before. 
If $v$ is the selected node, it checks that its neighbors have been given the same label as "label of the selected node", and then takes this label as its own, and runs the previous verification algorithm. 
If $v$ is a neighbor of the selected node, it runs the same verification algorithm as before, but simulating that the selected node has been given the certificate that was appended to its own certificate.

All nodes are simulating the computation in the graph where the selected node would have been given its certificates, thus the correctness of this new certification follows from the correctness of the original certification.
\end{proof}

Observe that the previous results can be easily iterated: one can always remove the labels of $k$ nodes (as long as they are pairwise non-adjacent) to the cost of a factor $k$ in the size of the certificates. Therefore, the result extends to the case of $k$-independent pointed classes (i.e. where an independent set of size at most $k$ is selected instead of only one vertex).

\begin{corollary}
\label{coro:pointed}
Consider a class that can be certified with certificates of size $f(n)$. 
One can certify the $k$-independent pointed class with  $O(kf(n)+k\log n)$ bits, without having to put certificates on the selected nodes. 
\end{corollary}

Moreover, with more constraints on the structure of the set of pointed vertices (for instance if they are all at distance at least 3), one could even obtain certificate of size $O(f(n)+k\log n)$ (since every node receives the certificate of at most one selected node).

\section{Compositions of certifications}
\label{sec:composition}

In this section, we show how to combine certification algorithms for several classes to certify larger ones, and we illustrate this idea on two constructions. The first one considers classes defined by the existence of some subgraph: we settle the intuition stating that it is often easier to test the existence of a structure rather than its absence, since we can pinpoint which nodes/edges lie in the structure. 

The second construction mimics a natural operation on graphs, consisting in replacing some vertex/edge by another graph. This operation occurs quite often in the literature: many classes, especially the ones defined by forbidden minors, get a characterization using this operation.

Some results of this section will not be used to certify minor-free classes later in the paper. 
They are proved here for completeness.

\subsection{Subgraphs}

\begin{proposition}\label{prop:subgraph}
Let $\mathcal{C}$ be a graph class that can be certified with $f(n)$-bit labels.
Let $\mathcal{C}'$ be the class of the graphs that contain a graph of $\mathcal{C}$ as subgraph. 
Then $\mathcal{C}'$ can be certified with certificates of size $O(f(n)+\log n)$ on the nodes and $O(1)$ on the edges.
\end{proposition}

\begin{proof}
On a \emph{yes}-instance $G$, the prover assigns the certificates on nodes and edges in the following way. 
First, it chooses a subgraph $H$ that belongs to $\mathcal{C}$ and assigns the certificates that certify that $H$ is in $\mathcal{C}$, as if the rest of the graph did not exist. This takes at most $f(n)$ bits.
Second to every node and edge that belongs to $H$, the prover assigns a special label.
Third, the prover describes and certifies a spanning tree pointing to a node that has the special label.

The verification algorithm is the following. 
The nodes that have the special label, run the verification algorithm for $\mathcal{C}$, taking into account only the nodes and edges that have the special label. 
The nodes also check the spanning tree structure, and the root of the tree checks that it does have the special label.
    
Because of the spanning tree, there must exist a node with the special label, thus there are nodes that run the verification algorithm for $\mathcal{C}$, and if they succeed it means that a graph of $\mathcal{C}$ appears as a subgraph in the graph $G$.    
\end{proof}

The edge certificates in Proposition~\ref{prop:subgraph} can be inconvenient if we want a classic certification (without edge certificates) and if the graph is not assumed to be degenerate, which prevents us from using Theorem~\ref{thm:degeneracy}. However, observe that we give non-empty certificates only to the edges of the subgraph, hence we can obtain a vertex-certification when the class $\mathcal{C}$ is degenerate. 

\begin{corollary}
\label{cor:subgraph_deg}
Let $\mathcal{C}$ be a $d$-degenerate graph class that can be certified with $f(n)$-bit labels.
Let $\mathcal{C}'$ be the class of the graphs that contain a graph of $\mathcal{C}$ as subgraph. 
Then $\mathcal{C}'$ can be certified with certificates of size $O(f(n)+d\log n)$ on the nodes.
\end{corollary}

Observe also that when considering induced subgraphs, we only have to specify which vertices are special, hence we do not need edge certificates either. Note that, since we do not need to label edges, we do not need the class $\mathcal{C}$ to be degenerate.

\begin{corollary}
\label{cor:subgraph_induced}
Let $\mathcal{C}$ be a graph class that can be certified with $f(n)$-bit labels. 
Let $\mathcal{C}'$ be the class of the graphs that contain a graph of $\mathcal{C}$ as an induced subgraph. 
Then $\mathcal{C}'$ can be certified with certificates of size $O(f(n)+\log n)$ on the nodes.
\end{corollary}

\subsection{Expansions}

Two common operations in characterizations of graph classes are what we call node and edge expansions.

\begin{definition}
Consider two graph classes $\mathcal{C}_1$ and $\mathcal{C}_2$.
\begin{itemize}
    \item The \emph{node expansion of $\mathcal{C}_1$ by $\mathcal{C}_2$} is the class of graphs obtained by the following operation. Take a graph $G$ in $\mathcal{C}_1$ and replace every node $v$ by a graph $H(v)$ in $\mathcal{C}_2$, in such a way that for every edge $uv\in E(G)$, there is (at least) one edge between $H(u)$ and $H(v)$ in $G$ (and no such edge if $uv\notin E(G)$). 
    \item The \emph{edge expansion of $\mathcal{C}_1$ by $\mathcal{C}_2$} is the class of graphs obtained by the following operation. Take a graph $G$ in $\mathcal{C}_1$ and replace every edge $uv$ by a graph $H(u,v)$ from $\mathcal{C}_2$, in such a way that the nodes of the original graph that are contained in $H(u,v)$ are exactly $u$ and $v$.
\end{itemize}
\end{definition}

We would like to have results of the form: if $\mathcal{C}_1$ and $\mathcal{C}_2$ can be certified with $f(n)$ and $g(n)$-bit labels respectively, then the expansion can be certified with $O(f(n)+g(n))$-bit labels. 
While the natural approach (almost) works for edge-expansion, it does not give such a result for node-expansion. However, we can actually make it work with a bound that takes into account the maximum degree of the expanded graph.

\begin{proposition}
\label{prop:node-expansion}
Consider two graph classes $\mathcal{C}_1$ and $\mathcal{C}_2$ that can be certified with $f(n)$-bit and $g(n)$-bit labels respectively, where all the graphs of $\mathcal{C}_1$ have maximum degree $\Delta$. Then the node-expansion of $\mathcal{C}_1$ by $\mathcal{C}_2$ can be certified with $O(\Delta \cdot f(n)+g(n)+\Delta\log n)$-bit certificates.
\end{proposition}

\begin{proof}
Consider a graph $G\in \mathcal{C}_1$ on $\ell$ nodes $v_1,\ldots,v_\ell$ of maximum degree $\Delta$, and let $H_1,\ldots,H_\ell$ be graphs of $\mathcal{C}_2$. We consider the node expansion of $G$ where every $v_i$ is replaced by $H_i$.

On a \emph{yes}-instance, the prover assigns the certificates the following way. 
First it assigns to every node the index $i$ corresponding to the graph $H_i$ it belongs to and the certification of the fact that $H_i$ belongs to $\mathcal{C}_2$ (without taking into accounts the other nodes and edges).
This takes at most $g(n)+\log n$ bits per node.
Second, the prover gives to each vertex of $H_i$ the original certificate of $v_i$ that $G$ belongs to $\mathcal{C}_1$ as well as the original certificate of all the vertices in $N(v_i)$ in $G$ together with their names, which takes $O(\Delta f(n))$ bits. Finally, for every $v_j\in N(v_i)$, the prover chooses a vertex $w_j$ in $H_i$ adjacent to a vertex in $H_j$, and certifies a spanning tree of $H_i$ rooted at $w_j$. This takes $O(\Delta\log n)$ bits.
    
The verification algorithm is the following. 
Every node (labeled as) in $H_i$ checks that the number of trees corresponds to the degree of $v_i$ in $G$.
Every node checks the correctness of the different trees. Moreover, every root $v$ of a spanning tree in $H_i$ checks that it has a neighbor in the corresponding $H_j$.
All the nodes of $H_i$ check that their neighbors are in $H_i$ or in some $H_j$ with $v_j$ incident to $v_i$ in $G$.
Every node of each $H_i$ runs the verification algorithm to check that $H_i$ does belong to $\mathcal{C}_2$.
Finally, every node of $H_i$ simulates the verification of the original node $v_i$, which is possible since every vertex of $H_i$ receives the certificate of $v_i$ and all its neighbors in $G$. And every vertex $w_i \in H_i$ incident to $w_j \in H_j$ checks that $v_j\in N_G(v_i)$ and that the certificate of $w_j$ indeed contains the certificates of $v_i$ and $v_j$ given for $G$. 
\end{proof}

\begin{proposition}
\label{prop:edge-expansion}
Consider two graph classes $\mathcal{C}_1$ and $\mathcal{C}_2$ that can be certified with $f(n)$-bit and $g(n)$-bit labels respectively. 
Then the edge-expansion of $\mathcal{C}_1$ by $\mathcal{C}_2$ can be certified with $O(f(n)+g(n)+\log n)$-bit certificates on the edges. 
\end{proposition}

\begin{proof}
We use a similar reasoning as for the proof of Proposition~\ref{prop:node-expansion}, except that we first transform the node certifications of $\mathcal{C}_1$ and $\mathcal{C}_2$ into edge certifications (by putting the label of a node on all the edges incident with it). 

Consider a graph $G\in \mathcal{C}_1$. We consider the edge expansion of $G$ where every $uv$ is replaced by $H(u,v)$. Each edge $e$ from $H(u,v)$ receives the labels of $u$ and $v$, the certificate of $uv$ in $G$ for $\mathcal{C}_1$, and the certificate of $e$ in $H(u,v)$ for $\mathcal{C}_2$. Therefore, the certificates have size $O(f(n)+g(n)+\log n)$.

Now each vertex can check that all the edges labeled in some $H(u,v)$ share the same certificate for $uv$. There are two kinds of nodes: some where all incident edges are labeled as in the same $H(u,v)$, and the others (the original vertices of $G$). All of them run the verification algorithm for $\mathcal{C}_2$ by considering each group of incident edges labeled as in the same $H(u,v)$. The latter also recover the certificates of  their neighbors in $G$ from the edge labeling, and run the verification algorithm for $\mathcal{C}_1$. 
\end{proof}

Before giving deeper applications of these results in future sections, let us prove that the \emph{existence} of a minor in the graph is easy to certify. 
This was already mentioned in previous papers without formal proofs \cite{FeuilloleyFMRRT20, FeuilloleyFMRRT21}. 
We prove it here to show a simple application of our techniques, and we think it is a meaningful illustration of the fact that certifying that a structure is present or absent are two very different tasks in our model.

\begin{corollary}
Given a graph $H$, one can certify that a graph has $H$ as a minor in $O(\log n)$ bits.
\end{corollary}

\begin{proof}
As we already observed, a graph $G$ has $H$ as minor if and only if $V(G)$ can be partitioned into $|H|$ connected sets such that there is an edge between $V_i$ and $V_j$ when the corresponding vertices in $H$ are connected. Free to delete edges, we can assume that each $V_i$ is actually a spanning tree and there is a unique edge from $V_i$ to $V_j$ if and only if the corresponding vertices are connected in $H$.
In other words, $G$ has a subgraph that is a node expansion of $H$ by trees. Moreover, we can choose such a subgraph with degree at most $|H|-1=O(1)$ since $H$ is fixed.

Let us start from a certification of $H$ and build a certification of $G$. 
The structure of $H$ can be certified in a brute-force way, by providing to every node the complete map of the graph which takes constant space (since $H$ is fixed). 
Then, since trees can be certified in $O(\log n)$ bits, thanks to Proposition~\ref{prop:node-expansion}, any node-expansion of $H$ by trees can be certified with $O(\log n)$ certificates.

We finally get a node certification with certificates of size $O(\log n)$ using Corollary~\ref{cor:subgraph_deg}.
\end{proof}

\section{Connectivity and connectivity decompositions}
\label{sec:connectivity}

In this section, we explain how to certify connectivity properties and connectivity decompositions, in particular the block-cut tree mentioned in the introduction. 

An \emph{ear decomposition} is a way to build a graph by iteratively adding paths, the so-called \emph{ears}. 
Ear decompositions are central tools for decades in structural graph theory and are used in many decomposition or algorithmic results.
There exists various variants of this process, that characterize different classes and properties. 
For certification, these decompositions happen to be easier to manipulate than some other types of characterizations since they are based on iterative construction of the graph, and use paths, which are easy to certify. 
These paths are convenient since we can "propagate" some quantity of information on them as long as every vertex belongs to a bounded number of paths. 
In this section, we remind several such decompositions, and use them to certify various connectivity properties and decompositions.

\subsection{Connectivity properties}

Let us start with 2-connectivity. A graph $G$ has an \emph{open ear decomposition} if $G$ can be built, by starting from a single edge, and iteratively applying the following process: take two different nodes of the current graph and link them by a path whose internal nodes are new nodes of the graph (such a path is called \emph{an ear}). Note that this path can be a single edge, and then there is no new node.
Let an \emph{inner node} of an ear be a vertex that is created with this ear, and let a \emph{long ear} be an ear with at least one inner node.

\begin{theorem}[\cite{Whitney32} (reformulated)]
A graph is 2-connected if and only if it has an open ear decomposition.
\end{theorem}

We use this characterization to certify 2-connectivity.
\begin{lemma}\label{lem:2-connected-graphs}
2-connected graphs can be certified with $O(\log n)$ bits.
\end{lemma}

\begin{proof}
First observe that one can obtain a long-ear decomposition from an ear decomposition (and vice versa) by removing/adding short ears, i.e. edges. Therefore, having an open ear decomposition is equivalent to having a subgraph with an open long-ear decomposition. Note that if a graph $G$ has an open long-ear decomposition, then it is $2$-degenerate. Indeed, the vertices of the last added long-ear have degree two and their removal is still a graph with an open long-ear decomposition. So in order to get the conclusion, Corollary~\ref{cor:subgraph_deg} ensures that we only have to certify open long-ear decomposition with $O(\log n)$ bits per vertex.

The certification works as follows. 
First the prover gives to every node the identifiers of the very first edge, and describes and certifies a spanning tree pointing to one of the endpoints of this edge.
The nodes of this edge are given an \emph{index} 0.
Second, the prover gives to every node the information related to the step when it has been added, and only about this step. 
That is, the prover gives the \emph{index} of the addition (that is the number of the ear in which the vertex is created), along with two oriented paths spanning the path and pointing to the two extremities of the ear. 
By Corollary~\ref{coro:pointed}, these paths can be certified without certificates on the extremities of the paths.

Every node checks the correctness of the spanning tree pointing to the first edge, and the fact that only these nodes have index 0.
Then, every node also checks that the spanning paths it has been given are correct, that is: (1) the distances and root-ID are consistent (2) all nodes have the same index, and that (3) the declared endpoints are different.
Also, the two nodes that are adjacent to the endpoints of the paths check that the endpoints have a smaller index than their own.

Let us now prove the correctness of the certification.
Because of the spanning tree, the original edge exists, is unique, and is the only set of nodes with index 0. Because of the certified paths spanning the ears, one can also recover the path structure and the fact that a path is added after its endpoints. 
Note that in an instance where all nodes accept, there might be two different paths with the same index $i$, but this is not a problem: the only important feature is the precedence order.
\end{proof}

With similar construction we can certify the edge connectivity instead of the vertex connectivity.

\begin{corollary}\label{coro:2-edge-connected}
2-edge-connected graphs can be certified with $O(\log n)$ bits.
\end{corollary}

\begin{proof}
Robbins proved in~\cite{Robbins39} that a graph $G$ is 2-edge connected if and only if $G$ has an ear decomposition. 
An ear decomposition is the same as an open ear decomposition, except that it starts from a cycle and that the two endpoints of an ear do not need to be different.
The proof above can thus be adapted to this class.

The only difference is that vertices with index $0$ form a cycle (which can be certified). 
Then during the verification procedure we simply do not have to check that the extremities of the path of the ear decomposition are distinct, in other words we do not have to check (3).
\end{proof}

A more refined type of ear decomposition characterizes the 3-vertex-connected graphs.

\begin{definition}[\cite{Schmidt16, Mondshein71, CheriyanM88}]
Let $ru$ and $rt$ be two edges of a graph $G$. A \emph{Mondshein sequence through $rt$, avoiding $u$} is an open ear decomposition of $G$ such that:
\begin{enumerate}
    \item $rt$ is in the first ear.
    \item the ear that creates node $u$ is the last long ear, $u$ is its only inner vertex, and it does not contain $ru$.
    \item the ear decomposition is non-separating, that is, for every long ear except the last one, every inner node has a neighbor that is created in a later ear.
\end{enumerate}
\end{definition}

\begin{theorem}[\cite{CheriyanM88, Schmidt16}]
Let $ru$ and $rt$ be two edges of a graph $G$. 
The graph $G$ is 3-vertex-connected if and only if it has a Mondshein sequence through $rt$ avoiding $u$, and there are three internally vertex-disjoint path between $t$ and $u$.
\end{theorem}

\begin{corollary}\label{coro:3-connected}
3-connected graphs can be certified with $O(\log n)$ bits on vertices and $O(1)$ bits on edges.
\end{corollary}

\begin{proof}
On \emph{yes}-instances the prover chooses an arbitrary edge $ru$ and certifies the ear decomposition as in Lemma~\ref{lem:2-connected-graphs}.
The prover also adds a spanning tree pointing to the edges $ru$ and $rt$, and gives to every vertex the index of the last long ear created.
These new pieces of information allow the nodes to check that the ear decomposition is a Mondshein sequence.
The prover also encode the three vertex disjoint paths, by pointer on the nodes of these paths, and number them 1, 2 and 3, to allow the nodes to check disjointness. 
\end{proof}

\subsection{Block-cut tree}

Now that we can certify connectivity properties, we introduce a way to certify decomposition of graphs into parts of higher connectivity. Let us start with a few definitions.

A \emph{2-connected component} of a graph $G$ is a connected subgraph $H$ maximal by inclusion such that the removal of one node does not disconnect $H$. Observe that a 2-connected component can consists of just one edge in the case of a bridge, i.e. an edge whose removal disconnects the graph.

The intersection of any pair $C,C'$ of 2-connected components has size at most one. 
Indeed, if it had size at least two, then we could merge these into a larger 2-connected component, which would contradict the maximality. 
So we can define an auxiliary graph from $G$ where every node corresponds to a 2-connected component and there is an edge between two components if and only if they intersect on exactly one node. This graph is a tree, because a cycle would again create a larger 2-connected component, contradicting maximality. 
This tree is called the \emph{block-cut tree}. 

Let $T$ be  a block-cut tree of $G$, and $D$ a maximal 2-connected component chosen to be the root of this tree. (Note that if $G$ is $2$-connected then the graph is reduced to this component). Let $C$ be a component that is not the root of the tree.  The \emph{connecting node} of a component $C$ is the node lying both in $C$ and in its parent component. The \emph{interior} of $C$ is the set of nodes of $C$ minus the connecting node of $C$. Note that the interior of a component is always non-empty.

This section is devoted to proving the following result and apply it for certification:

\ThmTwoConnected*

\begin{proof}[Proof of Theorem~\ref{thm:2-connected-to-general}]
Since $H$ is 2-connected, a graph $G$ is $H$-minor-free if and only if each of its 2-connected components is. (This is basically the observation we made at the beginning of Subesction~\ref{subsec:our-techniques}.) 
This is the property we certify. On a \emph{yes}-instance, the prover will assign the certificates the following way. It first computes the block-cut tree and root it on some node $C$.  It then does the following:
\begin{enumerate}
    \item For each 2-connected component, the prover chooses a node from the interior of the component to be the \emph{leader} of this component.
    Every node of the interior of a component $C$ is given the identifier of the leader of $C$ as well as a spanning tree of $C$ pointing towards it. Since the component is $2$-connected, the component minus the leader of the component is connected and such a tree exists. 
    
    \item Every node is given a label stating whether it is a connecting node or not.

    \item Every node is given the identifier of the connecting node of its component closest from the root in the block-cut tree (called the \emph{component} of the node), as well as a spanning tree pointing to it, using the certification of Lemma~\ref{lem:ST-with-no-root-label} that uses an empty certificate on the root.
    
    \item In order to check acyclicity of the block-cut tree, every node is given the distance of its component to the root-component (in terms of number of components).
    
    \item The prover certifies the $2$-connectivity of each component using the certification of Lemma~\ref{lem:2-connected-graphs} and the fact that it is $H$-minor free using the certification with $f(n)$ bits of the theorem. By Lemma~\ref{lem:pointed} this can be done by only  assigning labels to the interior nodes of the component.
\end{enumerate}

Before we move on to the verification and the correctness of the scheme, note that every node is given a certificate of size $O(f(n)+\log n)$. Indeed, each piece of information we have given to the node is of size $O(\log n)$ or $f(n)$, and we have given a constant number of those to every node. 
In particular, a connecting node in the interior of a component $C$, received only labels that are related $C$, and not labels related to other components it belongs to (since we consider pointed components).

Now, every node does the following verification. 
Every node checks that the spanning tree pointing to the leader is correct. If this step succeeds, we have a partition of the nodes in components. Every node also checks the correctness of the spanning tree pointing to the connecting node.

Every node $v$ checks that, if it has an edge to a connecting node $w$ with a different leader, then $w$ is the connecting node of its own component. Every connecting node $v$ checks that it is connected to a single node in its parent component and that it is the claimed neighbor in that component.
If this step succeeds, we have a decomposition into components linked by connecting nodes. 
The consistency of the component distances are also checked by the nodes: this distance should be decremented at each connecting node, and only there. This ensures the acyclicity of the component structure.
Finally, every node checks that the $2$-connectivity and the $H$-minor-freeness of its component. 
Globally this verification ensures that the graph is $H$-minor-free.
\end{proof}

\section{Application to $C_4$, Diamond, $K_4$ and $K_{2,3}$ minor-free graphs}
\label{sec:diamond-etc}

This section is devoted to the certification of $C_4$-minor-free, diamond-minor-free graphs, $K_4$-minor-free graphs and $K_{2,3}$-minor-free graphs.
All the proofs will follow the same structure: prove that the 2-connected components, which are more structured, can be certified with small labels, and then use Theorem~\ref{thm:2-connected-to-general} to conclude for the general case. 

Before going to this proof let us describe how to certify \emph{series-parallel graphs}, which in addition to be  interesting network topologies~\cite{FlocchiniL03}, are closely related to $K_4$-minor-free graphs.

\begin{definition}
A \emph{(2-terminal) series-parallel graph} is a graph with two labeled vertices called the \emph{source} and the \emph{sink} that can be built recursively as follows. 
A single edge is a series-parallel graph where one endpoint is the source and the other is the sink.
Let $G_1,G_2$ be two series-parallel graphs. The \emph{series} of $G_1$ and $G_2$ which consists in merging the sink of $G_1$ and the source of $G_2$ is a series-parallel graph.
The \emph{parallel} of $G_1$ and $G_2$, which consists in merging the sources of $G_1$ and $G_2$ together and merging the sinks of $G_1$ and $G_2$ together, is a series-parallel graph.
\end{definition}

A \emph{nested ear decomposition} is an open ear decomposition that starts from a path, with two properties: (1) both ends of an ear have to be connected to the same ear, and (2) for every ear, the ears that are plugged onto it are nested. 
Eppstein proved the following in~\cite{Eppstein92} about series-parallel graphs.

\begin{theorem}[\cite{Eppstein92}]\label{thm:series-parallel-nested}
A 2-connected graph is series-parallel if and only if it has a \emph{nested ear decomposition}.
\end{theorem}

We will use this decomposition theorem for our certification.

\begin{theorem}\label{thm:2-connected-series-parallel}
2-connected series-parallel graphs can be certified with $O(\log n)$-bit labels.
\end{theorem}

\begin{proof}
The prover certifies the decomposition of Theorem~\ref{thm:series-parallel-nested}.
We have already described how to certify an open ear decomposition in the proof of Lemma~\ref{lem:2-connected-graphs}.
We can easily adapt it so that it starts from a path instead of an edge: there is a spanning tree pointing to one of the endpoints of the paths, and the path itself is certified with distances, the usual way.

It is also easy to certify that each ear $e$ has both of its endpoints on the same older ear $e'$: just give to each vertex of $e$ the identifiers of the endpoints of $e$ and $e'$.
The endpoints of an ear can check the consistency of these announced identifiers with the identifiers of their paths.
A more tricky part is to certify that the ears are nested. 
Remember that Lemma~\ref{lem:path-outerplanar} states that a path with nested edges (a path-outerplanar graph) can be certified with $O(\log n)$-bit labels. 
This is exactly what we need except that we would like to have nested paths instead of nested edges. 
But then we can transfer the information from one endpoint of the paths to the other endpoint. 
\end{proof}

\begin{lemma}\label{lem:C5minor}
$2$-connected $C_5$-minor free graphs are either graphs of size at most $4$ or $K_{2,p}$ or $K'_{2,p}$ which is the complete bipartite graph $K_{2,p}$ plus an edge between the two vertices on the set of size $2$.
\end{lemma}
\begin{proof}
Since $G$ is $2$-connected, by Menger's theorem, for every pair $x,y$ of vertices, there exist at least two vertex disjoint $xy$-paths. Since $G$ is $C_5$-minor free, these paths have size at most $2$, in particular $x,y$ are at distance at most 2. 

Let $u,v$ be two non-adjacent vertices. Then the removal of $N(u) \cap N(v)$ disconnects $u$ from $v$ since otherwise we can find two vertex disjoint $uv$-paths, one being of size at least $3$, which provides a $C_5$. In particular, it implies that $|N(u) \cap N(v)| \ge 2$ since $G$ is $2$-connected.

Let $x\in N(u)\setminus (N(v)\cup \{v\})$. Since $x,v$ are non-adjacent, there must be an edge between $x$ and $N(v)$. But this creates a $C_5$ since $|N(u)\cap N(v)|\geqslant 2$. Therefore non-adjacent vertices are twins. 

Let $I$ be a maximum independent set in $G$. Note that all the vertices of $I$ are twins. Therefore, by maximality, if $u\notin I$ then $u$ is complete to $I$. Now either vertices of $I$ have degree at least 3, and $G$ contains $K_{3,3}$ hence a $C_5$-minor, or vertices of $I$ have degree 2 and $G$ is $K_{2,p}$ or $K'_{2,p}$.
\end{proof}

We can now prove easily the claimed certifications.

\begin{corollary}\label{coro:diamond-etc}
The following classes of graphs can be certified with $O(\log n)$ bit certificates: $C_4$-minor-free graphs,$C_5$-minor free graphs, diamond-minor-free graphs, house-minor free graphs\footnote{The \emph{house} being a $C_4$ plus a vertex connected to two consecutive vertices of the $C_4$.}, outerplanar graphs (that is $(K_{2,3}, K_4)$-minor-free graphs), $K_{2,3}$-minor-free and $K_4$-minor-free graphs.
\end{corollary}

\begin{proof}
By Theorem~\ref{thm:2-connected-to-general}, if we can certify the 2-connected graphs of these classes we obtain the conclusion. So we simply have to prove that for each class we can certify the 2-connected graph of the class. 
\begin{itemize}
    \item 2-connected $C_4$-minor-free graphs are $K_2$ and $K_3$~\cite{chimani2019cut}, which can be certified with $O(1)$ bits.
    \item 2-connected $C_5$-minor-free graphs are either graphs of size at most $4$ or a complete bipartite graph $K_{2,p}$ (with a potential edge between the two vertices in the set of size $2$ by Lemma~\ref{lem:C5minor}. Since such graphs can be certified with $O(\log n)$ bits, the conclusion follows. 
    \item 2-connected diamond-minor-free graphs are induced cycles. Cycles can be certified with $O(1)$ bits (see the discussion after Lemma~\ref{lem:acyclicity}).
    \item 2-connected house-minor-free graphs are either induced cycles or graphs of size at least four. Indeed, assume that there is a cycle of length at least $5$. 
    Then it should be induced, since otherwise it contains a house as a minor. Moreover, it should contain all the vertices of the graph otherwise there is an ear starting from this cycle and the cycle plus the ear provides a house. Since induced cycles can be easily certified with $O(\log n)$ bits, the conclusion follows.
    \item 2-connected outerplanar graphs are exactly path-outerplanar graphs with an edge between the first and the last node. Indeed, by 2-connectivity, the outer face must be a cycle, and removing any edge from it yields a path-outerplanar graph. One can then certify the existence and uniqueness of this edge using a spanning tree, and then certify that the rest of the graph is path-outerplanar. This yields a $O(\log n)$-bit certification by Theorem~\ref{lem:path-outerplanar}.
    \item Let $G$ be a 2-connected $K_{2,3}$-minor-free graph. If $G$ does not contain $K_4$, then it is a 2-connected outerplanar graph and the result follows from the previous item. Otherwise, if $G$ is not restricted to $K_4$, then it contains a fifth vertex $u$. Since $G$ is connected, there is a shortest path from $u$ to the $K_4$ ending at $v$. Since $v$ is not a cut-vertex, there should be another path between $u$ and the $K_4$ avoiding $v$, but this creates a $K_{2,3}$ minor. Therefore, $G$ is $K_4$, which can be certified easily.
    \item The 2-connected $K_4$-minor-free graphs are exactly the 2-connected series-parallel graphs. Then the results follow directly from Theorem~\ref{thm:2-connected-series-parallel}.\qedhere
\end{itemize}
\end{proof}

\section{Application to $K_{2,4}$-minor free graphs}
\label{sec:K24}

When the size of the minors are increasing (and for most of the decomposition theorems known in structural graph theory), 2-connectivity is not enough. In this example we will illustrate how to use the certificate of $3$-connectivity to conclude.

Let us illustrate it for this section on the characterization of $K_{2,4}$-minor-free graphs from~\cite{EllinghamMOT16}. 
It is more involved than the other characterizations we have seen so far. 
We will follow the structure of \cite{EllinghamMOT16}, restricting first to 3-connected graphs, then to 2-connected graphs, and finally all $K_{2,4}$-minor-free graphs. 

\subsection{3-connected case}

Let us start with the definition of a graph class. We use notations similar to \cite{EllinghamMOT16} (Section~2.1). See Figure~\ref{fig:Gnrsp}.

\begin{definition}\label{def:Gnrsp}
Let $n$, $r$, $s$ be  three integers, and $p$ a Boolean. The graph $G_{n,r,s,p}$ consists of 
a path  $v_1,...,v_n$, 
the edges $v_1v_{n-i}$ for $1\leq i\leq r$, 
the edges $v_nv_{1+j}$ for $1\leq j\leq s$, 
and the edge $v_1v_n$ if $p=1$. 
For a function $f(n,r,s,p)$ that associate a Boolean with each combination of parameters, let $\mathcal{G}[f]$ be the set of graphs $G_{n,r,s,p}$ such that $f(n,r,s,p)=1$.
\end{definition}

\begin{figure}[!h]
    \centering
    \begin{tikzpicture}[x=0.75pt,y=0.75pt,yscale=-1,xscale=1]

\draw [line width=1.5]    (123.45,121) -- (339.45,121) ;
\draw  [fill={rgb, 255:red, 0; green, 0; blue, 0 }  ,fill opacity=1 ] (332.9,121) .. controls (332.9,117.38) and (335.83,114.45) .. (339.45,114.45) .. controls (343.07,114.45) and (346,117.38) .. (346,121) .. controls (346,124.62) and (343.07,127.55) .. (339.45,127.55) .. controls (335.83,127.55) and (332.9,124.62) .. (332.9,121) -- cycle ;
\draw  [fill={rgb, 255:red, 0; green, 0; blue, 0 }  ,fill opacity=1 ] (153,120.55) .. controls (153,116.93) and (155.93,114) .. (159.55,114) .. controls (163.17,114) and (166.1,116.93) .. (166.1,120.55) .. controls (166.1,124.17) and (163.17,127.1) .. (159.55,127.1) .. controls (155.93,127.1) and (153,124.17) .. (153,120.55) -- cycle ;
\draw  [fill={rgb, 255:red, 0; green, 0; blue, 0 }  ,fill opacity=1 ] (123.45,121) .. controls (123.45,117.38) and (126.38,114.45) .. (130,114.45) .. controls (133.62,114.45) and (136.55,117.38) .. (136.55,121) .. controls (136.55,124.62) and (133.62,127.55) .. (130,127.55) .. controls (126.38,127.55) and (123.45,124.62) .. (123.45,121) -- cycle ;
\draw  [fill={rgb, 255:red, 0; green, 0; blue, 0 }  ,fill opacity=1 ] (183,120.55) .. controls (183,116.93) and (185.93,114) .. (189.55,114) .. controls (193.17,114) and (196.1,116.93) .. (196.1,120.55) .. controls (196.1,124.17) and (193.17,127.1) .. (189.55,127.1) .. controls (185.93,127.1) and (183,124.17) .. (183,120.55) -- cycle ;
\draw  [fill={rgb, 255:red, 0; green, 0; blue, 0 }  ,fill opacity=1 ] (214,120.55) .. controls (214,116.93) and (216.93,114) .. (220.55,114) .. controls (224.17,114) and (227.1,116.93) .. (227.1,120.55) .. controls (227.1,124.17) and (224.17,127.1) .. (220.55,127.1) .. controls (216.93,127.1) and (214,124.17) .. (214,120.55) -- cycle ;
\draw  [fill={rgb, 255:red, 0; green, 0; blue, 0 }  ,fill opacity=1 ] (244,121.55) .. controls (244,117.93) and (246.93,115) .. (250.55,115) .. controls (254.17,115) and (257.1,117.93) .. (257.1,121.55) .. controls (257.1,125.17) and (254.17,128.1) .. (250.55,128.1) .. controls (246.93,128.1) and (244,125.17) .. (244,121.55) -- cycle ;
\draw  [fill={rgb, 255:red, 0; green, 0; blue, 0 }  ,fill opacity=1 ] (273,121.55) .. controls (273,117.93) and (275.93,115) .. (279.55,115) .. controls (283.17,115) and (286.1,117.93) .. (286.1,121.55) .. controls (286.1,125.17) and (283.17,128.1) .. (279.55,128.1) .. controls (275.93,128.1) and (273,125.17) .. (273,121.55) -- cycle ;
\draw  [fill={rgb, 255:red, 0; green, 0; blue, 0 }  ,fill opacity=1 ] (303,120.45) .. controls (303,116.83) and (305.93,113.9) .. (309.55,113.9) .. controls (313.17,113.9) and (316.1,116.83) .. (316.1,120.45) .. controls (316.1,124.07) and (313.17,127) .. (309.55,127) .. controls (305.93,127) and (303,124.07) .. (303,120.45) -- cycle ;
\draw    (130,121) .. controls (152.1,70.75) and (271.1,60.75) .. (309.55,120.45) ;
\draw    (130,121) .. controls (163.1,81.75) and (232.1,72.75) .. (279.55,121.55) ;
\draw    (130,121) .. controls (170,91) and (213.1,87.75) .. (250.55,121.55) ;
\draw    (189.55,120.55) .. controls (228.1,152.75) and (299.45,151) .. (339.45,121) ;
\draw    (159.55,120.55) .. controls (199.1,167.75) and (309.1,164.75) .. (339.45,121) ;
\end{tikzpicture}
    \caption{Example for Definition~\ref{def:Gnrsp}. This graph is $G_{8,3,2,0}$: it has 8 nodes, edges $v_1v_{n-i}$ for $i\in \{1,2,3\}$, edges $v_nv_{1+j}$ for $j \in \{1,2\}$, and it does not have the edge $v_1v_n$.}
    \label{fig:Gnrsp}
\end{figure}
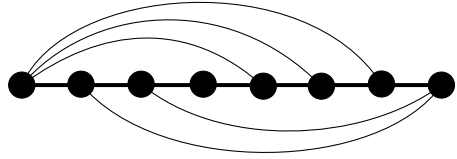

\begin{theorem}[Theorem 2.12 in~\cite{EllinghamMOT16} (adapted)]\label{thm:K24-3-connected}
There exists an $f$ such that the set of 3-connected $K_{2,4}$-minor-free graphs is $\mathcal{G}[f]$, plus nine graphs on at most 8 vertices.
\end{theorem}

In \cite{EllinghamMOT16}, the authors give an explicit description of $f$  but we can avoid going into details here because of the following general lemma. 

\begin{lemma}\label{lem:Gf}
For all $f$, $\mathcal{G}[f]$ can be certified with $O(\log n)$-bit labels.
\end{lemma}

\begin{proof}
On a \emph{yes}-instance, the prover certifies the spanning paths with root $v_1$, with last node $v_n$, and writes in each certificate the values $n,r,s$ and $p$.
The nodes check the structure of the path and the fact that $n,r,s$ and $p$ are the same on all nodes. Second $v_1$ checks the structure of its neighborhood, and in particular the values $r$ and $p$. Similarly, $v_n$ checks  the structure of its neighborhood, and in particular the values $s$, and the fact that its distance to the root is indeed $n$. Finally, all nodes check that $f(n,r,s,p)=1$.
The correctness of the scheme is straightforward.
\end{proof}

This directly yields the following lemma.
\begin{lemma}\label{lem:3-connected-K24}
3-connected $K_{2,4}$-minor-free graphs can be certified with $O(\log n)$-bit labels.
\end{lemma}

\begin{proof}
Consider a \emph{yes}-instance. By Theorem~\ref{thm:K24-3-connected}, either it is one of the nine small graphs of Theorem~\ref{thm:K24-3-connected}, and then we can use a constant size certification, or it is a graph of $\mathcal{G}[f]$ for the specific $f$ of Theorem~\ref{thm:K24-3-connected}, and then we can use Lemma~\ref{lem:Gf}.
\end{proof}

For the 2-connected case, one of the types of graphs that we want to certify is of the following form: a 3-connected graph, where a set of edges with a special property is expanded with another graph class. 
To be able to certify this, we will need the nodes to check that the set of edges that has been expanded has the special property. 
To capture the notion of special property, without going into the intricate details of what this property is exactly, let us define an \emph{edge-set decider}. 
A function $h$ is an edge-set decider if it takes as input a 3-connected $K_{2,4}$-minor-free graph whose edges are either unlabeled, or labeled with a special label, and outputs a Boolean.

\begin{lemma}\label{lem:h-label}
Let $h$ be an edge-set decider, such that for every graph $G$, there is at most $O(n)$ different sets of edges $S$ such that $h(G,S)=1$.
The set of 3-connected $K_{2,4}$-free graphs $G$ with labelled edges, where $h(G)$ is true, can be certified with $O(\log n)$ bits. 
\end{lemma}

\begin{proof}
First, for every graph $G$, we fix an indexing of edge sets $S$ such that $h(G,S)=1$.
The prover first uses the same certificates as in Lemma~\ref{lem:3-connected-K24} for the certification of unlabeled 3-connected $K_{2,4}$-minor-free graphs. 
Then it gives to all nodes the index of the set of labeled edges.
Following the certification of the proof of Theorem~\ref{thm:K24-3-connected}, every node knows in which graph it lives and what is its position in that graph. Then, every node just checks that the labeled edges in its neighborhood correspond to the index announced by the prover. The labels have size $O(\log n)$ because of Lemma~\ref{lem:3-connected-K24} and because there are at most $O(n)$ different sets of edges $S$ such that $h(G,S)=1$.
\end{proof}

\subsection{2-connected case}

We now state the characterization theorem of the 2-connected case. 

\begin{theorem}[Theorem~3.5 in \cite{EllinghamMOT16}]
There exists a function $h$ as in Lemma~\ref{lem:h-label} such that the following holds.
A graph $G$ is 2-connected $K_{2,4}$-minor-free graph if and only if one of the following holds:
\begin{enumerate}
    \item $G$ is outerplanar.
    \item $G$ is the union of three path-outerplanar graphs $H_1$ , $H_2$, $H_3$ with the same path endpoints $x$ and $y$, and possibly the edge $(x,y)$, where $|V (H_i)| \geq 3$, for each $i$ and $V (H_i) \cap V(H_j) = {x, y}$ for $i \neq j$.
    \item  $G$ is obtained from a 3-connected $K_{2,4}$-minor-free graph $G_0$ by choosing a subset $S$ such that $h(G_0,S)=1$, and replacing each edges $(x_i,y_i)$ of $S$ by a path-outerplanar graphs $H_i$ with endpoints $(x_i,y_i)$ , where $V(H_i)\cap V (G_0) = {x_i , y_i }$ for each $i$, and $V(H_i)\cap V(H_j) \subset V(G_0)$ for $i \neq j$. 
\end{enumerate}
\end{theorem}

In \cite{EllinghamMOT16}, $h$ is called the set of subdividable edges, and is fully characterized. Our proof works for any $h$, as long as it satisfies the properties of Lemma~\ref{lem:3-connected-K24}, and it is the case for the $h$ of \cite{EllinghamMOT16}.

\begin{lemma}\label{lem:2-connected-K24}
2-connected $K_{2,4}$-minor-free graphs can be certified with $O(\log n)$-bit labels.
\end{lemma}

\begin{proof}
We show that each of the three cases can be certified with $O(\log n)$ bits.
\begin{enumerate}
\item Outerplanar graphs can be certified with $O(\log n)$ bits (Corollary~\ref{coro:diamond-etc}).
\item This case basically consists in an edge expansion of a multigraph with three edges between two nodes by path outerplanar graphs. Note that the proof of Proposition~\ref{prop:edge-expansion} works here even if the original graph is a multigraph. Proposition~\ref{prop:edge-expansion} gives us an $O(\log n)$ edge certification because path-outerplanar graphs can be certified with $O(\log n)$ certificates, and the condition on the number of nodes can also be certified with $O(\log n)$ bits with a spanning tree counting the number of nodes (see \emph{e.g.} in \cite{Feuilloley19}). 
This edge certification can be transferred to a node certification with the same certificate size asymptotically because of Theorem~\ref{thm:degeneracy}, and because $H$-minor-free graphs have bounded degeneracy.
\item Again, this item basically corresponds to an edge-expansion: the edge expansion of a 3-connected $K_{2,4}$-minor-free graph by path-outerplanar graphs. We know by Lemma~\ref{lem:3-connected-K24} and Lemma~\ref{lem:path-outerplanar} that both these classes can be certified on $O(\log n)$ bits, so the vanilla edge-expansion can also be certified with $O(\log n)$ bits (using the degeneracy like in the previous item). 
The only issue left is the fact that the only edges of $G_0$ that are allowed to be expanded by something different from an edge need to belong to an $S$ such that $h(G_0,S)=1$. But this is easy with Lemma~\ref{lem:h-label}: the edges that have a path-outerplanar expansion are the one that are considered to have a special label.\qedhere
\end{enumerate}
\end{proof}

\section{Certifying $H$-free graphs with $|H| \le 4$}
\label{sec:4vertices}

In previous sections, we have proven that certifying $H$-minor free graphs can be done with $O(\log n)$ bits for some graphs $H$. 
The graphs we have treated in previous sections are somehow amongst the hardest graphs of small size. When the connectivity of the graph $H$ increases, the class of $H$-minor free graph contains more and more graphs, and then is (morally speaking) harder to certify.
Let us prove that the other graphs on $4$ vertices (which have fewer edges, and then are less connected) can also be certified, with arguments either simpler than or similar to what has been done in previous sections, to establish the following theorem.

\ThmFourVertices*

We consider two cases depending on whether $H$ contains a cycle. 

\begin{lemma}
If $|H| \le 4$, and $H$ contains a cycle, then $H$-free graphs can be certified with $O(\log n)$ bits.
\end{lemma}

\begin{proof}
Since $H$ contains a cycle, either it is $C_4$, and the result follows from Corollary~\ref{coro:diamond-etc}, or it contains a triangle. 
Let us distinguish the cases depending on how the fourth vertex is connected to the triangle. If it is connected to two or three vertices, then $H$ is either $K_4$, or a diamond, and then $H$-minor-free graphs can be certified with $O(\log n)$ bits by Corollary~\ref{coro:diamond-etc}.

So we can assume that $H$ is a triangle plus one vertex attached to at most one vertex of the triangle. If $G$ contains a cycle, let $C$ be a shortest cycle in $G$, that is a cycle that contains the minimum number of vertices. Then $C$ must contain all the vertices of the graph. Indeed, otherwise, since $G$ is connected, there exists a node $v$ attached to $C$, and $v \cup C$ contains $H$ as a minor. Therefore, $G$ is either a cycle or a tree, which can be both certified in $O(\log n)$ bits, see Subsection~\ref{subsec:building-blocks}.
\end{proof}

\begin{lemma}
If $|H|\le 4$ and $H$ is acyclic, then we can certify $H$-free graphs with $O(\log n)$ bits.
\end{lemma}
\begin{proof}
If $H$ has an isolated node then any graph $G$ contain $H$ has a minor as long as $G$ contains a (non necessarily induced) path on three nodes and an isolated vertex. Since this property holds for every connected graph on $4$ vertices, the conclusion follows.

So we can assume that $H$ is connected. There are only two acyclic connected graphs on $4$ vertices: the star $S_{1,3}$ with $3$ leaves, and the path $P_4$ on four vertices.
If $G$ does not contain a star with $3$ leaves as a minor, it means that $G$ is either a path or a cycle which can be easily certified. 
If $G$ does not contain a path on four nodes as a minor, it means that $G$ is a star which, can be certified the following way. Give the identifier of the center to all nodes, and let the nodes check that they have been given the same ID, and that the non-center nodes have exactly one neighbor, and that this neighbor has this ID.
\end{proof}

This completes the picture for graphs $H$ on at most $4$ vertices.

\section{Graphs on at most $5$ vertices}
\label{sec:5vertices}

Let us now focus on graphs $H$ with at most $5$ vertices. We were not able to deal with all of them, the most problematic one being $K_5$, as we will discuss later on. However, we proved that $H$-minor freeness can be certified for some dense graphs like $K_{2,3}$. The goal of this section is to provide evidence that again, the hardest case will be the case where $H$ is dense. Before entering into the details of the proof, let us study some necessary conditions on the graph to be minimally not $H$-minor-free.

\subsection{$H$-minimal graphs}
A graph $G$ is \emph{$H$-minimal} if $G$ admits a $H$-minor but, for any vertex $v$, $G\setminus v$ does not admit any $H$-minor. 
Consider such a model $V_1,\ldots,V_{|H|}$ of $H$ in a $H$-minimal graph $G$. 
Intuitively, for all $i$, the important part of the subgraph induced by $V_i$ is a spanning tree that makes it connected, and connected to the neighboring $V_j$'s. 
For example, if a $V_i$ contains a node that is only connected to other nodes of $V_i$, and whose removal does not disconnect the subgraph of $V_i$, then this node is unessential. 
In other words, such a node would not appear in a $H$-minimal graph, because we could remove it, and still have a model of $H$. Nevertheless, it is not true that the subgraph induced by every $V_i$ is a tree (see Figure~\ref{fig:almost-tree}).

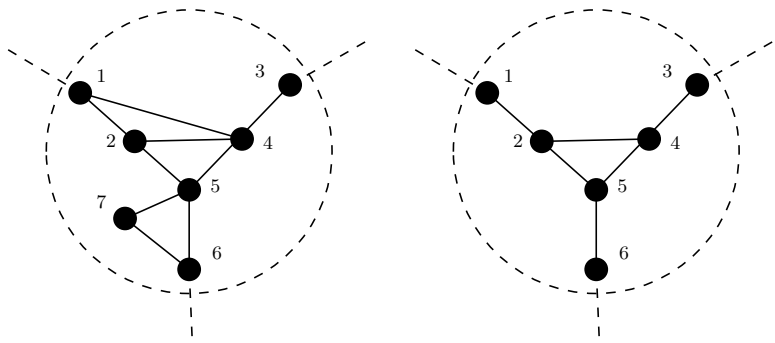
\begin{figure}[!h]
    \centering
    \begin{tabular}{cc}
    \scalebox{0.8}{
    \tikzset{every picture/.style={line width=0.75pt}} 

\begin{tikzpicture}[x=0.75pt,y=0.75pt,yscale=-1,xscale=1]

\draw  [dash pattern={on 4.5pt off 4.5pt}] (75.09,176.18) .. controls (75.09,126.94) and (115,87.03) .. (164.24,87.03) .. controls (213.47,87.03) and (253.39,126.94) .. (253.39,176.18) .. controls (253.39,225.41) and (213.47,265.33) .. (164.24,265.33) .. controls (115,265.33) and (75.09,225.41) .. (75.09,176.18) -- cycle ;
\draw  [fill={rgb, 255:red, 0; green, 0; blue, 0 }  ,fill opacity=1 ] (157.32,200.2) .. controls (157.31,196.38) and (160.39,193.27) .. (164.21,193.26) .. controls (168.03,193.25) and (171.14,196.33) .. (171.15,200.15) .. controls (171.16,203.97) and (168.08,207.08) .. (164.26,207.09) .. controls (160.44,207.1) and (157.34,204.02) .. (157.32,200.2) -- cycle ;
\draw  [fill={rgb, 255:red, 0; green, 0; blue, 0 }  ,fill opacity=1 ] (220.32,134.2) .. controls (220.31,130.38) and (223.39,127.27) .. (227.21,127.26) .. controls (231.03,127.25) and (234.14,130.33) .. (234.15,134.15) .. controls (234.16,137.97) and (231.08,141.08) .. (227.26,141.09) .. controls (223.44,141.1) and (220.34,138.02) .. (220.32,134.2) -- cycle ;
\draw  [fill={rgb, 255:red, 0; green, 0; blue, 0 }  ,fill opacity=1 ] (157.32,250.2) .. controls (157.31,246.38) and (160.39,243.27) .. (164.21,243.26) .. controls (168.03,243.25) and (171.14,246.33) .. (171.15,250.15) .. controls (171.16,253.97) and (168.08,257.08) .. (164.26,257.09) .. controls (160.44,257.1) and (157.34,254.02) .. (157.32,250.2) -- cycle ;
\draw  [fill={rgb, 255:red, 0; green, 0; blue, 0 }  ,fill opacity=1 ] (89.32,139.2) .. controls (89.31,135.38) and (92.39,132.27) .. (96.21,132.26) .. controls (100.03,132.25) and (103.14,135.33) .. (103.15,139.15) .. controls (103.16,142.97) and (100.08,146.08) .. (96.26,146.09) .. controls (92.44,146.1) and (89.34,143.02) .. (89.32,139.2) -- cycle ;
\draw  [fill={rgb, 255:red, 0; green, 0; blue, 0 }  ,fill opacity=1 ] (123.32,169.7) .. controls (123.31,165.88) and (126.39,162.77) .. (130.21,162.76) .. controls (134.03,162.75) and (137.14,165.83) .. (137.15,169.65) .. controls (137.16,173.47) and (134.08,176.58) .. (130.26,176.59) .. controls (126.44,176.6) and (123.34,173.52) .. (123.32,169.7) -- cycle ;
\draw  [fill={rgb, 255:red, 0; green, 0; blue, 0 }  ,fill opacity=1 ] (117.32,218.2) .. controls (117.31,214.38) and (120.39,211.27) .. (124.21,211.26) .. controls (128.03,211.25) and (131.14,214.33) .. (131.15,218.15) .. controls (131.16,221.97) and (128.08,225.08) .. (124.26,225.09) .. controls (120.44,225.1) and (117.34,222.02) .. (117.32,218.2) -- cycle ;
\draw  [fill={rgb, 255:red, 0; green, 0; blue, 0 }  ,fill opacity=1 ] (190.32,168.2) .. controls (190.31,164.38) and (193.39,161.27) .. (197.21,161.26) .. controls (201.03,161.25) and (204.14,164.33) .. (204.15,168.15) .. controls (204.16,171.97) and (201.08,175.08) .. (197.26,175.09) .. controls (193.44,175.1) and (190.34,172.02) .. (190.32,168.2) -- cycle ;
\draw    (96.24,139.18) -- (164.24,200.18) ;
\draw    (227.24,134.18) -- (164.24,200.18) ;
\draw    (164.24,200.18) -- (164.24,250.18) ;
\draw    (124.24,218.18) -- (164.24,200.18) ;
\draw    (124.24,218.18) -- (164.24,250.18) ;
\draw  [dash pattern={on 4.5pt off 4.5pt}]  (51.3,112.1) -- (96.24,139.18) ;
\draw  [dash pattern={on 4.5pt off 4.5pt}]  (227.24,134.18) -- (275.3,106.43) ;
\draw  [dash pattern={on 4.5pt off 4.5pt}]  (164.24,250.18) -- (166.3,298.43) ;
\draw    (130.24,169.68) -- (197.24,168.18) ;
\draw    (96.24,139.18) -- (197.24,168.18) ;

\draw (105,123.4) node [anchor=north west][inner sep=0.75pt]    {$1$};
\draw (111,164.4) node [anchor=north west][inner sep=0.75pt]    {$2$};
\draw (204,120.4) node [anchor=north west][inner sep=0.75pt]    {$3$};
\draw (209,165.4) node [anchor=north west][inner sep=0.75pt]    {$4$};
\draw (176,192.4) node [anchor=north west][inner sep=0.75pt]    {$5$};
\draw (177,234.4) node [anchor=north west][inner sep=0.75pt]    {$6$};
\draw (105,202.4) node [anchor=north west][inner sep=0.75pt]    {$7$};

\end{tikzpicture}}&
    \scalebox{0.8}{
    \tikzset{every picture/.style={line width=0.75pt}} 

\begin{tikzpicture}[x=0.75pt,y=0.75pt,yscale=-1,xscale=1]

\draw  [dash pattern={on 4.5pt off 4.5pt}] (75.09,176.18) .. controls (75.09,126.94) and (115,87.03) .. (164.24,87.03) .. controls (213.47,87.03) and (253.39,126.94) .. (253.39,176.18) .. controls (253.39,225.41) and (213.47,265.33) .. (164.24,265.33) .. controls (115,265.33) and (75.09,225.41) .. (75.09,176.18) -- cycle ;
\draw  [fill={rgb, 255:red, 0; green, 0; blue, 0 }  ,fill opacity=1 ] (157.32,200.2) .. controls (157.31,196.38) and (160.39,193.27) .. (164.21,193.26) .. controls (168.03,193.25) and (171.14,196.33) .. (171.15,200.15) .. controls (171.16,203.97) and (168.08,207.08) .. (164.26,207.09) .. controls (160.44,207.1) and (157.34,204.02) .. (157.32,200.2) -- cycle ;
\draw  [fill={rgb, 255:red, 0; green, 0; blue, 0 }  ,fill opacity=1 ] (220.32,134.2) .. controls (220.31,130.38) and (223.39,127.27) .. (227.21,127.26) .. controls (231.03,127.25) and (234.14,130.33) .. (234.15,134.15) .. controls (234.16,137.97) and (231.08,141.08) .. (227.26,141.09) .. controls (223.44,141.1) and (220.34,138.02) .. (220.32,134.2) -- cycle ;
\draw  [fill={rgb, 255:red, 0; green, 0; blue, 0 }  ,fill opacity=1 ] (157.32,250.2) .. controls (157.31,246.38) and (160.39,243.27) .. (164.21,243.26) .. controls (168.03,243.25) and (171.14,246.33) .. (171.15,250.15) .. controls (171.16,253.97) and (168.08,257.08) .. (164.26,257.09) .. controls (160.44,257.1) and (157.34,254.02) .. (157.32,250.2) -- cycle ;
\draw  [fill={rgb, 255:red, 0; green, 0; blue, 0 }  ,fill opacity=1 ] (89.32,139.2) .. controls (89.31,135.38) and (92.39,132.27) .. (96.21,132.26) .. controls (100.03,132.25) and (103.14,135.33) .. (103.15,139.15) .. controls (103.16,142.97) and (100.08,146.08) .. (96.26,146.09) .. controls (92.44,146.1) and (89.34,143.02) .. (89.32,139.2) -- cycle ;
\draw  [fill={rgb, 255:red, 0; green, 0; blue, 0 }  ,fill opacity=1 ] (123.32,169.7) .. controls (123.31,165.88) and (126.39,162.77) .. (130.21,162.76) .. controls (134.03,162.75) and (137.14,165.83) .. (137.15,169.65) .. controls (137.16,173.47) and (134.08,176.58) .. (130.26,176.59) .. controls (126.44,176.6) and (123.34,173.52) .. (123.32,169.7) -- cycle ;
\draw  [fill={rgb, 255:red, 0; green, 0; blue, 0 }  ,fill opacity=1 ] (190.32,168.2) .. controls (190.31,164.38) and (193.39,161.27) .. (197.21,161.26) .. controls (201.03,161.25) and (204.14,164.33) .. (204.15,168.15) .. controls (204.16,171.97) and (201.08,175.08) .. (197.26,175.09) .. controls (193.44,175.1) and (190.34,172.02) .. (190.32,168.2) -- cycle ;
\draw    (96.24,139.18) -- (164.24,200.18) ;
\draw    (227.24,134.18) -- (164.24,200.18) ;
\draw    (164.24,200.18) -- (164.24,250.18) ;
\draw  [dash pattern={on 4.5pt off 4.5pt}]  (51.3,112.1) -- (96.24,139.18) ;
\draw  [dash pattern={on 4.5pt off 4.5pt}]  (227.24,134.18) -- (275.3,106.43) ;
\draw  [dash pattern={on 4.5pt off 4.5pt}]  (164.24,250.18) -- (166.3,298.43) ;
\draw    (130.24,169.68) -- (197.24,168.18) ;

\draw (105,123.4) node [anchor=north west][inner sep=0.75pt]    {$1$};
\draw (111,164.4) node [anchor=north west][inner sep=0.75pt]    {$2$};
\draw (204,120.4) node [anchor=north west][inner sep=0.75pt]    {$3$};
\draw (209,165.4) node [anchor=north west][inner sep=0.75pt]    {$4$};
\draw (176,192.4) node [anchor=north west][inner sep=0.75pt]    {$5$};
\draw (177,234.4) node [anchor=north west][inner sep=0.75pt]    {$6$};

\end{tikzpicture}}
    \end{tabular}
    \caption{The two pictures represent some set $V_i$ in a $H$-model. The dashed edges represent connections to other nodes of the model. In the first picture, the graph cannot be $H$-minimal, indeed we can remove the nodes~7 and~2, and still have a proper model. In the second picture, no node can be removed without disconnecting the subgraph induced by $V_i$.}
    \label{fig:almost-tree}
\end{figure}

We now describe what the $V_i$'s subgraphs precisely look like in a $H$-minimal graph.
Let $T$ be a graph, and $S_1,\ldots,S_r$ be some prescribed subsets of vertices of $T$.
A \emph{Steiner tree} of~$T$ with respect to the $S_i$'s is a tree in $T$ containing at least one element of each $S_i$.
We say that $T$ is an \emph{almost tree} for the $S_i$'s if any Steiner tree with respect to the $S_i$'s contains all the vertices of $T$.  
Now, given a model $V_1,\ldots,V_{|H|}$ of $H$, and $v_i \in H$, the prescribed sets we are going to consider for $V_i$ are the subsets $S_j\subseteq V_i$ containing all the vertices connected to $V_j$, for every $j$ such that $v_iv_j$ is an edge of $H$. 
A Steiner tree of $V_i$ for the model $V_1,\ldots,V_{|H|}$ of $H$ is a Steiner tree containing at least one vertex of each prescribed set. When the model is clear from context, we simply say a Steiner tree of~$V_i$.

With these notions, let us describe some properties of $H$-minimal graphs:
\begin{lemma}
Let $H$ be a graph and let $G$ a $H$-minimal graph. For every $H$-model of $G$, each $V_i$ is an almost tree.
\end{lemma}
\begin{proof}
The proof is straightforward. If some $V_i$ is not an almost tree, then we can select a subset $V_i'$ of $V_i$ which is an almost tree. When we consider the subsets where all the $V_j$'s are the same but $V_i$ which is replaced $V_i'$, we still have a model of $H$, and it does not contain all the vertices, a contradiction with the fact that $G$ is $H$-minimal.
\end{proof}

It follows that we can characterize the form of the $V_i$'s such that $h_i$ has small degree in $H$.

\begin{corollary}\label{coro:Hminimal_prop}
Let $H$ be a graph, and $G$ be a $H$-minimal graph. There exists a $H$-model of $G$ such that: 
\begin{enumerate}
    \item If the degree of $h_i$ in $H$ is one, then $V_i$ is reduced to a single vertex.
    \item If the degree of $h_i$ in $H$ is two, then $V_i$ is reduced to a single path $P$ and, if $h_j,h_k$ are the two neighbors of $h_i$ then exactly one endpoint of $P$ is connected to  $V_j$, the other is connected to $V_k$, and all the other vertices of $P$ are neither connected to $V_j$ nor $V_k$.
    \item If the degree of $h_i$ in $H$ is three, then the subgraph induced by $V_i$ is of one of the three following types:
    \begin{itemize}
        \item Type A: the subgraph is a triangle with a path attached to each of the three corners (which might be reduced to a single vertex) where the other endpoint of the path is attached to a $V_j$, and no other vertex is attached to $V_j$.
        \item Type B: the subgraph is an induced subdivided star where only the last vertex of each branch is connected to a set $V_j$, and in that case it is connected to exactly one $V_j$.
        \item Type C. the subgraph is a path, and there exists $j,k$ such that the only connections with $V_j$ and $V_{k}$ are on the endpoints of the path. Any connection is possible for the vertices of the path with the last set $V_{\ell}$.
    \end{itemize}
\end{enumerate}
\end{corollary}

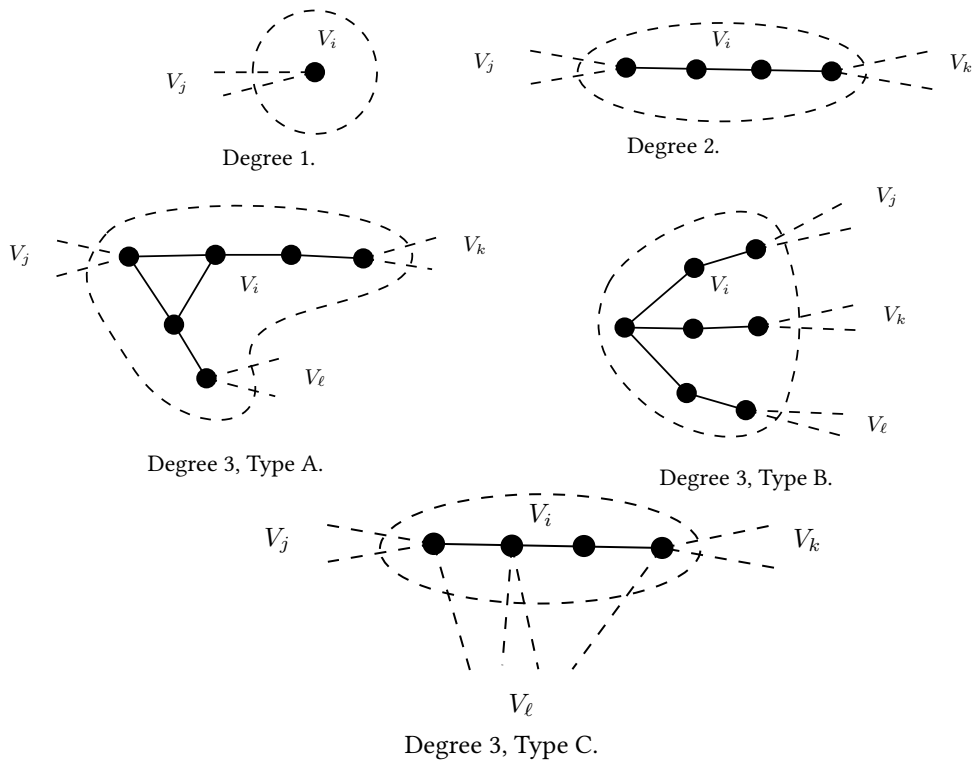
\begin{figure}[!h]
    \centering
    \begin{tabular}{cc}
        \begin{minipage}{0.3\textwidth}
        \centering
        \scalebox{0.9}{
        \tikzset{every picture/.style={line width=0.75pt}} 

\begin{tikzpicture}[x=0.75pt,y=0.75pt,yscale=-1,xscale=1]

\draw  [fill={rgb, 255:red, 0; green, 0; blue, 0 }  ,fill opacity=1 ] (88,62.15) .. controls (88,59.31) and (90.31,57) .. (93.15,57) .. controls (95.99,57) and (98.3,59.31) .. (98.3,62.15) .. controls (98.3,64.99) and (95.99,67.3) .. (93.15,67.3) .. controls (90.31,67.3) and (88,64.99) .. (88,62.15) -- cycle ;
\draw  [dash pattern={on 4.5pt off 4.5pt}]  (37.3,62.2) -- (93.15,62.15) ;
\draw  [dash pattern={on 4.5pt off 4.5pt}] (58.67,62.15) .. controls (58.67,43.11) and (74.11,27.67) .. (93.15,27.67) .. controls (112.19,27.67) and (127.62,43.11) .. (127.62,62.15) .. controls (127.62,81.19) and (112.19,96.62) .. (93.15,96.62) .. controls (74.11,96.62) and (58.67,81.19) .. (58.67,62.15) -- cycle ;
\draw  [dash pattern={on 4.5pt off 4.5pt}]  (93.15,62.15) -- (42.3,75.2) ;

\draw (93,37.4) node [anchor=north west][inner sep=0.75pt]    {$V_{i}$};
\draw (9,61.4) node [anchor=north west][inner sep=0.75pt]    {$V_{j}$};


\end{tikzpicture}}
        \\
        Degree 1.
        \end{minipage}
         & 
         \begin{minipage}{0.4\textwidth}
        \centering
        \scalebox{0.9}{
       \tikzset{every picture/.style={line width=0.75pt}} 

\begin{tikzpicture}[x=0.75pt,y=0.75pt,yscale=-1,xscale=1]

\draw  [fill={rgb, 255:red, 0; green, 0; blue, 0 }  ,fill opacity=1 ] (252,60.15) .. controls (252,57.31) and (254.31,55) .. (257.15,55) .. controls (259.99,55) and (262.3,57.31) .. (262.3,60.15) .. controls (262.3,62.99) and (259.99,65.3) .. (257.15,65.3) .. controls (254.31,65.3) and (252,62.99) .. (252,60.15) -- cycle ;
\draw  [fill={rgb, 255:red, 0; green, 0; blue, 0 }  ,fill opacity=1 ] (291,61.15) .. controls (291,58.31) and (293.31,56) .. (296.15,56) .. controls (298.99,56) and (301.3,58.31) .. (301.3,61.15) .. controls (301.3,63.99) and (298.99,66.3) .. (296.15,66.3) .. controls (293.31,66.3) and (291,63.99) .. (291,61.15) -- cycle ;
\draw  [fill={rgb, 255:red, 0; green, 0; blue, 0 }  ,fill opacity=1 ] (327,61.26) .. controls (327,58.46) and (329.31,56.2) .. (332.15,56.2) .. controls (334.99,56.2) and (337.3,58.46) .. (337.3,61.26) .. controls (337.3,64.05) and (334.99,66.32) .. (332.15,66.32) .. controls (329.31,66.32) and (327,64.05) .. (327,61.26) -- cycle ;

\draw  [fill={rgb, 255:red, 0; green, 0; blue, 0 }  ,fill opacity=1 ] (366,62.24) .. controls (366,59.45) and (368.31,57.18) .. (371.15,57.18) .. controls (373.99,57.18) and (376.3,59.45) .. (376.3,62.24) .. controls (376.3,65.04) and (373.99,67.3) .. (371.15,67.3) .. controls (368.31,67.3) and (366,65.04) .. (366,62.24) -- cycle ;
\draw    (257.15,60.15) -- (371.15,62.24) ;
\draw  [dash pattern={on 4.5pt off 4.5pt}] (230.3,62.7) .. controls (230.3,47.51) and (266.12,35.2) .. (310.3,35.2) .. controls (354.48,35.2) and (390.3,47.51) .. (390.3,62.7) .. controls (390.3,77.89) and (354.48,90.2) .. (310.3,90.2) .. controls (266.12,90.2) and (230.3,77.89) .. (230.3,62.7) -- cycle ;
\draw  [dash pattern={on 4.5pt off 4.5pt}]  (257.15,60.15) -- (201.3,50.2) ;
\draw  [dash pattern={on 4.5pt off 4.5pt}]  (257.15,60.15) -- (204.3,69.2) ;
\draw  [dash pattern={on 4.5pt off 4.5pt}]  (427,72.19) -- (371.15,62.24) ;
\draw  [dash pattern={on 4.5pt off 4.5pt}]  (424.3,51.2) -- (371.15,62.24) ;

\draw (303,38.4) node [anchor=north west][inner sep=0.75pt]    {$V_{i}$};
\draw (171,50.4) node [anchor=north west][inner sep=0.75pt]    {$V_{j}$};
\draw (435,51.4) node [anchor=north west][inner sep=0.75pt]    {$V_{k}$};

\end{tikzpicture}}
       \\
        Degree 2.
        \end{minipage}
    \end{tabular}
    \begin{tabular}{cc}
        \begin{minipage}{0.45\textwidth}
        \centering
        \scalebox{0.9}{
       \tikzset{every picture/.style={line width=0.75pt}} 

\begin{tikzpicture}[x=0.75pt,y=0.75pt,yscale=-1,xscale=1]

\draw  [fill={rgb, 255:red, 0; green, 0; blue, 0 }  ,fill opacity=1 ] (90,184.15) .. controls (90,181.31) and (92.31,179) .. (95.15,179) .. controls (97.99,179) and (100.3,181.31) .. (100.3,184.15) .. controls (100.3,186.99) and (97.99,189.3) .. (95.15,189.3) .. controls (92.31,189.3) and (90,186.99) .. (90,184.15) -- cycle ;
\draw  [fill={rgb, 255:red, 0; green, 0; blue, 0 }  ,fill opacity=1 ] (138,183.15) .. controls (138,180.31) and (140.31,178) .. (143.15,178) .. controls (145.99,178) and (148.3,180.31) .. (148.3,183.15) .. controls (148.3,185.99) and (145.99,188.3) .. (143.15,188.3) .. controls (140.31,188.3) and (138,185.99) .. (138,183.15) -- cycle ;
\draw  [fill={rgb, 255:red, 0; green, 0; blue, 0 }  ,fill opacity=1 ] (115,222.15) .. controls (115,219.31) and (117.31,217) .. (120.15,217) .. controls (122.99,217) and (125.3,219.31) .. (125.3,222.15) .. controls (125.3,224.99) and (122.99,227.3) .. (120.15,227.3) .. controls (117.31,227.3) and (115,224.99) .. (115,222.15) -- cycle ;

\draw  [fill={rgb, 255:red, 0; green, 0; blue, 0 }  ,fill opacity=1 ] (180,183.15) .. controls (180,180.31) and (182.31,178) .. (185.15,178) .. controls (187.99,178) and (190.3,180.31) .. (190.3,183.15) .. controls (190.3,185.99) and (187.99,188.3) .. (185.15,188.3) .. controls (182.31,188.3) and (180,185.99) .. (180,183.15) -- cycle ;
\draw  [fill={rgb, 255:red, 0; green, 0; blue, 0 }  ,fill opacity=1 ] (133,252.15) .. controls (133,249.31) and (135.31,247) .. (138.15,247) .. controls (140.99,247) and (143.3,249.31) .. (143.3,252.15) .. controls (143.3,254.99) and (140.99,257.3) .. (138.15,257.3) .. controls (135.31,257.3) and (133,254.99) .. (133,252.15) -- cycle ;
\draw    (95.15,184.15) -- (143.15,183.15) ;
\draw    (95.15,184.15) -- (120.15,222.15) ;
\draw    (120.15,222.15) -- (143.15,183.15) ;
\draw    (143.15,183.15) -- (185.15,183.15) ;

\draw  [fill={rgb, 255:red, 0; green, 0; blue, 0 }  ,fill opacity=1 ] (220,185.15) .. controls (220,182.31) and (222.31,180) .. (225.15,180) .. controls (227.99,180) and (230.3,182.31) .. (230.3,185.15) .. controls (230.3,187.99) and (227.99,190.3) .. (225.15,190.3) .. controls (222.31,190.3) and (220,187.99) .. (220,185.15) -- cycle ;
\draw    (185.15,183.15) -- (225.15,185.15) ;
\draw    (120.15,222.15) -- (138.15,252.15) ;
\draw  [dash pattern={on 4.5pt off 4.5pt}]  (138.15,252.15) -- (178.3,241.2) ;
\draw  [dash pattern={on 4.5pt off 4.5pt}]  (138.15,252.15) -- (176.3,262.2) ;
\draw  [dash pattern={on 4.5pt off 4.5pt}]  (55,195.1) -- (95.15,184.15) ;
\draw  [dash pattern={on 4.5pt off 4.5pt}]  (55.3,175.2) -- (95.15,184.15) ;

\draw  [dash pattern={on 4.5pt off 4.5pt}]  (225.15,185.15) -- (270.3,172.2) ;
\draw  [dash pattern={on 4.5pt off 4.5pt}]  (225.15,185.15) -- (263.3,191.2) ;
\draw  [dash pattern={on 4.5pt off 4.5pt}] (82.93,170.97) .. controls (104.3,146.2) and (254.22,152.17) .. (252.3,184.2) .. controls (250.37,216.22) and (150.3,206.2) .. (163.3,246.2) .. controls (176.3,286.2) and (121.3,284.2) .. (97.3,246.2) .. controls (73.3,208.2) and (61.55,195.75) .. (82.93,170.97) -- cycle ;

\draw (155,193.4) node [anchor=north west][inner sep=0.75pt]    {$V_{i}$};
\draw (27,176.4) node [anchor=north west][inner sep=0.75pt]    {$V_{j}$};
\draw (279,171.4) node [anchor=north west][inner sep=0.75pt]    {$V_{k}$};
\draw (191,245.4) node [anchor=north west][inner sep=0.75pt]    {$V_{\ell}$};

\end{tikzpicture}}
       \\
       Degree 3, Type A.
    \end{minipage}
    &
    \begin{minipage}{0.45\textwidth}
    \centering
        \scalebox{0.9}{
       \tikzset{every picture/.style={line width=0.75pt}} 

\begin{tikzpicture}[x=0.75pt,y=0.75pt,yscale=-1,xscale=1]

\draw  [fill={rgb, 255:red, 0; green, 0; blue, 0 }  ,fill opacity=1 ] (339.06,222.66) .. controls (336.21,222.61) and (333.95,220.27) .. (334,217.42) .. controls (334.04,214.58) and (336.39,212.31) .. (339.23,212.36) .. controls (342.08,212.41) and (344.34,214.76) .. (344.29,217.6) .. controls (344.24,220.44) and (341.9,222.71) .. (339.06,222.66) -- cycle ;
\draw  [fill={rgb, 255:red, 0; green, 0; blue, 0 }  ,fill opacity=1 ] (373.43,259.26) .. controls (370.58,259.21) and (368.32,256.87) .. (368.37,254.02) .. controls (368.42,251.18) and (370.76,248.91) .. (373.61,248.96) .. controls (376.45,249.01) and (378.72,251.36) .. (378.67,254.2) .. controls (378.62,257.04) and (376.27,259.31) .. (373.43,259.26) -- cycle ;
\draw  [fill={rgb, 255:red, 0; green, 0; blue, 0 }  ,fill opacity=1 ] (377.05,223.32) .. controls (374.21,223.27) and (371.94,220.92) .. (371.99,218.08) .. controls (372.04,215.23) and (374.38,212.97) .. (377.23,213.02) .. controls (380.07,213.07) and (382.34,215.41) .. (382.29,218.26) .. controls (382.24,221.1) and (379.89,223.37) .. (377.05,223.32) -- cycle ;
\draw  [fill={rgb, 255:red, 0; green, 0; blue, 0 }  ,fill opacity=1 ] (377.64,189.32) .. controls (374.79,189.27) and (372.53,186.93) .. (372.58,184.08) .. controls (372.63,181.24) and (374.97,178.97) .. (377.82,179.02) .. controls (380.66,179.07) and (382.92,181.42) .. (382.88,184.26) .. controls (382.83,187.11) and (380.48,189.37) .. (377.64,189.32) -- cycle ;
\draw  [fill={rgb, 255:red, 0; green, 0; blue, 0 }  ,fill opacity=1 ] (406.27,268.83) .. controls (403.42,268.78) and (401.16,266.43) .. (401.21,263.59) .. controls (401.26,260.75) and (403.6,258.48) .. (406.45,258.53) .. controls (409.29,258.58) and (411.56,260.92) .. (411.51,263.77) .. controls (411.46,266.61) and (409.11,268.88) .. (406.27,268.83) -- cycle ;
\draw  [fill={rgb, 255:red, 0; green, 0; blue, 0 }  ,fill opacity=1 ] (413.08,221.94) .. controls (410.24,221.89) and (407.97,219.54) .. (408.02,216.7) .. controls (408.07,213.86) and (410.41,211.59) .. (413.26,211.64) .. controls (416.1,211.69) and (418.37,214.03) .. (418.32,216.88) .. controls (418.27,219.72) and (415.92,221.99) .. (413.08,221.94) -- cycle ;

\draw  [fill={rgb, 255:red, 0; green, 0; blue, 0 }  ,fill opacity=1 ] (411.82,178.91) .. controls (408.98,178.86) and (406.71,176.52) .. (406.76,173.67) .. controls (406.81,170.83) and (409.16,168.56) .. (412,168.61) .. controls (414.84,168.66) and (417.11,171.01) .. (417.06,173.85) .. controls (417.01,176.69) and (414.67,178.96) .. (411.82,178.91) -- cycle ;
\draw    (339.14,217.51) -- (377.73,184.17) ;
\draw    (377.73,184.17) -- (411.91,173.76) ;
\draw    (377.14,218.17) -- (413.17,216.79) ;
\draw    (373.52,254.11) -- (406.36,263.68) ;
\draw    (339.14,217.51) -- (373.52,254.11) ;
\draw    (339.14,217.51) -- (377.14,218.17) ;
\draw  [dash pattern={on 4.5pt off 4.5pt}]  (411.91,173.76) -- (460.17,147.93) ;

\draw  [dash pattern={on 4.5pt off 4.5pt}]  (411.91,173.76) -- (465.17,161.93) ;
\draw  [dash pattern={on 4.5pt off 4.5pt}]  (413.17,216.79) -- (466.42,204.96) ;
\draw  [dash pattern={on 4.5pt off 4.5pt}]  (413.17,216.79) -- (467.17,218.93) ;
\draw  [dash pattern={on 4.5pt off 4.5pt}]  (406.36,263.68) -- (460.36,265.82) ;
\draw  [dash pattern={on 4.5pt off 4.5pt}]  (406.36,263.68) -- (459.17,277.93) ;

\draw  [dash pattern={on 4.5pt off 4.5pt}] (330.17,192.93) .. controls (342.33,167.87) and (415.17,131.93) .. (427.17,167.93) .. controls (439.17,203.93) and (438.17,227.93) .. (429.17,257.93) .. controls (420.17,287.93) and (376.17,281.93) .. (351.17,260.93) .. controls (326.17,239.93) and (318,218) .. (330.17,192.93) -- cycle ;

\draw (384.88,187.66) node [anchor=north west][inner sep=0.75pt]    {$V_{i}$};
\draw (477,137.4) node [anchor=north west][inner sep=0.75pt]    {$V_{j}$};
\draw (481,204.4) node [anchor=north west][inner sep=0.75pt]    {$V_{k}$};
\draw (472,265.4) node [anchor=north west][inner sep=0.75pt]    {$V_{\ell}$};

\end{tikzpicture}}
       \\
        Degree 3, Type B.
    \end{minipage}
    \end{tabular}
    \newline
    \vspace{0.5cm}
    \begin{minipage}{0.45\textwidth}
    \centering
       \tikzset{every picture/.style={line width=0.75pt}} 

\begin{tikzpicture}[x=0.75pt,y=0.75pt,yscale=-1,xscale=1]

\draw  [fill={rgb, 255:red, 0; green, 0; blue, 0 }  ,fill opacity=1 ] (212,374.15) .. controls (212,371.31) and (214.31,369) .. (217.15,369) .. controls (219.99,369) and (222.3,371.31) .. (222.3,374.15) .. controls (222.3,376.99) and (219.99,379.3) .. (217.15,379.3) .. controls (214.31,379.3) and (212,376.99) .. (212,374.15) -- cycle ;

\draw  [fill={rgb, 255:red, 0; green, 0; blue, 0 }  ,fill opacity=1 ] (251,375.15) .. controls (251,372.31) and (253.31,370) .. (256.15,370) .. controls (258.99,370) and (261.3,372.31) .. (261.3,375.15) .. controls (261.3,377.99) and (258.99,380.3) .. (256.15,380.3) .. controls (253.31,380.3) and (251,377.99) .. (251,375.15) -- cycle ;

\draw  [fill={rgb, 255:red, 0; green, 0; blue, 0 }  ,fill opacity=1 ] (287,375.26) .. controls (287,372.46) and (289.31,370.2) .. (292.15,370.2) .. controls (294.99,370.2) and (297.3,372.46) .. (297.3,375.26) .. controls (297.3,378.05) and (294.99,380.32) .. (292.15,380.32) .. controls (289.31,380.32) and (287,378.05) .. (287,375.26) -- cycle ;

\draw  [fill={rgb, 255:red, 0; green, 0; blue, 0 }  ,fill opacity=1 ] (326,376.24) .. controls (326,373.45) and (328.31,371.18) .. (331.15,371.18) .. controls (333.99,371.18) and (336.3,373.45) .. (336.3,376.24) .. controls (336.3,379.04) and (333.99,381.3) .. (331.15,381.3) .. controls (328.31,381.3) and (326,379.04) .. (326,376.24) -- cycle ;
\draw    (217.15,374.15) -- (331.15,376.24) ;

\draw  [dash pattern={on 4.5pt off 4.5pt}] (190.3,376.7) .. controls (190.3,361.51) and (226.12,349.2) .. (270.3,349.2) .. controls (314.48,349.2) and (350.3,361.51) .. (350.3,376.7) .. controls (350.3,391.89) and (314.48,404.2) .. (270.3,404.2) .. controls (226.12,404.2) and (190.3,391.89) .. (190.3,376.7) -- cycle ;

\draw  [dash pattern={on 4.5pt off 4.5pt}]  (217.15,374.15) -- (161.3,364.2) ;
\draw  [dash pattern={on 4.5pt off 4.5pt}]  (217.15,374.15) -- (164.3,383.2) ;
\draw  [dash pattern={on 4.5pt off 4.5pt}]  (387,386.19) -- (331.15,376.24) ;
\draw  [dash pattern={on 4.5pt off 4.5pt}]  (384.3,365.2) -- (331.15,376.24) ;

\draw  [dash pattern={on 4.5pt off 4.5pt}]  (256.15,375.15) -- (269.37,437.23) ;
\draw  [dash pattern={on 4.5pt off 4.5pt}]  (256.15,375.15) -- (251.37,435.23) ;
\draw  [dash pattern={on 4.5pt off 4.5pt}]  (331.15,376.24) -- (286.37,437.23) ;
\draw  [dash pattern={on 4.5pt off 4.5pt}]  (217.15,374.15) -- (235.37,438.23) ;

\draw (263,352.4) node [anchor=north west][inner sep=0.75pt]    {$V_{i}$};
\draw (131,364.4) node [anchor=north west][inner sep=0.75pt]    {$V_{j}$};
\draw (395,365.4) node [anchor=north west][inner sep=0.75pt]    {$V_{k}$};
\draw (253,448.4) node [anchor=north west][inner sep=0.75pt]    {$V_{\ell}$};

\end{tikzpicture}
        Degree 3, Type C.
    \end{minipage}    
    \caption{The types of the $V_i$'s in Corollary~\ref{coro:Hminimal_prop}}
    \label{fig:my_label}
\end{figure}

\begin{proof}
\begin{enumerate}
    \item If the degree of $h_i$ is one, then a Steiner tree only needs the node that is connected to the rest of the model, so $V_i$ has only one node.
    \item If the degree is two, then any Steiner tree contains a path between a node in $V_j$ and a node in $V_k$, and the shortest such paths is an induced path, thus the subgraph induced by $V_i$ must be an induced path, with only the endpoints connected to the rest of the graph.
    \item For degree 3, there are several cases. 
    \begin{itemize}
        \item If $V_i$ contains a cycle, then by minimality the removal of any vertex on this cycle disconnects $V_i$ from another branch. It follows that $V_i$ has type A.
        \item If $V_i$ does not have a cycle, it has at most three leaves. If it has exactly three leaves then it has type B.
        \item Otherwise, $V_i$ is a path, and by the degree-2 case, only the endpoints can connect to some sets $V_j$ and $V_k$, but the connections to the third set $V_{\ell}$ are not controlled. This is type C.
    \end{itemize}
\end{enumerate}
\end{proof}

\subsection{$H$ with an isolated vertex and extension}
\begin{theorem}
\label{thm:H5_isol}
Let $H$ be a graph on $5$ vertices containing an isolated vertex. We can certify $H$-free graphs with certificates of size $O(\log n)$.
\end{theorem}

The rest of this section is devoted to the proof of Theorem~\ref{thm:H5_isol}.

Let $H'$ be the graph $H$ where an isolated vertex has been removed. A $H$-free graph is either a $H'$-free graph, or it is a graph $G$ such that all the models of $H'$ contain all the vertices of $G$. Since $H'$-minor free graphs can be certified within $O(\log n)$ bits by Theorem~\ref{thm:4vertices}, we can assume that $G$ is $H'$-minimal. 

The core of the proof consists in proving the following lemma. 

\begin{lemma}
\label{lem:H'-minimal}
For every graph $H'$ on four vertices, any $H'$-minimal graph is in one of the following categories:
\begin{enumerate}
    \item \label{item:sub-copies} subdivided copies of $H'$,
    \item \label{item:size-4} graphs of size 4,
    \item \label{item:ind-cycles} induced cycles,
    \item \label{item:ind-cycles-plus-1} induced cycle plus a node,
    \item \label{item:figure} the graphs of the type of Figure~\ref{fig:k4+K1},
    \item \label{item:at-most-5} graphs with at most five vertices of degree larger than 2,
    \item \label{item:last_case2} the complete bipartite graph $K_{3,3}$.
\end{enumerate} 
\end{lemma}

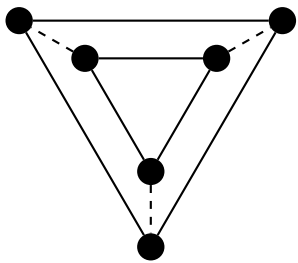
\begin{figure}[!ht]
\centering
\begin{tikzpicture}[thick, every node/.style={circle,draw,fill}]
\node (A) at (30:1) {};
\node (B) at (150:1) {};
\node (C) at (-90:1) {};
\draw (A) -- (B) -- (C) -- (A);
\node (A2) at (30:2) {};
\node (B2) at (150:2) {};
\node (C2) at (-90:2) {};
\draw[dashed] (A) to (A2);
\draw[dashed] (B) to (B2);
\draw[dashed] (C) to (C2);
\draw (A2) to (B2);
\draw (B2) to (C2);
\draw (C2) to (A2);
\end{tikzpicture}
\caption{This drawing represents a class of graphs built by taking two vertex-disjoint triangles, and linking pairs of corners of the triangles by vertex-disjoint paths.}
\label{fig:k4+K1}
\end{figure}


\begin{proof}
Let $v_1,\ldots,v_4$ be the vertices of $H'$ and $V_1,\ldots,V_4$ be a model of $H'$. 
Let us now distinguish the cases depending on the maximum degree in  $H'$. For each case, we characterize the form of $G$.
\medskip

\noindent
\textbf{Case 1: $H'$ is acyclic.}
We claim that if $H'$ is acyclic, then $G$ has to be a copy of $H'$ (which is  Item~\ref{item:sub-copies} in the lemma).
If $H'$ has a node of degree 3, then $G$ should have a node of degree 3, and by minimality $G$ is exactly one node with degree 3 and its three neighbors, that is, exactly the same as $H'$.
If $H'$ has no node of degree 3, then in the model of $H'$ every $V_i$ contains exactly one node by minimality (using Corollary~\ref{coro:Hminimal_prop}), and again $G$ has to be a copy of $H'$.
\medskip

\noindent
\textbf{Case 2: $H'$ has at most one degree-$3$ vertex.}
By Case 1, $H'$ contains a $C_3$ or a $C_4$. 

If $H'=C_4$, it means that $G$ contains a cycle of size at least 4 and that every such cycle contains all the nodes of the $G$. In other words, either $G$ has size exactly four, or $G$ is an induced cycle (Items~\ref{item:size-4} and~\ref{item:ind-cycles} in the lemma).

If $H'$ is a triangle plus a node, then we claim that $G$ is an induced cycle plus a unique vertex (Item~\ref{item:ind-cycles-plus-1} in the lemma). 
Indeed, since $H'$ has at most one degree $3$ node, the vertex not in the triangle has degree at most one. Thus, for any cycle $C$ of $G$, the cycle plus any node incident to $C$ is a $H'$-minor. Since $G$ is $H'$-minimal, the graph $G$ is an induced cycle plus a node.
\medskip

\noindent
\textbf{Case 3: $H'$ has two degree $3$ vertices.}
In this case, $H'$ is a triangle plus a vertex of degree two (that is, a diamond) or three (that is, a $K_4$). 
Let $V_1, V_2,V_3,V_4$ be a model of $H'$ where (at least) $V_3$ and $V_4$ are associated to degree $3$ vertices of $H'$.

Assume that both $V_3$ and $V_4$ have type A or B. In this case, there is a unique vertex $x \in V_3$ incident to a vertex $y$ in $V_4$. If we add $y$ to $V_3$, and remove it from $V_4$, then the size of $V_4$ is decreasing, and the $V_i$'s still form a model of $H'$. We can repeat this operation until $V_4$ does not have type~A or~B. 
Therefore, we can assume that $V_3$ or $V_4$ has type C. 
\smallskip

\noindent
\textit{Case 3.a.}
Assume first that $H'$ is a diamond. 
Then $v_1$ and $v_2$ have degree 2, and by Corollary~\ref{coro:Hminimal_prop}, the subsets $V_1$ and $V_2$ are paths. 
Moreover, if there is an edge between them, and one $V_i$ has two vertices, we could remove one of these vertices, and still have a model of $H'$. 
Therefore, each $V_i$ is reduced to a single vertex (otherwise $G$ is not $H'$-minimal) hence $G$ is $K_4$ (Item~\ref{item:size-4} in the lemma).
Otherwise, if one of $V_3,V_4$ has type~A, then $G\setminus V_1$ contains a diamond as a minor, a contradiction since $G$ is $H'$-minimal. Moreover, as we already observed, at least one of $V_3,V_4$, say $V_3$, has type C. 

Assume that $V_4$ has type B, and let $u\in V_4$ be the vertex of $V_4$ adjacent to $V_3$. 
If $u$ sees two vertices of $V_3$, then replacing $(V_3,V_4)$ by $(V_3\cup\{u\},V_4\setminus u)$ gives a model where $V_3\cup\{u\}$ has type~A, a contradiction. Therefore, $u$ sees a unique vertex of $V_3$ and $G$ is a subdivided diamond (Item~\ref{item:sub-copies} in the lemma).

Otherwise, $V_4$ has type C, hence $V_3$ and $V_4$ induce two paths (with maybe edges between them).

There cannot be two edges between $V_3$ and $V_4$, otherwise, $G\setminus V_1$ contains a diamond-minor, a contradiction. Therefore, there is only one edge between $V_3$ and $V_4$ and $G$ is again a subdivision of a diamond (Item~\ref{item:sub-copies} in the lemma). And this finishes the analysis for the case where $H'$ is a diamond.

\smallskip

\noindent
\textit{Case 3.b.}
Assume that $H'=K_4$. Let us first prove the following claim:

\begin{claim}\label{clm:K4+K1}
If $G$ is $K_4$-minimal, and $G$ contains two vertex-disjoint cycles $C_1,C_2$ with three pairwise non-incident edges between $C_1$ and $G\setminus C_1$, then $G$ is the graph depicted on Figure~\ref{fig:k4+K1}.
\end{claim}

\begin{proof}
Let $a_1b_1,a_2b_2,a_3b_3$ be the edges from the statement, with $a_i\in C_1$.

Assume first that $C_2$ is not a triangle. Then we can remove a vertex of $G \setminus C_1$ in such a way it remains connected and still contains the $b_i$'s. This gives a $K_4$-model, a contradiction. Hence, we assume that $C_2$ is a triangle.

We say that $u\in C_2$ has a private path to one of the $a_i$'s if it has such a path that avoids $C_2\setminus\{u\}$. If some vertex $u\in C_2$ has no such path, then $(G\setminus C_1)\setminus\{u\}$ is connected, hence $G\setminus u$ contains a $K_4$ minor, a contradiction. 
Moreover, if two vertices $u,v\in C_2$ have a private path to the same $a_i$, then we get a $K_4$ minor avoiding some other $a_j$. Therefore, each vertex of $C_2$ is associated with a unique $a_i$ by considering private paths. Observe that there is exactly one path for each of the three choice of endpoints (since if there were two paths, one could remove a vertex which lies in one path but not in the other and get a $K_4$ minor). 

It remains to show that $C_1$ is a triangle. To this end, observe that the structure we found on $G\setminus C_2$ ensures that the hypotheses of the statement are still met when exchanging $C_1$ with $C_2$, and the first argument of the proof shows that $C_1$ is a triangle.
\end{proof}

The remarks at the beginning of Case 3 ensure that all but at most one set $V_i$ (say $V_4$) are of type C. We now do a case analysis of the type of $V_4$.

Assume that $V_4$ has type A. Then $V_4$ contains a triangle $C$. Moreover, $G \setminus C$ contains a cycle since it contains $V_1,V_2,V_3$ which are pairwise connected. So by Claim~\ref{clm:K4+K1}, the graph is of the form of Figure~\ref{fig:k4+K1} (Item~\ref{item:figure} in the lemma).

Assume now that $V_4$ has type B. That is, $V_4$ is a subdivided star with three branches. 
Let $x$ be the vertex of $V_4$ of degree three and $a_1,a_2,a_3$ be the endpoints of the subdivided star rooted in $x$ (note that the $a_i$'s are indeed distinct from $x$). 
Without loss of generality, each $a_i$ is connected to $V_i$ (and not to some other $V_j$ since otherwise $G\setminus a_j$ contains a $K_4$-minor). 
If some $a_i$ is connected to at least two vertices of $V_i$, then $G \setminus (\{ a_i \} \cup V_i)$ contains a cycle as well as $V_i \cup \{a_i \}$ with the conditions of Claim~\ref{clm:K4+K1}. So the graph is the graph of Figure~\ref{fig:k4+K1} (Item~\ref{item:figure} in the lemma).
Therefore, each $a_i$ is connected to exactly one vertex of $V_i$ and $G$ is a subdivided $K_4$ (Item~\ref{item:sub-copies} in the lemma). 

Assume now that $V_4$ has type C. That is, the four sets are of type C. 
We focus on $V_1$. 
We extend $V_1$ greedily, that is, if we can add a vertex $v$ of some $V_j$ to $V_1$, in such a way can still find a model of~$K_4$ where one of the set is $V_1 \cup v$, we do it. 
So from now on, we can assume that any addition of a neighbor of a vertex of $V_j$ to $V_1$ does not keep a model. 
We can assume that $V_1$ has still type C, otherwise we can apply the previous cases. 
Now we claim that $V_2 \cup V_3 \cup V_4$ is a cycle $C$. Indeed, it must contain a cycle, since $V_2,V_3,V_4$ is a model of the triangle. If $w$ is not in this cycle, either it is adjacent to $V_1$ and then can be added to $V_1$ (a contradiction) or it is not, and then $G \setminus w$ has a model of $K_4$ and then $H=K_4+K_1$ appears as a minor in $G$.
Let $C$ be the cycle containing all the vertices of $V_2,V_3,V_4$ and $X_1,X_2,X_3$ the neighbors of $V_1$ in respectively $V_2,V_3,V_4$.

Assume that the cycle $C$ is not induced.
Any chord of the cycle separate the cycle into two sides. 
If there is a chord of $C$ that leaves one side of the cycle with at least one element of each of $X_1,X_2,X_3$, then we can remove a vertex on the opposite side of the cycle and still have a $K_4$-minor, a contradiction with the $H$-minimality of $G$. 
So, without loss of generality, the chord separates $X_1$ on one side and at least one element from $X_2,X_3$ on the other side. Let $e$ be the chord and $P,P'$ be the two parts of $C$ separated by $e$ where $P'$ only contains $X_1$.
In this case, we can apply Claim~\ref{clm:K4+K1} with the cycle $e+P'$, and a cycle  using $V_1$, a part of $P$ and edges between $V_1$ and $X_2,X_3$. Hence, this case again boils down to the graphs of Figure~\ref{fig:k4+K1} (Item~\ref{item:figure} in Lemma~\ref{lem:H'-minimal}). 

So we can assume that $C$ is an induced cycle. If $V_1$ is reduced to a single vertex, then the graph $G$ is a wheel (an induced cycle plus a vertex incident to at least $3$ vertices of the cycle) and it is $K_4$-minimal (Item~\ref{item:ind-cycles-plus-1} in Lemma~\ref{lem:H'-minimal}). 
So we can assume that $V_1$ has at least two vertices. And it is a path since it has type~C and both endpoints of the path have neighbors in $C$ (otherwise the model is not minimal).  
Since we have a $K_4$-model, we need the whole set $V_1$ to have at least three different neighbors on $C$. 

First, note that every vertex of $V_1$ has at most $2$ neighbors on $C$ (otherwise, we have a $K_4+K_1$ model since $V_1$ contains at least two vertices). More generally, if a subpath of $V_1$ has at least three neighbors on $C$, we have a contradiction. So we can assume that, both extremities of $V_1$ have a private neighbor in $C$ and, in total $V_1$ is adjacent to at most $4$ vertices in $C$.

Now if $N(V_1) \cap C$ has size $4$, we claim that only extremities of $V_1$ have neighbors in $C$ and each extremity has exactly two neighbors. Let us denote by $v_1,\ldots,v_\ell$ the vertices of $V_1$ with neighbors in $C$ (in that order in the path $V_1$). Since $|N(\cup_{i \le \ell-1} v_i) \cap C| \le 2$ and $|N(\cup_{i=2}^{\ell} v_i) \cap C| \le 2$, we have $\ell=2$.
So $v_1$ and $v_2$ are the two extremities of $V_1$ and $N(v_1) \cap C = \{a,b\}$ and $N(v_2) \cap C= \{c,d \}$ where $a,b,c,d$ are pairwise distinct. If $a,b,c,d$ appear in that order then we are in the case of Item~\ref{item:at-most-5} in Lemma~\ref{lem:H'-minimal}. So we can assume up to symmetry that the vertices $a,c,b,d$ appear in that order in $C$. If the cycle is has length $4$ and $v_1v_2$ is an edge, then the graph is $K_{3,3}$ (Item~\ref{item:last_case2}). Now if at least one of $ac,cb,bd,da,v_1v_2$ is not an edge, we have a $K_4+K_1$. Assume that one of $ac,cb,bd,da$ is subdivided, w.l.o.g. $ac$. Let $x$ be a vertex between $a$ and $c$ in $C$, then $C'=v_1bC_{bd}C_{da}av_1$ is an induced cycle and there are three paths in $G \setminus x$ from $v_2$ to $C'$ (to $v_1,b,d$), a contradiction. If $v_1v_2$ is subdivided then we get the conclusion with the same cycle but three paths are leaving from $c$ (to $a,b$ and $d$)


So, we can now assume that $V_1$ has at most $3$ neighbors in $C$. But now, let $v_1,\ldots,v_\ell$ be the vertices of $V_1$ with neighbors in $C$. Since $v_1$ and $v_\ell$ have private neighbors on $C$, all the other vertices are adjacent to the same vertex $w$ of $C$. Let us denote by $a$ and $b$ respectively the private neighbors of $v_1$ and $v_\ell$. Note that $v_\ell V_1v_1aP_{ab}bv_\ell$ is an induced cycle $C'$ (where $P_{ab}$ is the subpath of $C$ from $a$ to $b$ avoiding $w$). If the only vertex outside of $C'$ is $w$, then we are in Item~\ref{item:ind-cycles-plus-1} in Lemma~\ref{lem:H'-minimal}. If $w$ sees only one $v_i$, then the graph is a subdivided $K_4$, which is Item~\ref{item:sub-copies}. Otherwise, there are at least $4$ paths starting from $w$ to $C'$ and one of them is subdivided. The removal of a vertex in a subdivided path still leaves a graph with a $K_4$ minor, a contradiction. It completes the proof.
\end{proof}

We can now derive the theorem from the lemma. All the classes of Lemma~\ref{lem:H'-minimal} can be certified with certificates of size $O(\log n)$, and in these classes it is easy to certify $H'$-minimality. This is because in all these classes there is a constant number of special vertices: the nodes before the subdivision for Item~\ref{item:sub-copies}, the additional node for Item~\ref{item:ind-cycles-plus-1}, the corners of the triangles for Item~\ref{item:figure}, the nodes of degree larger than 2 in Item~\ref{item:at-most-5}, and none for Items~\ref{item:size-4} and~\ref{item:ind-cycles}.  
The vertices that are not special have degree 2.
Certifying these classes boils down to having spanning trees pointing to the special vertices, and having a certification of every path of non-special nodes, to transfer the knowledge of the endpoints from one side of the paths to the other. Then basic consistency checks verify the certification. 
Because the structure is so constrained, it is easy also to check whether the graph is $H'$-minimal.

Let us finish this section with an observation. 
Since the graph $G$ is connected, the same proof holds for a graph $H'$ plus a single vertex of degree one as long as all the vertices of $H'$ are equivalent. A graph is \emph{vertex transitive} if for every pair of vertices $(u,v)$, there exists an automorphism of $H$ mapping $u$ to $v$. 

\begin{lemma}
Let $H$ be a graph on 5 vertices obtained by adding a pending edge to a vertex transitive graph. Then $H$-minor free graph can be certified with $O(\log n)$ bits.
\end{lemma}

\section{Lower bounds}
\label{sec:lower-bound}

In this section, we show logarithmic lower bounds for $H$-minor-freeness for every 2-connected graph $H$. 
These results generalize the lower bounds of \cite{FeuilloleyFMRRT20} for $K_k$ and $K_{p,q}$. Our technique is a simple reduction from the certification of paths, via a local simulation. 
In contrast, the proofs of \cite{FeuilloleyFMRRT20} were ad-hoc adaptations of the constructions of \cite{GoosS16} and \cite{FeuilloleyH18}, with explicit counting arguments.
Moreover, our lower bounds apply in the stronger model of locally checkable proofs, where the verifier can look at a constant distance. 

\begin{theorem}
\label{thm:lower-bound}
For every 2-connected graph $H$, certifying $H$-minor-freeness requires $\Omega(\log n)$ bits.
\end{theorem}

Let us start by proving a couple of lemmas. 
Let $H$ be a 2-connected graph, and let $e=uv$ be an arbitrary edge of $H$. 
Let $\Hm$ be the graph $H \setminus e$. Note that $H^{-}$ is connected.
We are going to consider copies of $\Hm$, that we  index as $\Hm_i$'s, and where the copies of the nodes $u$ and $v$ will be called $u_i$ and~$v_i$.
Let $\mathcal{P}$ be the class of all the graphs that can be made by taking some $k$ copies of $\Hm$, and by identifying for every $i\in [1,k-1]$, $v_i$ with $u_{i+1}$. 
In other words, $\mathcal{P}$ is the set of paths, where every edge is a copy of $\Hm$.
The class $\mathcal{C}$ is the same as $\mathcal{P}$ except that we close the paths into cycles, that is, we identify $v_k$ with $u_1$.

\begin{lemma}
\label{lem:P-C}
The graphs of $\mathcal{P}$ are all $H$-minor-free, and the graphs of $\mathcal{C}$ all contain $H$ as a minor.
\end{lemma}

\begin{proof}
Let $G$ be a graph of $\mathcal{P}$. Note that every vertex $v_i$ (identified with $u_{i+1}$) for $i\in \{1,...,k-1\}$, is a cut vertex of $G$. Therefore, since $H$ is 2-connected, a model of $H$ can only appear between two such nodes. 
By construction this cannot happen, as the graphs between the cut vertices are all~$\Hm$.
Thus $G$ is $H$-minor-free.

Now let $G$ be a graph of $\mathcal{C}$. We claim that $G$ contains $H$ as a minor. Consider the following model of $H$. Any $\Hm_i$ is a model of $H$ except for the edge $uv$. Since we have made a cycle of $\Hm_i$'s, there is a path between $v_i$ and $u_i$ outside $\Hm_i$, and this path finishes the model of $H$.
\end{proof}

\begin{lemma}
\label{lem:reduction-path}
Let $H$ be a 2-connected graph.
If there is a certification with $O(f(n))$ bits for $H$-minor-free graphs, then there is a $O(f(n))$ certification for paths.
\end{lemma}

\begin{proof}
Suppose there exists a certification with $O(f(n))$ bits for $H$-minor-free graphs.
The certification of paths boils down to differentiate between paths and cycles, since the nodes can locally check that they have degree 2. Consider the following certification of paths. 
The idea is that the nodes of the path (or cycle) will simulate the computation they would do if instead of being linked by edges, they were linked by copies of $\Hm$. 
The prover will give to every node the certificates of $H$-minor-freeness for these simulated graphs, that is, for every node the certificates of the two copies of $\Hm$ adjacent to it in the simulated graph. Every node will check with its neighbor in the real graph that they have been given the same certificates for these virtual $\Hm$. 
Then every node will run the verification algorithm for $H$-freeness in the simulated graph. 

By construction, the simulated graph is either in $\mathcal{P}$ or in $\mathcal{C}$. Thus, if the verification algorithm accepts, that is, if the simulated graph is $H$-minor-free, then the graph is in $\mathcal{P}$, and then the real graph is a path. 
If the verification algorithm rejects, that is if the simulated graph is not $H$-minor-free then the graph is in~$\mathcal{C}$, and then the real graph is a cycle.
In other words we have designed a local certification for paths, with certificates of size $O(f(n))$.
\end{proof}

\begin{proof}[Proof of Theorem~\ref{thm:lower-bound}]
Now Theorem~\ref{thm:lower-bound} follows from the fact that paths cannot be certified with $o(\log n)$ bits \cite{KormanKP10, GoosS16}.
Note that the proof applies in the locally checkable proof setting, as soon as the number of copies of $\Hm$ is large enough, since the lower bound for paths also applies to locally checkable proofs.
\end{proof}

\section{Discussion}
\label{sec:discussion}

\paragraph*{Milestones to go further}
In this paper, we develop several tools and use them to show that some minor closed graph classes can be certified with $O(\log n)$ bits. One can probably use the tools we developed to certify new classes, we simply wanted to illustrate the interest of these tools. 
Let us now discuss the tools that are missing in order to tackle the general question on $H$-minor-freeness and which steps can be interesting to tackle it.

First, as we explained in Section~\ref{sec:5vertices}, certification of $H$-minor free classes seems easier when $H$ is sparse. One first question that might be interested to look at is the following:

\begin{question}\label{ques:tree}
Let $T$ be a tree. Can $T$-minor free graphs be certified with $O(\log n)$ bits? 
\end{question}

The answer to this question for small graphs $H$ (up to $5$ vertices) is not very interesting since the number of vertices of degree at least $3$ is bounded (and then the whole structure of the graph is "simple"). Even if it remains simple for any $H$, there is no trivial argument allowing us to certify these nodes with $O(\log n)$ bits.
\medskip

A natural approach to tackle Conjecture~\ref{conj:minorconj} would consist in an induction on the size of $H$. Indeed, knowing how to certify $H \setminus x$ for any possible $x$ may help to certify $H$. The basic idea would  consist in separating two cases. 1) When $H$ is not heavily connected where we can heavily use the fact that we can $H \setminus x$ can be certified. And 2) when $H$ is heavily connected, try to use a more general argument. 
A first step toward step 1) would consist in proving that if $H$-minor-freeness can be certified then so is $H+K_1$-minor-freeness\footnote{$H+K_1$ is the graph $H$ plus an isolated vertex.}. We proved it for five vertices in Theorem~\ref{thm:H5_isol}, but the proof heavily uses the structure of the graphs on four vertices. One can then naturally ask the following general question:

\begin{question}\label{quest:H+K1}
Let $H$ be a graph. Can $(H+K_1)$-minor free graphs be certified with $O(\log n)$ bits when $H$ can be certified with $O(\log n)$ bits? 
\end{question}

As in the proof of Theorem~\ref{thm:H5_isol}, we know that we can assume that $G$ is $H$-minimal.
Even if most of the techniques for Lemma~\ref{thm:H5_isol} are specific, Corollary~\ref{coro:Hminimal_prop} gives some (basics) general properties of $H$-minimal graphs which might be useful to tackle this question. 

In structural graph theory, a particular class of $H$-minimal graphs received a considerable attention which are minimally non-planar graphs, in order words, graphs $G$ that are minimal and that contains either a $K_5$ or a $K_{3,3}$ as a minor. It might be interesting to determine if minimally non-planar graphs can be certified with $O(\log n)$ bits.

Note that if we can answer positively Question~\ref{quest:H+K1} positively, the second step would consist in proving the conjecture when we add to $H$ a vertex attached to a single vertex of $H$. Proving this case would, in particular, imply a positive answer to Question~\ref{ques:tree}.

If we want to consider dense graphs, the questions seem to become even harder. In particular, one of the first main complicated $H$-minor class to deal with is probably the class of $K_5$-free graphs. There are several reasons for that. First, it is the smallest $4$-connected graph  and the hardness to certify seem to be highly related to the connectivity of the graph that is forbidden as a minor. The second reason is that it is the smallest graph for which $H$-minor free graphs is a super class of planar graphs.  In other words, we cannot take advantage of the ``planarity'' of the graph (formally or informally) to certify the graph class. 
We then ask the following question:

\begin{question}
Can $K_5$-minor free graphs be certified with $O(\log n)$ bits?
\end{question}

Wagner proved in~\cite{Wagner} that a graph is $K_5$-minor-free if and only if it can be built from planar graphs and from a special graph $V_8$ by repeated clique sums. A \emph{clique sum} consists in taking two graphs of the class and gluing them on a clique and then (potentially) remove edges of that clique. While it should have been easy to certify this sum if we keep the edges of the clique, the fact that they might disappear makes the work much more complicated for certification.

More generally, many decompositions are using the fact that we replace a subgraph by a smaller structure (a single vertex or an edge for instance) only connected to the initial neighbors of that structure in the graph. Certifying such structures is a challenging question whose positive answer can probably permit to break several of the current hardest cases.

\paragraph*{Obstacles towards lower bounds}

There are also several obstacles preventing us to prove extra-logarithmic lower bounds for the certificate size of $H$-minor-free graphs. 
Basically, the only techniques we know consist in (explicit or implicit) reductions to communication complexity. 
In particular \cite{GoosS16} and \cite{Censor-HillelPP20} designed lower bounds for respectively non-3-colorable graphs and bounded diameter graphs as reductions from the disjointness problem in non-deterministic communication complexity. 

Let us remind what these reductions look like.
In such a reduction, one considers a family of graphs with two vertex sets $A$ and $B$, with few edges in between. 
These graphs are defined in such a way that the input of Alice for the disjointness problem can be encoded in the edges of $A$ and the input of Bob in the edges of $B$.
Then, given a certification scheme, Alice and Bob can basically simulate the verification algorithm, and deduce an answer for the disjointness problem. 
If a certification with small labels existed for the property at hand, then the communication protocol would contradict known lower bounds which proves a lower bound for certification.

The difficulty of using this proof for $H$-minor free graphs comes from the fact that it is difficult to control where a minor can appear, that is, to control the models of $H$.
For example, it is difficult to control that if $H$ appears in the graph, then the nodes $V_i$ associated with some node $i$ of $H$ are on Alice's side.
As a comparison, for proving properties on the diameter, \cite{Censor-HillelPP20} used a construction where all the longest paths in the graph had to start from Alice side and finish in Bob side, but such a property seems difficult to obtain for minors. 

\paragraph*{Connectivity questions}

A large part of the paper is devoted to certify connectivity and related notions that are of independent importance, for instance to certify the robustness of a network.
For these, we do not have lower bounds, and leave the following question open.

\begin{question}
Does the certification of $k$-connectivity require $\Omega(\log n)$ bits?
\end{question}

For this question it is tempting to try a construction close to the one we have used for $H$-minor-free graphs. For example, one could think that the nodes of the path/cycle could simulate the $k$-th power of the graph which is $k$-connected if and only if the graph is a cycle. 
But this does not work: we want the \emph{yes}-instances for the property (\emph{e.g.} the $k$-connected graphs) to be in mapped to \emph{yes}-instances for acyclicity (\emph{e.g.} paths), and not with the \emph{no}-instances, which are the cycles. 

An interesting open problem about $k$-connectivity also is on the positive side:

\begin{question}
Can $k$-connectivity be certified with $O(\log n)$ bits for any $k \ge 4$?
\end{question}

Beyond the question of certifying the connectivity itself, we would like to be able to decompose graphs based on $k$-connected components, like what we did with the block-cut tree for 2-connectivity. 
Such decomposition are more complicated and less studied than block-cut trees, but for 3-connectivity such a tool is  SPQR trees \cite{BattistaT89}. Unfortunately, similarly to the clique sum operation we mentioned earlier, some steps of the SPQR tree construction are based on edges that can be removed in later steps, making it hard to certify this structure.

\section*{Acknowledgments}

We thank Jens M. Schmidt for pointing out a mistake in the characterization of 3-connexity we used in an earlier version of this paper. We also thank the reviewers of an earlier version for their comments.

\bibliography{biblio.bib}

\begin{thebibliography}{38}
\providecommand{\natexlab}[1]{#1}
\providecommand{\url}[1]{\texttt{#1}}
\expandafter\ifx\csname urlstyle\endcsname\relax
  \providecommand{\doi}[1]{doi: #1}\else
  \providecommand{\doi}{doi: \begingroup \urlstyle{rm}\Url}\fi

\bibitem[Afek et~al.(1990)Afek, Kutten, and Yung]{AfekKY90}
Yehuda Afek, Shay Kutten, and Moti Yung.
\newblock Memory-efficient self stabilizing protocols for general networks.
\newblock In \emph{Distributed Algorithms, 4th International Workshop, {WDAG}
  '90}, volume 486, pages 15--28, 1990.
\newblock \doi{10.1007/3-540-54099-7\_2}.

\bibitem[Appel et~al.(1976)Appel, Haken, et~al.]{appel1976every}
Kenneth Appel, Wolfgang Haken, et~al.
\newblock Every planar map is four colorable.
\newblock \emph{Bulletin of the American mathematical Society}, 82\penalty0
  (5):\penalty0 711--712, 1976.

\bibitem[Battista and Tamassia(1989)]{BattistaT89}
Giuseppe~Di Battista and Roberto Tamassia.
\newblock Incremental planarity testing (extended abstract).
\newblock In \emph{30th Annual Symposium on Foundations of Computer Science,},
  pages 436--441, 1989.
\newblock \doi{10.1109/SFCS.1989.63515}.

\bibitem[Brakerski and Patt{-}Shamir(2011)]{BrakerskiP11}
Zvika Brakerski and Boaz Patt{-}Shamir.
\newblock Distributed discovery of large near-cliques.
\newblock \emph{Distributed Comput.}, 24\penalty0 (2):\penalty0 79--89, 2011.
\newblock \doi{10.1007/s00446-011-0132-x}.

\bibitem[Censor{-}Hillel et~al.(2019)Censor{-}Hillel, Fischer, Schwartzman, and
  Vasudev]{Censor-HillelFS19}
Keren Censor{-}Hillel, Eldar Fischer, Gregory Schwartzman, and Yadu Vasudev.
\newblock Fast distributed algorithms for testing graph properties.
\newblock \emph{Distributed Comput.}, 32\penalty0 (1):\penalty0 41--57, 2019.
\newblock \doi{10.1007/s00446-018-0324-8}.

\bibitem[Censor{-}Hillel et~al.(2020{\natexlab{a}})Censor{-}Hillel, Fischer,
  Gonen, Gall, Leitersdorf, and Oshman]{Censor-HillelFG20}
Keren Censor{-}Hillel, Orr Fischer, Tzlil Gonen, Fran{\c{c}}ois~Le Gall, Dean
  Leitersdorf, and Rotem Oshman.
\newblock Fast distributed algorithms for girth, cycles and small subgraphs.
\newblock In \emph{34th International Symposium on Distributed Computing,
  {DISC} 2020}, volume 179 of \emph{LIPIcs}, pages 33:1--33:17,
  2020{\natexlab{a}}.
\newblock \doi{10.4230/LIPIcs.DISC.2020.33}.

\bibitem[Censor{-}Hillel et~al.(2020{\natexlab{b}})Censor{-}Hillel, Paz, and
  Perry]{Censor-HillelPP20}
Keren Censor{-}Hillel, Ami Paz, and Mor Perry.
\newblock Approximate proof-labeling schemes.
\newblock \emph{Theor. Comput. Sci.}, 811:\penalty0 112--124,
  2020{\natexlab{b}}.
\newblock \doi{10.1016/j.tcs.2018.08.020}.

\bibitem[Cheriyan and Maheshwari(1988)]{CheriyanM88}
Joseph Cheriyan and S.~N. Maheshwari.
\newblock Finding nonseparating induced cycles and independent spanning trees
  in 3-connected graphs.
\newblock \emph{J. Algorithms}, 9\penalty0 (4):\penalty0 507--537, 1988.
\newblock \doi{10.1016/0196-6774(88)90015-6}.

\bibitem[Chimani et~al.(2019)Chimani, Juhnke-Kubitzke, Nover, and
  R{\"o}mer]{chimani2019cut}
Markus Chimani, Martina Juhnke-Kubitzke, Alexander Nover, and Tim R{\"o}mer.
\newblock Cut polytopes of minor-free graphs.
\newblock \emph{arXiv preprint arXiv:1903.01817}, 2019.

\bibitem[Chudnovsky et~al.(2006)Chudnovsky, Robertson, Seymour, and
  Thomas]{chudnovsky2006strong}
Maria Chudnovsky, Neil Robertson, Paul Seymour, and Robin Thomas.
\newblock The strong perfect graph theorem.
\newblock \emph{Annals of mathematics}, pages 51--229, 2006.

\bibitem[Dolev(2000)]{Dolev2000}
Shlomi Dolev.
\newblock \emph{Self-Stabilization}.
\newblock {MIT} Press, 2000.
\newblock ISBN 0-262-04178-2.
\newblock URL \url{http://www.cs.bgu.ac.il/\%7Edolev/book/book.html}.

\bibitem[Ellingham et~al.(2016)Ellingham, Marshall, Ozeki, and
  Tsuchiya]{EllinghamMOT16}
Mark~N. Ellingham, Emily~A. Marshall, Kenta Ozeki, and Shoichi Tsuchiya.
\newblock A characterization of {K}\({}_{\mbox{2, 4}}\)-minor-free graphs.
\newblock \emph{{SIAM} J. Discret. Math.}, 30\penalty0 (2):\penalty0 955--975,
  2016.
\newblock \doi{10.1137/140986517}.

\bibitem[Eppstein(1992)]{Eppstein92}
David Eppstein.
\newblock Parallel recognition of series-parallel graphs.
\newblock \emph{Inf. Comput.}, 98\penalty0 (1):\penalty0 41--55, 1992.
\newblock \doi{10.1016/0890-5401(92)90041-D}.

\bibitem[Esperet and L{\'{e}}v{\^{e}}que(2021)]{EsperetL21}
Louis Esperet and Benjamin L{\'{e}}v{\^{e}}que.
\newblock Local certification of graphs on surfaces.
\newblock \emph{CoRR}, abs/2102.04133, 2021.
\newblock URL \url{https://arxiv.org/abs/2102.04133}.

\bibitem[Feuilloley(2019)]{Feuilloley19}
Laurent Feuilloley.
\newblock Introduction to local certification.
\newblock \emph{CoRR}, abs/1910.12747, 2019.

\bibitem[Feuilloley(2020)]{Feuilloley20}
Laurent Feuilloley.
\newblock Bibliography of distributed approximation on structurally sparse
  graph classes.
\newblock \emph{CoRR}, abs/2001.08510, 2020.

\bibitem[Feuilloley and Fraigniaud(2016)]{FeuilloleyF16}
Laurent Feuilloley and Pierre Fraigniaud.
\newblock Survey of distributed decision.
\newblock \emph{Bulletin of the {EATCS}}, 119, 2016.
\newblock {\footnotesize\sf{url:
  \href{http://bulletin.eatcs.org/index.php/beatcs/article/view/411/391}{bulletin.eatcs.org
  link}}}{\footnotesize , \sf{arXiv:
  \href{https://arxiv.org/abs/1606.04434}{1606.04434}}}.

\bibitem[Feuilloley and Hirvonen(2018)]{FeuilloleyH18}
Laurent Feuilloley and Juho Hirvonen.
\newblock Local verification of global proofs.
\newblock In \emph{32nd International Symposium on Distributed Computing,
  {DISC} 2018}, volume 121 of \emph{LIPIcs}, pages 25:1--25:17, 2018.
\newblock \doi{10.4230/LIPIcs.DISC.2018.25}.

\bibitem[Feuilloley et~al.(2020{\natexlab{a}})Feuilloley, Fraigniaud,
  Montealegre, Rapaport, R{\'{e}}mila, and Todinca]{FeuilloleyFMRRT20}
Laurent Feuilloley, Pierre Fraigniaud, Pedro Montealegre, Ivan Rapaport,
  {\'{E}}ric R{\'{e}}mila, and Ioan Todinca.
\newblock Compact distributed certification of planar graphs.
\newblock In \emph{{PODC} '20: {ACM} Symposium on Principles of Distributed
  Computing}, pages 319--328. {ACM}, 2020{\natexlab{a}}.
\newblock \doi{10.1145/3382734.3404505}.

\bibitem[Feuilloley et~al.(2020{\natexlab{b}})Feuilloley, Fraigniaud,
  Montealegre, Rapaport, R{\'{e}}mila, and Todinca]{FeuilloleyFMRRT21}
Laurent Feuilloley, Pierre Fraigniaud, Pedro Montealegre, Ivan Rapaport, Eric
  R{\'{e}}mila, and Ioan Todinca.
\newblock Local certification of graphs with bounded genus.
\newblock \emph{CoRR}, abs/2007.08084, 2020{\natexlab{b}}.

\bibitem[Flocchini and Luccio(2003)]{FlocchiniL03}
Paola Flocchini and Flaminia~L. Luccio.
\newblock Routing in series parallel networks.
\newblock \emph{Theory Comput. Syst.}, 36\penalty0 (2):\penalty0 137--157,
  2003.
\newblock \doi{10.1007/s00224-002-1033-y}.

\bibitem[Fraigniaud and Olivetti(2019)]{FraigniaudO19}
Pierre Fraigniaud and Dennis Olivetti.
\newblock Distributed detection of cycles.
\newblock \emph{{ACM} Trans. Parallel Comput.}, 6\penalty0 (3):\penalty0
  12:1--12:20, 2019.
\newblock \doi{10.1145/3322811}.

\bibitem[Ghaffari and Haeupler(2016)]{GhaffariH16}
Mohsen Ghaffari and Bernhard Haeupler.
\newblock Distributed algorithms for planar networks {II:} low-congestion
  shortcuts, mst, and min-cut.
\newblock In \emph{Proceedings of the Twenty-Seventh Annual {ACM-SIAM}
  Symposium on Discrete Algorithms, {SODA} 2016}, pages 202--219. {SIAM}, 2016.
\newblock \doi{10.1137/1.9781611974331.ch16}.

\bibitem[Ghaffari and Haeupler(2021)]{GhaffariH21}
Mohsen Ghaffari and Bernhard Haeupler.
\newblock Low-congestion shortcuts for graphs excluding dense minors.
\newblock In \emph{Proceedings of the 2021 {ACM} Symposium on Principles of
  Distributed Computing, {PODC} 2021}, page To appear., 2021.

\bibitem[G{\"{o}}{\"{o}}s and Suomela(2016)]{GoosS16}
Mika G{\"{o}}{\"{o}}s and Jukka Suomela.
\newblock Locally checkable proofs in distributed computing.
\newblock \emph{Theory of Computing}, 12\penalty0 (19):\penalty0 1--33, 2016.
\newblock \doi{10.4086/toc.2016.v012a019}.

\bibitem[Haeupler et~al.(2016)Haeupler, Izumi, and Zuzic]{HaeuplerIZ16}
Bernhard Haeupler, Taisuke Izumi, and Goran Zuzic.
\newblock Near-optimal low-congestion shortcuts on bounded parameter graphs.
\newblock In \emph{Distributed Computing - 30th International Symposium, {DISC}
  2016}, volume 9888, pages 158--172. Springer, 2016.
\newblock \doi{10.1007/978-3-662-53426-7\_12}.

\bibitem[Haeupler et~al.(2018)Haeupler, Li, and Zuzic]{HaeuplerLZ18}
Bernhard Haeupler, Jason Li, and Goran Zuzic.
\newblock Minor excluded network families admit fast distributed algorithms.
\newblock In \emph{Proceedings of the 2018 {ACM} Symposium on Principles of
  Distributed Computing, {PODC} 2018}, pages 465--474, 2018.

\bibitem[Korman et~al.(2010)Korman, Kutten, and Peleg]{KormanKP10}
Amos Korman, Shay Kutten, and David Peleg.
\newblock Proof labeling schemes.
\newblock \emph{Distributed Computing}, 22\penalty0 (4):\penalty0 215--233,
  2010.
\newblock \doi{10.1007/s00446-010-0095-3}.

\bibitem[Kostochka(1982)]{kostochka1982minimum}
Alexandr~V Kostochka.
\newblock The minimum hadwiger number for graphs with a given mean degree of
  vertices.
\newblock \emph{Metody Diskret. Analiz.}, \penalty0 (38):\penalty0 37--58,
  1982.

\bibitem[Mondshein(1971)]{Mondshein71}
Lee~F. Mondshein.
\newblock \emph{Combinatorial Ordering and the Geometric Embedding of Graphs}.
\newblock PhD thesis, M.I.T. Lincoln Laboratory / Harvard University, 1971.

\bibitem[Montealegre et~al.(2020)Montealegre, Ram{\'{\i}}rez{-}Romero, and
  Rapaport]{MontealegreRR20}
Pedro Montealegre, Diego Ram{\'{\i}}rez{-}Romero, and Iv{\'{a}}n Rapaport.
\newblock Compact distributed interactive proofs for the recognition of
  cographs and distance-hereditary graphs.
\newblock \emph{CoRR}, abs/2012.03185, 2020.

\bibitem[Naor et~al.(2020)Naor, Parter, and Yogev]{NaorPY20}
Moni Naor, Merav Parter, and Eylon Yogev.
\newblock The power of distributed verifiers in interactive proofs.
\newblock In \emph{Proceedings of the 2020 {ACM-SIAM} Symposium on Discrete
  Algorithms, {SODA} 2020}, pages 1096--115. {SIAM}, 2020.
\newblock \doi{10.1137/1.9781611975994.67}.

\bibitem[Robbins(1939)]{Robbins39}
H.~E. Robbins.
\newblock A theorem on graphs, with an application to a problem of traffic
  control.
\newblock \emph{The American Mathematical Monthly}, 46\penalty0 (5):\penalty0
  281--283, 1939.
\newblock ISSN 00029890, 19300972.

\bibitem[Robertson and Seymour(1985)]{RobertsonS85}
Neil Robertson and Paul~D Seymour.
\newblock Graph minors—a survey.
\newblock \emph{Surveys in combinatorics}, 103:\penalty0 153--171, 1985.

\bibitem[Schmidt(2016)]{Schmidt16}
Jens~M. Schmidt.
\newblock Mondshein sequences (a.k.a. (2, 1)-orders).
\newblock \emph{{SIAM} J. Comput.}, 45\penalty0 (6):\penalty0 1985--2003, 2016.
\newblock \doi{10.1137/15M1030030}.

\bibitem[Thomason(1984)]{thomason84}
Andrew Thomason.
\newblock An extremal function for contractions of graphs.
\newblock In \emph{Mathematical Proceedings of the Cambridge Philosophical
  Society}, volume~95, pages 261--265. Cambridge University Press, 1984.

\bibitem[Wagner(1937)]{Wagner}
K.~Wagner.
\newblock Uber eine eigenschaft der ebenen komplex.
\newblock In \emph{Math. Ann.}, volume 114, pages 570--590, 1937.

\bibitem[Whitney(1932)]{Whitney32}
Hassler Whitney.
\newblock Non-separable and planar graphs.
\newblock \emph{Transactions of the American Mathematical Society},
  34:\penalty0 339–362, 1932.
\newblock \doi{10.1090/S0002-9947-1932-1501641-2}.

\end{thebibliography}

\end{document}